\documentclass[runningheads]{llncs}

\usepackage{url}
\usepackage{color}

\usepackage{graphicx}
\usepackage{tikz}
\usetikzlibrary{arrows,decorations.pathmorphing,backgrounds,fit,positioning,calc,shapes}
\usepackage{pgfmath}
\usepackage{cite}
\usepackage{amssymb,amsfonts}
\usepackage{mathtools}
\usepackage{subfig}
\usepackage{textcomp}
\usepackage{cprotect}
\usepackage{appendix}
\usepackage{paralist}
\usepackage{etoolbox}
\usepackage{microtype}
\usepackage{dsfont}
\usepackage{lscape}
\usepackage{thmtools}
\usepackage{multicol}
\usepackage{adjustbox}
\def\BibTeX{{\rm B\kern-.05em{\sc i\kern-.025em b}\kern-.08em T\kern-.1667em\lower.7ex\hbox{E}\kern-.125emX}}
\usetikzlibrary{arrows, shapes, automata, patterns, matrix, positioning}
\tikzset{
  scaledown/.style = {scale=0.9},
  scaledown2/.style = {scale=0.7, inner sep=0pt, minimum size=1.8em, minimum height=1.8em},
  treenode/.style = {shape=circle, draw, align=center, top color=white, scaledown},
  root/.style = {treenode, font=\Large},
  rtreenode/.style = {shape=rectangle, rounded corners, draw, align=center, top color=white},
  root/.style = {treenode, font=\Large},
  rednode/.style = {treenode, font=\ttfamily\normalsize, bottom color=red!30},
  bluenode/.style = {treenode, font=\ttfamily\normalsize, , bottom color=blue!30},
  greennode/.style = {treenode, font=\ttfamily\normalsize, bottom color=green!30},
  dummy/.style = {circle,draw},
  p1node/.style = {rednode, shape=diamond},
  p2node/.style = {state, scaledown, align=center, inner sep=0pt, minimum size=4em, minimum height=4em},
  randnode/.style = {greennode},
  weight/.style = {scaledown, text=red},
  guard/.style = {scaledown, text=teal},
  invariant/.style = {scaledown, text=violet},
  proba/.style = {scaledown, text=brown},
  clock/.style = {scaledown, text=blue},
  channel/.style = {scaledown, text=olive},
  goal/.style = {state, scaledown2, align=center, shape=rectangle},
  attack/.style = {state, scaledown2, align=center, shape=diamond},
  defense/.style = {state, scaledown2, align=center},
  goalUseCase/.style = {state, align=center, shape=rectangle, scale=0.7},
  attackUseCase/.style = {goalUseCase, color=red, chamfered rectangle, chamfered rectangle xsep=6pt},
  defenseUseCase/.style = {goalUseCase, color=blue, rounded corners, radius=10pt}
}
\usepackage{forest}
\forestset{
  angle below/.style={
    tikz+={%
      \draw ($(!1.child anchor)!.35!(.parent anchor)$) [bend right=15] to ($(.parent anchor)!.65!(!l.child anchor)$);
    },
  },
}
\usepackage{wasysym}
\include{lib/defines}

\usepackage[acronym, section=section, nonumberlist, nomain, nopostdot]{glossaries}
\usepackage{glossaries-extra}

\setabbreviationstyle[acronym]{long-short}
\usepackage[plainpages=false]{hyperref}
\usepackage[all]{hypcap}

\usepackage{doi}
\usepackage{cleveref}           %% Replace Section with a symbol

\ifdefined\LINKSCOLORS
\definecolor{ForestGreen} {RGB}{34,  139,  34}
\definecolor{HeraldRed2}   {rgb}{0.81, 0.12, 0.15}

\newcommand{\refscolor} {blue}
\newcommand{\linkscolor}{HeraldRed2}
\newcommand{\urlscolor} {ForestGreen}

\hypersetup{
	colorlinks  = true,
	breaklinks  = true,
	linkcolor   = \linkscolor,
	urlcolor    = \urlscolor,
	citecolor   = \refscolor,
	anchorcolor = black,
	pdfpagemode=UseOutlines,
}
\fi % LINKSCOLORS

\newacronym[longplural={Debugging Information Entities}]{DIE}{DIE}{Debugging Information Entity}
\newacronym[plural={OSes}, firstplural={operating systems (OSes)}]{OS}{OS}{operating system}

\newacronym{mtd}{MTD}{Moving Target Defense}
\newacronym{dag}{DAG}{Directed Acyclic Graph}
\newacronym{at}{AT}{Attack Tree}
\newacronym{adt}{ADT}{Attack Defense Tree}
\newacronym{amg}{AMG}{Attack Moving target defense DAG}
\newacronym{ptmdp}{PTMDP}{Priced Timed Markov Decision Process}
\newacronym{cdf}{CDF}{Cumulative Distribution Function}
\newacronym{pta}{PTA}{Priced Timed Automata}
\newacronym{ptg}{PTG}{Priced Timed Game}
\newacronym{mp}{MP}{moving parameter}
\newacronym{dlsef}{DLSeF}{Dynamic key-Length-based Security Framework}
\newacronym{dare}{DARE}{Dynamic Application Rotation Environment}
\newacronym{aslr}{ASLR}{Address Space Layout Randomization}
\newacronym{more}{MORE}{Multiple OS Rotational Environment}
\newacronym{iot}{IoT}{Internet of Things}
\newacronym{qos}{QoS}{Quality of Service}

\title{Reasoning about Moving Target Defense in Attack Modeling Formalisms}
\author{Gabriel Ballot%\orcidID{0000-0001-5316-0102}
\and
Vadim Malvone \and
Jean Leneutre \and
Etienne Borde}
\authorrunning{G. Ballot et al.}
\institute{LTCI, Telecom Paris, Institut Polytechnique de Paris, Palaiseau, France\\
\email{\{name.surname\}@telecom-paris.fr}}

\begin{document}
\maketitle              % typeset the header of the contribution
\begin{abstract}
% word limit: 250
Since 2009, \textit{Moving Target Defense} (MTD) has become a new paradigm of defensive mechanism %~\cite{ghosh2009moving}
that frequently changes the state of the target system to confuse the attacker. This frequent change is costly and leads to a trade-off between misleading the attacker and disrupting the quality of service.
Optimizing the MTD activation frequency is necessary to develop this defense mechanism when facing realistic, multi-step attack scenarios. Attack modeling formalisms based on DAG are prominently used to specify these scenarios.

Our contribution is a new DAG-based formalism for MTDs and its translation into a \textit{Price Timed Markov Decision Process} %~\cite{david2014time}
to find the best activation frequencies against the attacker's time/cost-optimal strategies. For the first time, MTD activation frequencies are analyzed in a state-of-the-art DAG-based representation. Moreover, this is the first paper that considers the specificity of MTDs in the automatic analysis of attack modeling formalisms.
Finally, we present some experimental results using \textsc{Uppaal Stratego} to demonstrate its applicability and relevance.
\end{abstract}

\section{Introduction}
\label{sec:introduction}
There is an asymmetry between the attacker and the defender. The defender mostly has static defenses, and the attacker can spend a quasi-unlimited time analyzing the defensive system and finding a vulnerability.
\textit{\gls{mtd}} is a defense paradigm formalized in 2009~\cite{ghosh2009moving} that aims at breaking this asymmetry by frequently changing the defended system state.
An \gls{mtd} is defined with three attributes:
\begin{inparaenum}[(i)]
    \item the \shrinkalt{\textit{\glsfmtlong{mp}}}{\textit{\gls{mp}}}, that is, the system parameter that will be changed, % (\eg the IP address of a server or an encryption protocol),
    \item the set of valid values for the \shrinkalt{\glsfmtlong{mp}}{\gls{mp}} %(\eg IPs of the form 192.122.X.Y or AES with 192-bits keys)
    and a transition function for its next value, and %(\eg increase the IP by 27 each time or choose a key uniformly at random), and 
    \item how frequently the state changes. %(\eg when detecting some activity or after 10MB of data has been transferred).
\end{inparaenum}
Changing a server IP address uniformly at random in IPs of the form 192.122.X.Y every 20 minutes is an example of well-defined \gls{mtd} (\cf IP shuffling~\cite{antonatos2007defending,dunlop2011mt6d,clark2013effectiveness}).
Many scientific publications have addressed \glspl{mtd} since 2009, including the surveys~\cite{navas2020mtd,sengupta2020survey}. However, it is not a mature research field because some challenges like the cost-benefits trade-off remain unsolved.
The choice of the activation frequency for time-based \glspl{mtd} has a great impact on the defense effectiveness and applicability. A higher frequency implies less time for the attacker to exploit vulnerability but also implies cost and may reduce the quality of service. The problems addressed by this paper are
\begin{inparaenum}[(i)]
    \item how to model multi-step attacks on a complex system defended with \glspl{mtd} and
    \item how to find optimal \gls{mtd} activation frequencies in such a model.
\end{inparaenum}

We take inspiration from prominent attack modeling formalisms based on \textit{\gls{dag}}~\cite{kordy2014dag}, such as \textit{\gls{at}}~\cite{mauw2005foundations} or \textit{\gls{adt}}~\cite{kordy2010foundations}. It permits to hierarchically model threats, their causes, and defenses (for \gls{adt}) to represent the possible attack paths and countermeasures. Using this hierarchical representation, we optimize activation frequencies for the \glspl{mtd} using a two-player game on a \textit{Priced Timed Automata}~\cite{david2014time} between the attacker and the defender.
We try to find an optimal balance between expressibility, ease of use, and intuition of the formalism.
We can use state-of-the-art strategic model checkers like \textsc{Uppaal Stratego}~\cite{david2015uppaal} to extract the optimal strategies.

Our contribution is twofold. First, we introduce a \gls{dag}-based graphical model of attack scenarios with \gls{mtd} countermeasures called the \textit{\gls{amg}}. Second, we propose a way to automatically construct a \textit{\gls{ptmdp}}~\cite{david2014time} from our \gls{amg} to compute the attack time and cost distributions under different optimal attacker's strategies. We compute it for different \gls{mtd} activation frequencies to optimize them. To our knowledge, our contribution is the first that proposes a method to analyze the impact of \glspl{mtd} activation frequencies and helps to evaluate an optimal set of activation frequencies for a given system defended with \glspl{mtd}.

The paper is organized as follows. Section~\ref{sec:background} introduces background concepts, \shrinkalt{}{Section~\ref{sec:example} gives a motivating example,} Section~\ref{sec:AMG} presents the \gls{amg} model translated into a \gls{ptmdp} in Section~\ref{sec:construction-PTMDP}. \shrinkalt{}{Section~\ref{sec:using} shows how to use the translated model and} Section~\ref{sec:experiment} solves a concrete use case with \textsc{Uppaal Stratego}. Related works are presented in Section~\ref{sec:RW} and finally, Section~\ref{sec:conclusion} concludes the paper.

\section{Background}
\label{sec:background}

We present, as a background, \shrinkalt{\glspl{at}}{\gls{mtd}, \glspl{at},} and \glspl{ptmdp}. We also define strategies and runs on the \gls{ptmdp}.

\shrinkalt{}{\mysubsection{\glsfmtlong{mtd}}\label{sec:mtd}
In this section, we give a brief overview of \gls{mtd}. It helps to apprehend the characteristics of the defenses considered in this paper, but this section is not necessary to understand the rest of this work.
\subsubsection{Motivation for \glsfmtshortpl{mtd}.}

Traditional defense mechanisms are primarily static (\eg a firewall rejecting a sub-network) or reactive (\eg comparing program signatures with a malware signature database and aborting matching programs).
On the contrary, \glspl{mtd} tend to be proactive: they take action before detecting any suspicious activity. Proactivity has two significant implications: on the one hand, the security against \textit{zero-day attacks} (unknown attacks) or \textit{evasive attacks} (covert attacks) may be considerably strengthened; on the other hand, they can be activated even when there are no attacks, implying unnecessary cost. Thus, a good activation strategy for \glspl{mtd} is essential.

The typical attack structure (\eg categorized in the \textit{cyber-kill-chain}~\cite{hutchins2011intelligence}) always starts with probing the system to understand its configuration and plan vulnerability exploitation. \gls{mtd} aims to change the system state frequently, so the attacker cannot spend unlimited time finding a vulnerability in a single state of the target. As a result, \gls{mtd} invalidates the knowledge acquired by the attacker during its reconnaissance phase.

\subsubsection{Advantages of \glsfmtshortpl{mtd}.}

\glspl{mtd} are meant to increase the complexity, uncertainty, and unpredictability of the system for the attacker. Such defense mechanisms are called \textit{security hardening} and are efficient for slowing down implementations of zero-day attacks, which are a significant challenge as about one zero-day attack is discovered every week (\figurename~\ref{fig:0-day}). Furthermore, \gls{mtd} is not supposed to replace traditional defense methods but is complementary to add an extra layer of confusion for the attacker. Some commercial implementations already exist\cite{polyverse,trapx,morphisec,cybervan}.

\begin{figure}
    \centering
    \includegraphics[width=0.7\columnwidth]{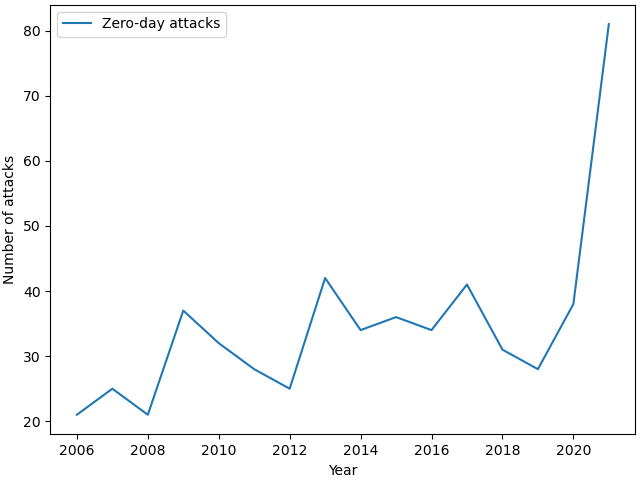}
    \caption[Zero-day vulnerabilities discovered per year.]{Zero-day vulnerabilities discovered per year\cite{Zeroday}.}
    \label{fig:0-day}
\end{figure}

\subsubsection{\glsfmtshort{mtd} Design.}

Three characteristics must be specified to have a well-defined \gls{mtd}:
\begin{inparaenum}[(i)]
    \item the \glsxtrfull{mp}, \ie the configuration of the system that will change,
    \item the set of valid \gls{mp} values, called \textit{configuration set}, and the \textit{movement function}, \ie how to choose the next \gls{mp} value, and
    \item when to change the configuration.
\end{inparaenum}
This common terminology is advocated in several surveys about \gls{mtd}~\cite{sengupta2020survey,navas2020mtd}.

\subsubsection{Choice of the \glsfmtshort{mp}.}

\renewcommand{\arraystretch}{1.3}
\begin{table}[t]
    %\footnotesize
    \centering
    \begin{tabular}{|c|c|c|}
    \hline
      \textbf{Layer} & \textbf{Possible \gls{mp}} & \textbf{Example}\\
      \hline
      \textbf{Data} & format, encoding, representation &
      \glsfmtshort{dlsef}~\cite{puthal2015dynamic}\\
      %\hline
      \textbf{Software} & binary,
      application & \glsfmtshort{dare}~\cite{thompson2016dynamic}\\
      %\hline
      \textbf{Runtime env.} & RAM addresses, instruction
      set & \glsfmtshort{aslr}\\
      %\hline
      \textbf{Platform} & CPU architecture, OS, VM &
      \glsfmtshort{more}~\cite{thompson2014multiple}\\
      %\hline
      \textbf{Network} & protocol, topology, IP address &
      Stream Splitting~\cite{evans2019stream}\\
      \hline
    \end{tabular}
    \caption{System layer classification for the \glsfmtshort{mp}.}
    \label{tab:mp-layers}
\end{table}

\renewcommand{\arraystretch}{1}

The \gls{mp} of the system can generally fit in one of the five system layers, as described in Table~\ref{tab:mp-layers}. The \textit{data layer} includes the representation of the sensitive data. It can be the storage representation (format) or the transmission representation (encoding). As an exmaple, in \gls{dlsef}\cite{puthal2015dynamic}, the protagonists of a communication agree on a local synchronized key generator to change the encryption regularly and synchronously. This reduces the quantity of cyphertext an attacker can intercept for a given key.

The second layer is the \textit{software layer}. It includes changing the software version or altering the program binary code. A textbook example is the \gls{dare}~\cite{thompson2016dynamic}. In this framework, a web interface can have two web servers (Nginx and Apache) that repeatedly take over each other, making the attacker confused about the current webserver.

The third layer is the \textit{runtime environment layer}. Mechanisms in this category typically change the memory management like the well-spread \gls{aslr}. This mechanism is implemented by default on Linux, Windows, and Mac OS. It randomizes the position of the memory sections (like \textit{text, stack, heap}) to prevent the exploit of memory attacks such as buffer overflows.

The fourth layer is the \textit{platform layer}. It includes OS,  VM, or even hardware (\eg CPU) changes. We can cite the \gls{more} as an example~\cite{thompson2014multiple}. In this mechanism, the active OS is taken periodically from a queue of safe OSs. After the rotation, the old OS is scanned to detect any anomaly and deleted or put back on the queue accordingly. Even if the attacker manages to infect an OS, he may not harm the system because the OS is put in quarantine, and a new one is activated.

\def\angle{0}
\def\radius{2}
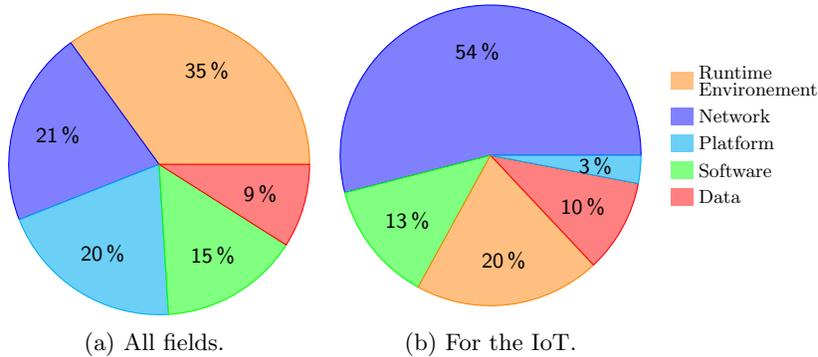
\begin{figure}
    \centering
    \subfloat[All fields.]{
    \begin{tikzpicture}[nodes = {font=\sffamily}]
    \foreach \percent/\name/\color in {
        35/Runtime Environment/orange,
        21/Network/blue,
        20/Platform/cyan,
        15/Software/green,
        9/Data/red,
    }   {
        \ifx\percent\empty\else               % If \percent is empty, do nothing
        % Draw angle and set labels
        \draw[fill={\color!50},draw={\color}] (0,0) -- (\angle:\radius)
          arc (\angle:\angle+\percent*3.6:\radius) -- cycle;
        \node [scaledown] at (\angle+0.5*\percent*3.6:0.7*\radius) {\percent\,\%};
        %\node[pin=\angle+0.5*\percent*3.6:\name]
        %  at (\angle+0.5*\percent*3.6:\radius) {};
        \pgfmathparse{\angle+\percent*3.6}  % Advance angle
        \xdef\angle{\pgfmathresult}         %   and store in \angle
        \fi
    };
    \end{tikzpicture}
    \label{fig:gp-layers}
    }
    \subfloat[For the \glsfmtshort{iot}.]{
    \begin{tikzpicture}[nodes = {font=\sffamily}]
    \foreach \percent/\name/\color in {
        54/Network/blue,
        13/Software/green,
        20/Runtime Environment/orange,
        10/Data/red,
        3/Platform/cyan,
    }   {
        \ifx\percent\empty\else               % If \percent is empty, do nothing
        \draw[fill={\color!50},draw={\color}] (0,0) -- (\angle:\radius)
          arc (\angle:\angle+\percent*3.6:\radius) -- cycle;
        \node [scaledown] at (\angle+0.5*\percent*3.6:0.7*\radius) {\percent\,\%};
        %\node[pin=\angle+0.5*\percent*3.6:\name]
        %  at (\angle+0.5*\percent*3.6:\radius) {};
        \pgfmathparse{\angle+\percent*3.6}  % Advance angle
        \xdef\angle{\pgfmathresult}         %   and store in \angle
        \fi
    };
    \end{tikzpicture}
    \label{fig:iot-layers}
    }
    \subfloat{
    \trimbox{0cm 0cm 0cm 3.5cm}{\begin{tikzpicture}[]
    \matrix [below left, inner sep=2pt, column sep=1pt, row sep=2pt, ampersand replacement=\&] {
        \node [goal, fill=orange!50, minimum size=1em, minimum width=1.3em, draw=none, label={[scale=0.8, align=left]right:{Runtime\\[-0.3em]Environement}}] {};\\
        \node [goal, fill=blue!50, minimum size=1em, minimum width=1.3em, draw=none, label={[scale=0.8]right:{Network}}] {};\\
        \node [goal, fill=cyan!50, minimum size=1em, minimum width=1.3em, draw=none, label={[scale=0.8]right:{Platform}}] {};\\
        \node [goal, fill=green!50, minimum size=1em, minimum width=1.3em, draw=none, label={[scale=0.8]right:{Software}}] {};\\
        \node [goal, fill=red!50, minimum size=1em, minimum width=1.3em, draw=none, label={[scale=0.8]right:{Data}}] {};\\
        };
    \end{tikzpicture}
    }}
    \caption[Taxonomy distribution of \glsfmtshortpl{mtd} according to the system layers.]{Taxonomy distribution of \glspl{mtd} according to the system layers~\cite{navas2020mtd}.}
    \label{fig:pies-mtd}
\end{figure}

Finally, the last layer is the \textit{network layer}. A great variety of proposals in this layer exists for the \gls{iot} (\figurename~\ref{fig:iot-layers}). Possible \glspl{mp} are the network topology, protocols, or IP addresses and ports. An industrial example is Stream Splitting~\cite{evans2019stream}, where the payload of a TCP transmission is split on several flows taking different routes. If a router is infected, the attacker has only access to a part of the sensitive data.

\figurename~\ref{fig:gp-layers} displays the distribution of new \gls{mtd} proposals (from a survey of 89 papers~\cite{navas2020mtd}) among the system layers, and \figurename~\ref{fig:iot-layers} presents the distribution restricted to \gls{iot} proposals. We can notice the prominence of network-level proposals, probably due to the lower implementation cost for embedded systems, while the platform layer only counts one proposal among the 39 surveyed.

\subsubsection{Configuration Set and Movement Function.}

%The "HOW" addresses two questions. Firstly, defining the set of valid states for the \gls{mp}. Secondly, defining the transition function of one state to the next one.
Intuitively, a bigger configuration set induces more security as the attacker cannot brute-force all the states. Nevertheless, even with a vast configuration set, the movement function must ensure the unpredictability of the next \gls{mp} value. Consequently, the next value is usually chosen randomly.
For a random variable $X$ representing the next \gls{mp} value, the unpredictability is often measured with the entropy $H(X)$.
\begin{equation*}
    H(X) = - \sum_{x\in X(\Omega)} \Prob{X=x}\log\left(\Prob{X=x}\right)
\end{equation*}
Higher entropy means more unpredictability and is desired for the \gls{mp}. It can be shown that the entropy is maximal for a uniform distribution. However, a movement function always choosing the next \gls{mp} value according to a uniform distribution in the configuration set is not necessarily optimal. Indeed, two states in a row should expose an \textit{attack surface} (intuitively, the set of channels, methods, and data items that an attacker could use for an intrusion~\cite{manadhata2010attack}) as different as possible. If two states in a row have the same vulnerability, the attacker can exploit the vulnerability during the two states' exposure time.

The configuration set can be obvious (\eg the set of the available IP address for IP shuffling~\cite{antonatos2007defending,dunlop2011mt6d,clark2013effectiveness}). Otherwise, two techniques exist to generate states: \textit{diversification} and \textit{redundancy}. Diversification means using different objects with the same functionality, for example, semantically identical programs with different implementations. Program diversity can be achieved through compiler diversification or inserting \verb+NOP+ (no-operation) randomly within the binary. Redundancy means that several identical components are used to achieve the function. For example, a system could unpredictably use different hardware RAMs to defend against a side-channel memory attack. Similarly, the network traffic could travel through different paths to avoid a centralized attack on a single node as in Stream Splitting~\cite{evans2019stream}.

The most common movement function in the literature is randomization, \ie, choosing the next state independently at each reconfiguration with uniform distribution in the configuration set. This leads to optimal entropy as the unpredictability is maximal.

\subsubsection{When to Change the \glsfmtshort{mp} Value?}

Although it is less studied, the question of ``when'' is essential. The reconfiguration can cost many resources and disrupt the \gls{qos}. However, it should happen often enough, so the attacker's knowledge about the system is invalidated between the probing and the actual attack. We tackle this trade-off in this paper.

We distinguish two methods to determine when to change states. The first one is \textit{time-based}: the system is reconfigured proactively after a time period. %For example, IP address randomization, which is discussed by many proposals~\cite{antonatos2007defending,dunlop2011mt6d,clark2013effectiveness}, could be triggered every 5 minutes.
The second method is \textit{event-based}. In this case, the movement happens either reactively (when detecting a suspicious activity) or on a non-time-based event (\eg the reception of every megabyte of data). This paper focuses on time-based \glspl{mtd}.

\subsubsection{Challenges for \glsfmtshortpl{mtd}.}
We identify three main challenges for \glspl{mtd}:
\begin{inparaenum}[(i)]
    \item the costs,
    \item the state design, and
    \item the activation frequency.
\end{inparaenum}

\paragraph{Costs.}

There are three sources of cost in \glspl{mtd}:
\begin{inparaenum}[(i)]
    \item the state \textit{generation} cost,
    \item the state \textit{storage} cost, and
    \item the state \textit{migration} cost.
\end{inparaenum}
The value of these costs highly depends on the actual \gls{mtd} mechanism. In an \gls{mtd} that shuffles the service/port mapping, the state generation requires running a pseudo-random number generator to obtain the new mapping. The storage cost is an array of integers (generally with 65535 entries), and finally, the migration cost would depend on how each service interruption is handled.

\paragraph{State Design.}

The choice of the \gls{mp} and the configuration set is a significant challenge. First, the configuration set must be large enough so the attacker cannot prepare an attack for each state.% If there are only two configurations, the attacker can prepare two attack scenarios.
Second, two configurations in a row should, ideally, have an empty intersection of their attack surface. In service/port shuffling, if two states in a row have a service that is not remapped, an attack on this service can be carried on.

\paragraph{Activation Frequency.}

This subsection deals with time-based \glspl{mtd}. If the activation frequency is high, the attacker has less time to plan his attack, but the \gls{qos} could be disrupted due to resource consumption and the discontinuity of the service. This trade-off between the cost and the effectiveness of time-based \gls{mtd} has not been addressed enough in the literature, even though it is a crucial point of \gls{mtd} applicability in real life. As a result, This challenge requires an optimal timing strategy for the \gls{mtd} activation frequency that considers the costs, the \gls{qos}, and the threat characteristics (\eg the expected time of the attacks and the success probability).
The following work tackles this problem.}

\mysubsection{\glsfmtlong{at}}\label{subsec:AT}
\shrinkalt{}{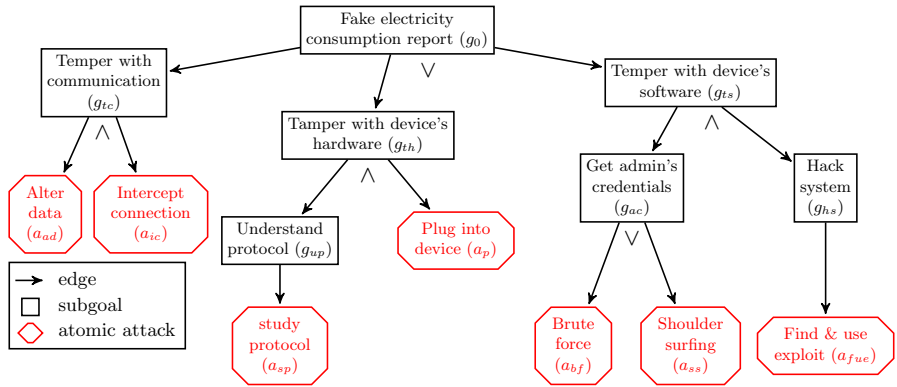
\begin{figure}
    \centering
    \begin{tikzpicture}[
            ->,
            >=stealth',
            shorten >=1pt,
            auto,
            sloped,
            semithick,
            xscale=0.78,
            yscale=0.7
        ]
        
        \node[goalUseCase, label=300:{$\lor$}] at (0, 0) (g_0) {Fake electricity\\consumption report ($\AMGroot$)};
        \node[goalUseCase, label=268:{$\land$}] at (-5, -1) (g_tc) {Temper with\\communication\\($g_{tc}$)};
        \node[goalUseCase, label=below:{$\land$}] at (-0.5, -2) (g_th) {Tamper with device's\\hardware ($g_{th}$)};
        \node[goalUseCase, label=-80:{$\land$}] at (5, -1) (g_ts) {Temper with device's\\software ($g_{ts}$)};
        \node[attackUseCase, label={}] at (-6, -3.5) (a_ad) {Alter\\data\\($a_{ad}$)};
        \node[attackUseCase, label={}] at (-4.2, -3.5) (a_ic) {Intercept\\connection\\($a_{ic}$)};
        \node[attackUseCase, label={}] at (1, -4) (a_p) {Plug into\\device ($a_{p}$)};
        \node[goalUseCase, label=below:{}] at (-2, -4) (g_up) {Understand\\protocol ($g_{up}$)};
        \node[attackUseCase] at (-2,-6) (a_sp) {study\\protocol\\($a_{sp}$)};
        \node[goalUseCase, label=below:{$\lor$}] at (4, -3) (g_ac) {Get admin's\\credentials\\($g_{ac}$)};
        \node[attackUseCase] at (3,-6) (a_bf) {Brute\\force\\($a_{bf}$)};
        \node[attackUseCase] at (5,-6) (a_ss) {Shoulder\\surfing\\($a_{ss}$)};
        \node[goalUseCase, label=below:{}] at (7.3, -3) (g_hs) {Hack\\system\\($g_{hs}$)};
        \node[attackUseCase] at (7.3,-6) (a_fue) {Find \& use\\exploit ($a_{fue}$)};

        \path (g_0) edge (g_tc) edge (g_th) edge (g_ts);
        \path (g_tc) edge (a_ad) edge (a_ic);
        \path (g_th) edge (g_up) edge (a_p);
        \path (g_up) edge (a_sp);
        \path (g_ts) edge (g_ac) edge (g_hs);
        \path (g_ac) edge (a_bf) edge (a_ss);
        \path (g_hs) edge (a_fue);
        
        \matrix [draw, below left, column sep=3pt, row sep=0pt, inner sep=2pt, ampersand replacement=\&, column 1/.style={anchor=base}, column 2/.style={anchor=base west}] at (-3.6, -4.4) {
        \draw[line width=0.5pt, >=stealth', semithick](-0.2, 0) -- (0.2, 0); \& \node[right, scale=0.8, align=left] {edge};\\
        \node [goalUseCase, minimum size=1em, scale=1] {\phantom{a}}; \& \node[scale=0.8]{subgoal};\\
        \node [attackUseCase, minimum size=0.8em, scale=0.8] {\phantom{a}}; \& \node[scale=0.8]{atomic attack}; \\
        };
                
    \end{tikzpicture}
    \caption[Example of an \glsfmtshort{at} for an electricity meter.]{Example of an \gls{at} for an electricity meter. When a subgoal has only one child, the refinement can be omitted.}
    \label{fig:example-at}
\end{figure}}

\glspl{at}~\cite{mauw2005foundations} and their derivatives are graphical security models representing the hierarchical structures of attacks in a tree. The original inspiration comes from Weiss' \textit{threat logic trees} for reliability in 1991~\cite{weiss1991system}.
\shrinkalt{}{Nowadays many derivatives have been proposed to encompass attacks and defenses (Attack Defense Tree~\cite{kordy2014attack,roy2010cyber}), probabilistic computation and time~\cite{arnold2014time}, sequentiality and dynamic aspects~\cite{camtepe2007modeling,lv2011space,ardi2006towards}.  Moreover, \gls{at}-based modeling has been the focus of two surveys in 2014~\cite{kordy2014dag} and 2019~\cite{widel2019beyond}.}

\begin{definition}[\glsxtrlong{at}]
An \gls{at} be is a rooted tree structure $\brac{\Nodes, \Edges, \AMGroot, \operation{}}$ with a finite set of nodes $\Nodes$, edges $\Edges \subseteq \Nodes \times \Nodes$, root $\AMGroot \in \Nodes$ called the \textit{main goal}. Let $\mathsf{I} \subseteq \Nodes$ be the set of inner-nodes (nodes with descendants) called \textit{subgoals}. The list $\operation{} = (\operation{g})_{g \in \mathsf{I}}$ assigns a refinement $\operation{g} \in \{\land, \lor\}$ for each subgoal $g$. The leaves of the tree are called \textit{atomic attacks} or \textit{basic actions}.
\end{definition}

Let $g$ be \shrinkalt{a}{an \gls{at}} subgoal. If $\operation{g}$ is a conjunctive refinement ($\land$), then $g$'s achievement requires all its children to be completed. If $g$ is a disjunctive refinement ($\lor$), then $g$'s achievement requires at least one of its children to be completed.\shrinkalt{Sometimes, a ``sequential and'' refinement is considered~\cite{lv2011space,camtepe2007modeling}, but we will disregard it.}{The atomic attacks can be augmented with \textit{attributes} (\eg probability, cost, or risk), and if we provide a way to propagate these attributes to the parent nodes, we can compute the attributes for the main goal. The overall tree gives an easily understandable description of the possible attack scenarios, and automatic computation of the attributes for the whole tree is possible. \figurename~\ref{fig:example-at} is an example of \gls{at}.

In our work, we introduce \gls{amg} that are an extension of \gls{at} that allows a rooted \gls{dag} structure instead of a tree. Moreover, we will use \textit{atomic attack} instead of \textit{basic actions} to emphasize the attacker part in a setting where the defender will be present too.
}

\mysubsection{\glsfmtlong{ptmdp}}\label{subsec:PTMDP}
\fromptaptg{
In this section we construct the \gls{ptmdp} through three steps: 
\begin{inparaenum}[(i)]
    \item we define the \gls{pta} that is an automata with time and costs,
    \item we extend it as a \gls{ptg} by attributing the actions to two players,
    \item and we add a stochastic environment to get a \gls{ptmdp}.
\end{inparaenum}
}{
In this section we construct the \gls{ptmdp} through two steps: 
\begin{inparaenum}[(i)]
    \item we define the \gls{ptg} that is a two player automata with time and costs,
    \item and we add a stochastic environment to get a \gls{ptmdp}.
\end{inparaenum}
}{}
Our formalism for \fromptaptg{\gls{pta}, \gls{ptg}, and \gls{ptmdp}}{\gls{ptg} and \gls{ptmdp}}{\gls{ptmdp}} is identical to the one defined in~\cite{david2014time}, except that we allow a transition cost in addition to the location cost. These two types of costs are standard. For example, we find them in the \gls{pta} definition in~\cite{behrmann2004priced}.

\shrinkalt{}{\subsubsection{Clocks, Valuations, Constraints.}}

To define a \fromptaptg{\gls{pta}}{\gls{ptg}}{\gls{ptmdp}}, we first need to specify what a clock is and define the clock constraints. 
A \textit{clock} is variable in the non-negative real numbers (denoted $\Delays$) representing the time. Let $\Clocks$ be a set of clocks. It is always implicitly assumed that the clocks of a clock set progress synchronously. We define $\Guard(\Clocks)$ as the set of \textit{clock constraints} generated by the grammar with start symbol and non-terminal $g$ and rule 
$g \rightarrow x \bowtie n \mid x - y \bowtie n \mid g \land g \mid \varepsilon$ where $x,y\in \Clocks$, $\bowtie \in \{ \leq, <, =, >, \geq\}$, $n\in \mathbb N$, and $\varepsilon$ is the empty string.
%We denote $\Delays$ for $\mathbb{R}^+ \cup \{+\infty\}$, the non-negative reals with the positive infinity.
Given a clock set $\Clocks$, a \textit{valuation} $v$ is a function $v : \Clocks \mapsto \Delays$. We call $\Valuations_\Clocks$ the set of valuations on $\Clocks$, or simply $\Valuations$ when the clock set is clear in the context. For $v \in \Valuations$, $v$ is \textit{valid} given a clock constraint $s\in \Guard(\Clocks)$, denoted $v\vDash s$, if the formula $s$ is true when we evaluate the clocks in $s$ with $v$. 
For $\delay \in \Delays$, we denote $v+\delay$ the valuation s.t. $(v+\delay)(x) = v(x) + \delay$ for $x\in \Clocks$, and for a subset $Y\subseteq \Clocks$, we denote $v[Y]$ the valuation where $v[Y](x) = v(x)$ for $x\in \Clocks\setminus Y$ and $v[Y](x) = 0$ otherwise.

\fromptaptg{
    \subsubsection{\glsfmtlong{pta}.}
    \begin{figure}
    \centering
    \begin{tikzpicture}[
            ->,
            >=stealth',
            shorten >=1pt,
            auto,
            node distance=2.8cm,
            semithick,
            sloped
        ]

        \node[state, very thick, label={[invariant]below:{$x\leq 5$}}] (A) {$\loc_0$};
        \node[state] (B) [above left of=A, label={[invariant]above:{$x\leq 6$}}]  {$\loc_1$};
        \node[state] (D) [above right of=A, label={[invariant]right:{$x\leq 3$}}] {$\loc_2$};
        \node[state] (C) [above right of=B, label={[invariant]above:{$x\leq 3$}}] {$\loc_3$};
        \node[state] (E) [right of=D, label={[invariant]above:{$x\leq 2$}}] {$\loc_4$};
        \node[weight] (weightE) [above=0.35 of E]{$c'=1$};
    
        \path (A) edge node [scaledown, above] {$\action_1$} node [guard, below] {$x = 3$}(B);
        \path (A) edge node [scaledown] {$\action_2$} node [clock, below] {$x\leftarrow 0$} (C);
        \path (B) edge [loop left] node [scaledown] {$\action_4$} (B);
        \path (B) edge node [scaledown] {$\action_3$} node [clock, below] {$x\leftarrow 0$} (C);
        \path (C) edge node [scaledown] {$\action_1$} node [weight, below] {$c \mathrel{+}=1$} (D);
        \path (C) edge [bend left]  node [scaledown] {$\action_4$} node [guard, below] {$x < 1$} node [clock, below=0.3] {$x \leftarrow 0$} (E);
        \path (D) edge [loop below] node [scaledown] {$\action_2$} (D);
        \path (D) edge node [scaledown] {$\action_3$} node [weight, below] {$c\mathrel{+}=2$} (A);
        \path (E) edge [bend left]  node [scaledown, below=0.3] {$\action_1$} node [guard, below] {$x \leq 5 \land x \geq 1$} (A);
        
        \matrix [draw, right of=E, column sep=5pt, row sep=3pt, inner sep=2pt, ampersand replacement=\&, column 1/.style={anchor=base}, column 2/.style={anchor=base west}]{
        \node [state, minimum size=2em, scale=0.7] {$\loc$}; \& \node[scale=0.8] {location $\loc$};\\
        \draw[line width=0.5pt, >=stealth', semithick](-0.25, 0) -- (0.25,0); \& \node[right, scale=0.8] {transition};\\
        \node [invariant, scale=0.8] {$x \bowtie y$}; \& \node[right, scale=0.8, align=left] at (0, 0.2em) {location\\[-0.5em]invariant};\\
        \node [guard, scale=0.8] {$x \bowtie y$}; \& \node[scale=0.8, align=left] {edge guard};\\
        \node [weight, scale=0.8] {$c \bowtie y$}; \& \node[right, scale=0.8, align=left] at (0, 0.2em) {location and\\[-0.5em]transition cost};\\
        \node [scale=0.8] {$\action^x_y$}; \& \node[scale=0.8, align=left] {action type};\\
        \node [clock, scale=0.8] {$x \leftarrow 0$}; \& \node[scale=0.8, align=left] {clock reset};\\
        };
    \end{tikzpicture}
    \caption[Example of the structure of a \glsfmtshort{pta}.]{Example of the structure of the \gls{pta} $\PTA = \PTADef$. Where $\Locations = \{\loc_0, \loc_1, \loc_2, \loc_3, \loc_4\}$, $X=\{x\}$, $\Actions= \{\action_1, \action_2, \action_3, \action_4\}$. The variable $c$ represent the cost. We assume that the cost is null where there it is omitted.}
    \label{fig:example-pta}
\end{figure}
    A \gls{pta} is an automata with cost and time attributes.
    \begin{definition}[\glsxtrlong{pta}]
A \gls{pta} is defined as a tuple $\PTA = \PTADef$ where $\Locations$ is the finite set  of \textit{locations}, $\loc_0\in \Locations$ is the \textit{initial location}, $\Clocks$ is a set of synchronous \textit{clocks}, $\Actions$ is the finite set of \textit{actions}, $\Transitions \subseteq \Locations \times \Guard(\Clocks) \times \Actions \times \Locations$ is a transition relation,
%s.t. $(\loc, s, e, \loc')\in \Transitions$ and $(\loc, s', e, \loc'')\in \Transitions$ implies $s=s'$ and $\loc'= \loc''$,
$\price: \Locations\cup\Transitions \mapsto \mathbb N$ assigns \textit{cost rates} to locations and \textit{costs} to edges, $\clockReset: \Transitions\mapsto 2^\Clocks$ gives the set of clocks that are reset after a transition, and $\invariants: \Locations \mapsto \Guard(\Clocks)$ assigns \textit{invariants} to locations.
\end{definition}
    \figurename~\ref{fig:example-pta} is an example of \gls{pta}.
    \subsubsection{\glsfmtlong{ptg}.}
    \begin{figure}
    \centering
    \begin{tikzpicture}[
            ->,
            >=stealth',
            shorten >=1pt,
            auto,
            node distance=2.8cm,
            semithick,
            sloped
        ]

        \node[state, very thick, label={[invariant]below:{$x\leq 5$}}] (A) {$\loc_0$};
        \node[state] (B) [above left of=A, label={[invariant]above:{$x\leq 6$}}]  {$\loc_1$};
        \node[state] (D) [above right of=A, label={[invariant]right:{$x\leq 3$}}] {$\loc_2$};
        \node[state] (C) [above right of=B, label={[invariant]above:{$x\leq 3$}}] {$\loc_3$};
        \node[state] (E) [right of=D, label={[invariant]above:{$x\leq 2$}}] {$\loc_4$};
        \node[weight] (weightE) [above=0.35 of E]{$c'=1$};
    
        \path (A) edge node [scaledown, above] {$\action_1$} node [guard, below] {$x = 3$}(B);
        \path (A) edge node [scaledown] {$\action_2$} node [clock, below] {$x\leftarrow 0$} (C);
        \path (B) edge [loop left, dashed] node [scaledown] {$\action_4$} (B);
        \path (B) edge [dashed] node [scaledown] {$\action_3$} node [clock, below] {$x\leftarrow 0$} (C);
        \path (C) edge node [scaledown] {$\action_1$} node [weight, below] {$c \mathrel{+}=1$} (D);
        \path (C) edge [bend left, dashed]  node [scaledown] {$\action_4$} node [guard, below] {$x < 1$} node [clock, below=0.3] {$x \leftarrow 0$} (E);
        \path (D) edge [loop below] node [scaledown] {$\action_2$} (D);
        \path (D) edge [dashed] node [scaledown] {$\action_3$} node [weight, below] {$c\mathrel{+}=2$} (A);
        \path (E) edge [bend left]  node [scaledown, below=0.3] {$\action_1$} node [guard, below] {$x \leq 5 \land x \geq 1$} (A);

        \matrix [draw, right of=E, column sep=5pt, row sep=3pt, inner sep=2pt, ampersand replacement=\&, column 1/.style={anchor=base}, column 2/.style={anchor=base west}]{
        \node [state, minimum size=2em, scale=0.7] {$\loc$}; \& \node[scale=0.8] {location $\loc$};\\
        \draw[dashed, line width=0.5pt, >=stealth', semithick](-0.25, 0) -- (0.25,0); \& \node[right, scale=0.8, align=left] {uncontrollable\\[-0.5em]transition};\\
        \draw[line width=0.5pt, >=stealth', semithick](-0.25, 0) -- (0.25,0); \& \node[right, scale=0.8, align=left] {controllable\\[-0.5em]transition};\\
        \node [invariant, scale=0.8] {$x \bowtie y$}; \& \node[right, scale=0.8, align=left] at (0, 0.2em) {location\\[-0.5em]invariant};\\
        \node [guard, scale=0.8] {$x \bowtie y$}; \& \node[scale=0.8, align=left] {edge guard};\\
        \node [weight, scale=0.8] {$c \bowtie y$}; \& \node[right, scale=0.8, align=left] at (0, 0.2em) {location and\\[-0.5em]transition cost};\\
        \node [scale=0.8] {$\action^x_y$}; \& \node[scale=0.8, align=left] {action type};\\
        \node [clock, scale=0.8] {$x \leftarrow 0$}; \& \node[scale=0.8, align=left] {clock reset};\\
        };
    \end{tikzpicture}
    \caption[Example of the structure of a \glsfmtshort{ptg}.]{Example of the structure of the \gls{ptg} $\PTG = \PTGDef$ where $\PTA$ is the \gls{pta} from \figurename~\ref{fig:example-pta} and $\ControllableActions= \{\action_1, \action_2\}$, and $\UncontrollableActions=\{\action_3, \action_4\}$. The variable $c$ represent the cost. We assume that the cost is null where there it is omitted.}
    \label{fig:example-ptg}
\end{figure}
    A \gls{ptg} is a \gls{pta} where the set of actions $\Actions$ is partitioned into a set of \textit{controllable} actions and \textit{uncontrollable} actions.

\begin{definition}[\glsxtrlong{ptg}]
A \gls{ptg} is defined as a tuple $\PTG = \PTGDef$ where $\PTA = \PTADef$ is a \gls{pta}, $\ControllableActions$ is the finite set of \textit{controllable actions}, $\UncontrollableActions$ is the finite set of \textit{uncontrollable actions}, s.t., $\Actions = \ControllableActions \cup \UncontrollableActions$ and $\ControllableActions \cap \UncontrollableActions = \emptyset$.
\end{definition}

Alternatively, a \gls{ptg} can be defined from scratch by enumerating the \gls{pta} as $\PTG = \PTGDefFromScratch$. \figurename~\ref{fig:example-ptg} extends the example of \figurename~\ref{fig:example-pta} in a \gls{ptg}.
    \subsubsection{\glsfmtlong{ptmdp}.}
    A \gls{ptmdp} is a \gls{ptg} where the uncontrollable player has a predefined strategy modeling the environment.

\begin{definition}[\glsxtrlong{ptmdp}]
A \gls{ptmdp} is defined as a couple $\PTMDP = \PTMDPDef$ where $\PTG = \PTGDefFromScratch$ is a \gls{ptg} and $\density{} :  \Locations \times \Valuations \mapsto (\Delays \times \UncontrollableActions \mapsto [0,1])$ gives a density function for each location $\loc$ and valid clock valuation $v\in \Valuations$ s.t. for $B = \{\delay \in \Delays \mid \forall \delay' \in [0, \delay], v+\delay' \vDash \invariants(\loc)\}$, it holds,
\begin{enumerate}[(i)]
    \item \label{item:ptmdp-sum-1} $\sum_{\action \in \UncontrollableActions}\int_{\delay \in B} \density{(\loc,v)}(\delay, \action)= 1$ and,
    \item \label{item:ptmdp-valid-stay} For $\action \in \UncontrollableActions$ and $\delay \in \Delays \setminus B$, $\density{(\loc,v)}(\delay, \action)= 0$ and,
    \item \label{item:ptmdp-valid-transition} For $\action \in \UncontrollableActions$ and $\delay \in B$, if $\density{(\loc,v)}(\delay, \action) > 0$, there exists $s\in \Guard(\Clocks)$ and $\loc'\in \Locations$ s.t. $(\loc, s, \action, \loc') \in \Transitions$ and $v+\delay \vDash s \land \invariants(\loc')$.
\end{enumerate}
\end{definition}
}{
    \subsubsection{\glsfmtlong{ptg}.}
    A \gls{ptg} is two-player game on an automata with cost and time attributes.

\begin{definition}[\glsxtrlong{pta}]
A \gls{pta} is defined as a tuple $\PTG = \PTGDef$ where $\Locations$ is the finite set  of \textit{locations}, $\loc_0\in \Locations$ is the \textit{initial location}, $\Clocks$ is a set of synchronous \textit{clocks}, $\ControllableActions$ is the finite set of \textit{controllable actions}, $\UncontrollableActions$ is the finite set of \textit{uncontrollable actions}, $\Transitions \subseteq \Locations \times \Guard(\Clocks) \times (\ControllableActions \cup \UncontrollableActions) \times \Locations$ is a transition relation,
%s.t. $(\loc, s, e, \loc')\in \Transitions$ and $(\loc, s', e, \loc'')\in \Transitions$ implies $s=s'$ and $\loc'= \loc''$,
$\price: \Locations\cup\Transitions \mapsto \mathbb N$ assigns \textit{cost rates} to locations and \textit{costs} to edges, $\clockReset: \Transitions\mapsto 2^\Clocks$ gives the set of clocks that are reset after a transition, and $\invariants: \Locations \mapsto \Guard(\Clocks)$ assigns \textit{invariants} to locations.
\end{definition}
    \subsubsection{\glsfmtlong{ptmdp}.}
    A \gls{ptmdp} is a \gls{ptg} where the uncontrollable player has a predefined strategy modeling the environment.

\begin{definition}[\glsxtrlong{ptmdp}]
A \gls{ptmdp} is defined as a couple $\PTMDP = \PTMDPDef$ where $\PTG = \PTGDefFromScratch$ is a \gls{ptg} and $\density{} :  \Locations \times \Valuations \mapsto (\Delays \times \UncontrollableActions \mapsto [0,1])$ gives a density function for each location $\loc$ and valid clock valuation $v\in \Valuations$ s.t. for $B = \{\delay \in \Delays \mid \forall \delay' \in [0, \delay], v+\delay' \vDash \invariants(\loc)\}$, it holds,
\begin{enumerate}[(i)]
    \item \label{item:ptmdp-sum-1} $\sum_{\action \in \UncontrollableActions}\int_{\delay \in B} \density{(\loc,v)}(\delay, \action)= 1$ and,
    \item \label{item:ptmdp-valid-stay} For $\action \in \UncontrollableActions$ and $\delay \in \Delays \setminus B$, $\density{(\loc,v)}(\delay, \action)= 0$ and,
    \item \label{item:ptmdp-valid-transition} For $\action \in \UncontrollableActions$ and $\delay \in B$, if $\density{(\loc,v)}(\delay, \action) > 0$, there exists $s\in \Guard(\Clocks)$ and $\loc'\in \Locations$ s.t. $(\loc, s, \action, \loc') \in \Transitions$ and $v+\delay \vDash s \land \invariants(\loc')$.
\end{enumerate}
\end{definition}
}{
    \shrinkalt{}{\begin{figure}
    \centering
    \begin{tikzpicture}[
            ->,
            >=stealth',
            shorten >=1pt,
            auto,
            node distance=2.8cm,
            semithick,
            sloped
        ]

        \node[state, very thick, label={[invariant]below:{$x\leq 5$}}] (A) {$\loc_0$};
        \node[state] (B) [above left of=A, label={[invariant]above:{$x\leq 6$}}]  {$\loc_1$};
        \node[state] (D) [above right of=A, label={[invariant]right:{$x\leq 3$}}] {$\loc_2$};
        \node[state] (C) [above right of=B, label={[invariant]above:{$x\leq 3$}}] {$\loc_3$};
        \node[state] (E) [right of=D, label={[invariant]above:{$x\leq 2$}}] {$\loc_4$};
        \node[weight] (weightE) [above=0.35 of E]{$c'=1$};
    
        \path (A) edge node [scaledown, above] {$\action_1$} node [guard, below] {$x = 3$}(B);
        \path (A) edge node [scaledown] {$\action_2$} node [clock, below] {$x\leftarrow 0$} (C);
        \path (B) edge [loop left, dashed] node [scaledown] {$\action_4$} (B);
        \path (B) edge [dashed] node [scaledown] {$\action_3$} node [clock, below] {$x\leftarrow 0$} (C);
        \path (C) edge node [scaledown] {$\action_1$} node [weight, below] {$c \mathrel{+}=1$} (D);
        \path (C) edge [bend left, dashed]  node [scaledown] {$\action_4$} node [guard, below] {$x < 1$} node [clock, below=0.3] {$x \leftarrow 0$} (E);
        \path (D) edge [loop below] node [scaledown] {$\action_2$} (D);
        \path (D) edge [dashed] node [scaledown] {$\action_3$} node [weight, below] {$c\mathrel{+}=2$} (A);
        \path (E) edge [bend left]  node [scaledown, below=0.3] {$\action_1$} node [guard, below] {$x \leq 5 \land x \geq 1$} (A);
        
        \matrix [draw, right of=E, column sep=5pt, row sep=3pt, inner sep=2pt, ampersand replacement=\&, column 1/.style={anchor=base}, column 2/.style={anchor=base west}]{
        \node [state, minimum size=2em, scale=0.7] {$\loc$}; \& \node[scale=0.8] {location $\loc$};\\
        \draw[dashed, line width=0.5pt, >=stealth', semithick](-0.25, 0) -- (0.25,0); \& \node[right, scale=0.8, align=left] {uncontrollable\\[-0.5em]transition};\\
        \draw[line width=0.5pt, >=stealth', semithick](-0.25, 0) -- (0.25,0); \& \node[right, scale=0.8, align=left] {controllable\\[-0.5em]transition};\\
        \node [invariant, scale=0.8] {$x \bowtie y$}; \& \node[right, scale=0.8, align=left] at (0, 0.2em) {location\\[-0.5em]invariant};\\
        \node [guard, scale=0.8] {$x \bowtie y$}; \& \node[scale=0.8, align=left] {edge guard};\\
        \node [weight, scale=0.8] {$c \bowtie y$}; \& \node[right, scale=0.8, align=left] at (0, 0.2em) {location and\\[-0.5em]transition cost};\\
        \node [scale=0.8] {$\action^x_y$}; \& \node[scale=0.8, align=left] {action type};\\
        \node [clock, scale=0.8] {$x \leftarrow 0$}; \& \node[scale=0.8, align=left] {clock reset};\\
        };
    \end{tikzpicture}
    \caption[Example of the structure of a \glsfmtshort{ptmdp}.]{Example of the structure of the \gls{ptmdp} $\PTMDP = \PTMDPDef$ where $\ControllableActions= \{\action_1, \action_2\}$, and $\UncontrollableActions=\{\action_3, \action_4\}$. The variable $c$ represent the cost. We assume that the cost is null where there it is omitted.}
    \label{fig:example-ptmdp}
\end{figure}
    \subsubsection{\glsfmtlong{ptmdp}.}}
    A \gls{ptmdp} is a two-player game on a priced timed stochastic game structure \shrinkalt{}{(example in Fig.~\ref{fig:example-ptmdp})} where a player has a predefined strategy modeling the environment.

\begin{definition}[\glsxtrlong{ptmdp}]
A \gls{ptmdp} is a tuple $\PTMDP = \PTMDPDef$ where $\Locations$ is the finite set  of \textit{locations}, $\loc_0\in \Locations$ is the \textit{initial location}, $\Clocks$ is a set of synchronous \textit{clocks}, $\ControllableActions$ is the finite set of \textit{controllable actions}, $\UncontrollableActions$ is the finite set of \textit{uncontrollable actions}, $\Transitions \subseteq \Locations \times \Guard(\Clocks) \times (\ControllableActions \cup \UncontrollableActions) \times \Locations$ is a transition relation,
%s.t. $(\loc, s, e, \loc')\in \Transitions$ and $(\loc, s', e, \loc'')\in \Transitions$ implies $s=s'$ and $\loc'= \loc''$,
$\price: \Locations\cup\Transitions \mapsto \mathbb N$ assigns \textit{cost rates} to locations and \textit{costs} to edges, $\clockReset: \Transitions\mapsto 2^\Clocks$ gives the set of clocks reset after a transition, $\invariants: \Locations \mapsto \Guard(\Clocks)$ assigns \textit{invariants} to locations, and $\density{} :  \Locations \times \Valuations \mapsto (\Delays \times \UncontrollableActions \mapsto [0,1])$ gives a density function for each location $\loc$ and valid valuation $v\in \Valuations$ s.t. for $B = \{\delay \in \Delays \mid \forall \delay' \in [0, \delay], v+\delay' \vDash \invariants(\loc)\}$, it holds:
\begin{itemize}
    \item \label{item:ptmdp-sum-1} $\sum_{\action \in \UncontrollableActions}\int_{\delay \in B} \density{(\loc,v)}(\delay, \action)= 1$ and,
    \item \label{item:ptmdp-valid-stay} For $\action \in \UncontrollableActions$ and $\delay \in \Delays \setminus B$, $\density{(\loc,v)}(\delay, \action)= 0$ and,
    \item \label{item:ptmdp-valid-transition} For $\action \in \UncontrollableActions$ and $\delay \in B$, if $\density{(h,v)}(\delay, \action) > 0$, there exists $s\in \Guard(\Clocks)$, $\loc'\in \Locations$, and $e=(\loc, s, \action, \loc') \in \Transitions$ s.t. $v+\delay \vDash \invariants(\loc) \land s$ and $(v+b)[\clockReset(e)] \vDash \invariants(\loc')$.
\end{itemize}
\end{definition}
}

In the definition, $\density{(\loc,v)}(\delay, \action)$ is the density for the environment aiming at taking an uncontrollable action $\action\in \UncontrollableActions$ after a delay $\delay\in \Delays$ respecting the transitions and invariants. The set $B$ contains delays s.t. it is still possible to stay in $\loc$.\shrinkalt{}{The condition~(\ref{item:ptmdp-sum-1}) insures that, for all $\loc \in \Locations$ and $v\in \Valuations$, $\density{(\loc,v)}$ is a probability density. The condition~(\ref{item:ptmdp-valid-stay}) imposes that the density takes an action while it is still possible to stay in the current location, and the condition~(\ref{item:ptmdp-valid-transition}) imposes that the possible transitions lead to valid states with valid valuations.}\fromptaptg{For the rest of this paper, when a general \gls{ptmdp} $\PTMDP$ is given, we assume $\PTMDP=\PTMDPDef$. with $\PTG = \PTGDefFromScratch$ a general \gls{ptg}.}{For the rest of this paper, when a general \gls{ptmdp} $\PTMDP$ is given, we assume $\PTMDP=\PTMDPDef$. with $\PTG = \PTGDefFromScratch$ a general \gls{ptg}.}{For the rest of this paper, when a general \gls{ptmdp} $\PTMDP$ is given, we assume $\PTMDP=\PTMDPDef$.}
\shrinkalt{}{Given a \gls{ptmdp} $\PTMDP$, we denote $\Locations^*$ the set of finite suits of locations respecting the transitions, called \textit{history}.}
For a location $\loc\in \Locations$, we say that a valuation $v$ on $\Clocks$ is valid in $\loc$ if $v\vDash \invariants(\loc)$.

\shrinkalt{}{\subsubsection{Strategies.}}

\shrinkalt{
    The concept of \textit{memoryless strategy} on a \gls{ptmdp} is formalized as follows: intuitively it is a function that assigns a density to the subsequent possible actions of the player given the current state of the game.
    \begin{definition}[Memoryless Strategy]
A memoryless strategy $\strategy{}$ over a PTMDP $\PTMDP$ is a function $\strategy{} : \strategy{}: \Locations \times \Valuations \mapsto (\Delays \times \ControllableActions \mapsto [0,1])$ s.t. for $\loc \in \Locations$, $v \in \Valuations$, and $B = \{\delay \in \Delays \mid \forall \delay' \in [0, \delay], v+\delay' \vDash \invariants(\loc)\}$, it holds:
\begin{itemize}
    \item $\sum_{\action \in \ControllableActions}\int_{\delay \in B} \strategy{(\loc,v)}(\delay, \action)= 1$ and,
    \item For $\action \in \ControllableActions$ and $\delay \in \Delays \setminus B$, $\strategy{(\loc,v)}(\delay, \action)= 0$ and,
    \item For $\action \in \ControllableActions$ and $\delay \in \Delays$, if $\strategy{(\loc,v)}(\delay, \action) > 0$, there exists $s\in \Guard(\Clocks)$, $\loc'\in \Locations$, and $e=(\loc, s, \action, \loc') \in \Transitions$ s.t. $v+\delay \vDash \invariants(\loc) \land s$ and $(v+b)[\clockReset(e)] \vDash \invariants(\loc')$.
\end{itemize}
\end{definition}
    We can notice the similarity with the environment's $\density{}$, which is a memoryless strategy for uncontrollable actions. We can extend the definition by allowing Dirac distributions for discrete probabilities and adding an extra action $\action_\text{wait}$ which means waiting forever and is available only is $v+\delay \vDash \invariants(\loc)$ for all $\delay \in \Delays$.
}{
    The concept of \textit{strategy} on a \gls{ptmdp} is formalized as follows: intuitively, it is a function that assigns a density to the subsequent possible actions of the player given the state of the game.
    \begin{definition}[Strategy]
A strategy $\strategy{}$ over a PTMDP $\PTMDP$ is a function $\strategy{}: \Locations^* \times \Valuations \mapsto (\Delays \times \ControllableActions \mapsto [0,1])$ s.t. for $h = (\loc_0, \dots, \loc_n) \in \Locations^*$, $v \in \Valuations$, and $B = \{\delay \in \Delays \mid \forall \delay' \in [0, \delay], v+\delay' \vDash \invariants(\loc_n)\}$, it holds
\begin{itemize}
    \item $\sum_{\action \in \ControllableActions}\int_{\delay \in B} \strategy{(h,v)}(\delay, \action)= 1$ and,
    \item For $\action \in \ControllableActions$ and $\delay \in \Delays \setminus B$, $\strategy{(h,v)}(\delay, \action)= 0$ and,
    \item For $\action \in \ControllableActions$ and $\delay \in \Delays$, if $\strategy{(h,v)}(\delay, \action) > 0$, there exists $s\in \Guard(\Clocks)$, $\loc'\in \Locations$, and $e=(\loc, s, \action, \loc') \in \Transitions$ s.t. $v+\delay \vDash \invariants(\loc) \land s$ and $(v+b)[\clockReset(e)] \vDash \invariants(\loc')$.
\end{itemize}
\end{definition}
    We can extend the definition by allowing Dirac distributions for discrete probabilities and adding an extra action $\action_\text{wait}$ which means waiting forever and is available only is $v+\delay \vDash \invariants(\loc)$ for all $\delay \in \Delays$.
    
    We will now describe some classes of strategies.
    A \textit{deterministic strategy} assigns only deterministic probability distributions.
    \begin{definition}[Deterministic strategy]
We say that a strategy $\strategy{}$ is deterministic if the codomain of $\strategy{(q)}$ is $\{0, 1\}$ for all $q \in \Locations^* \times \Delays$.
\end{definition}
    A \textit{memoryless strategy} depends only on the last state of the history and the valuation.
    \begin{definition}[Memoryless strategy]
We say that a strategy $\strategy{}$ is memoryless if for two histories $h = (\loc_0, \dots, \loc_n)$ and $h' = (\loc_0', \dots, \loc_k')$ s.t. $\loc_n = \loc_k'$, and a valuation $v$, we have $\strategy{(h,v)} = \strategy{(h',v)}$.
\end{definition}
    A controller with a \textit{non-lazy strategy} will either move directly from a location or wait until the environment makes a move.
    \begin{definition}[Memoryless non-lazy strategy]
We say that a strategy $\strategy{}$ is memoryless non-lazy if for all history $h \in \Locations^*$, valuation $v\in \Valuations$, and action $\action \in \ControllableActions \cup \{\action_\text{wait}\}$
%either $\strategy{(h,v)}(\delay, \action) = 0$ when $\delay > 0$, or $\strategy{(h,v)}(\delay, \action) = 0$ when $\delay < \infty$.
$\strategy{(h,v)}(\delay, \action) = 0$ when $\delay > 0$. Notice that waiting forever is possible.
\end{definition}
    We can notice the similarity with the environment's $\density{}$, which is a memoryless strategy for uncontrollable actions.
}

\shrinkalt{}{\subsubsection{Runs.}}

Let $\PTMDP$ be a \gls{ptmdp} and $\LocVal = \{(\loc, v) \in \Locations \times \Valuations \mid v\vDash \invariants(\loc)\}$ be the valid state-valuation pairs. \textit{Runs} are valid sequences from $\LocVal$ recording time and cost.

\begin{definition}[Run]
For a \gls{ptmdp} $\PTMDP$, let $\RunElements \subseteq \LocVal \times (\Delays \cup \ControllableActions \cup \UncontrollableActions) \times \mathbb{N}^2 \times \LocVal$, s.t. for $((\loc_1, v_1), e, t, c, (\loc_2,v_2)) \in \RunElements$, $e$ represents an action or a delay, $t$ the cumulative time, $c$ the cumulative cost, and we have
\begin{itemize}
    \item if $e \in \Delays$ then $\loc_2 = \loc_1$ and $v_2 = v_1 + e$.
    \item if $e \in \ControllableActions \cup \UncontrollableActions$ then there exists $s \in \Guard(\Clocks)$ with $v_1 \vDash s$ such that $e' = (\loc_1, s, e, \loc_2) \in \Transitions$ and
    \shrinkalt{$v_2(x) = 0$ if $x\in \clockReset(e')$ or $v_2(x) = v_1(x)$ otherwise.}{$$v_2(x) = \begin{cases}
        0 & \text{if } x\in \clockReset(e')\\
        v_1(x) & \text{otherwise}
        \end{cases}$$}
\end{itemize}
\shrinkalt{
Let $\nexttc(t, c, e, (\loc, v)) = \begin{cases}
t + e, c + e \price(\loc) & \text{ if } e\in \Delays\\
t, c+ \price(e) & \text{ if } e\in \ControllableActions \cup \UncontrollableActions
\end{cases}$ be defined on $\mathbb{N}^2\times (\Delays \cup \ControllableActions \cup \UncontrollableActions) \times \LocVal$ and returning the next cumulative time and cost.
}{
Let $\nexttc: \mathbb{N}^2\times (\Delays \cup \ControllableActions \cup \UncontrollableActions) \times \LocVal \mapsto \mathbb{N}^2$ be a function returning the next cumulative time and cost, s.t, for $(t, c, e, (\loc, v)) \in \mathbb{N}^2\times (\Delays \cup \ControllableActions \cup \UncontrollableActions) \times \LocVal$,
$$\nexttc(t, c, e, (\loc, v)) = \begin{cases}
t + e, c + e \price(\loc) & \text{ if } e\in \Delays\\
t, c+ \price(e) & \text{ if } e\in \ControllableActions \cup \UncontrollableActions
\end{cases}$$

}
A run in $\PTMDP$ is a finite or infinite sequence $S$ of elements of $\RunElements$ s.t. two consecutive elements $(q_1, e, t, c, q_2)$ and $(q_1', e', t', c', q_2')$ verify $q_2 = q_1'$ and $(t', c') = \nexttc(t, c, e', q_1')$.
%\begin{inparaenum}[(i)]
%    \item $q_2 = q_1'$,
%    \item $(t', c') = \nexttc(t, c, e', q_1')$.
    %\item if $e\in \Delays$ then $t' = t + e$ and $c'=c + e \price(\loc_1')$, and
    %\item if $e \in \ControllableActions \cup \UncontrollableActions$ then $t'=t$ and $c'=c+ \price(e)$.
%\end{inparaenum}
We denote $\Runs$ the set of runs, $\Runs^k$ the set of runs of length $k\in \mathbb{N}$, and $\Runs_0$ the set of \shrinkalt{runs from the top,}{runs s.t. the initial location is $\loc_{0}$, the initial valuation is the null function $v_0$, the initial cumulative time and cost is $(t_1, c_1) = \nexttc(0, 0, e_1, (\loc_0, v_0))$ where $e_1$ is the initial delay or action,} i.e., $\Runs_0 = \{(q_i, e_i, t_i, c_i, q_i')_{i \in \{1,\dots\}} \in \Runs \mid q_1 = (\loc_{0}, v_0) \land (t_1, c_1) = \nexttc(0, 0, e_1, (\loc_0, v_0))\}$ \shrinkalt{where $v_0$ is the null valuation}{}.
\end{definition}

For $k\in \mathbb N$ and a run $r\in \Runs^k$ we denote $\cumulativeTime{r} = c$ and $\cumulativeCost{r} = c$ where $(q_1, e, t, c q_2)$ is the last element of $r$.
A \gls{ptmdp} $\PTMDP$ with a controller strategy $\strategy{}$ defines a probability measure $\mathbb{P}_{\PTMDP, \strategy{}}$ on subsets of $\Runs_0$ giving their probability to happen. Consequently, $\cumulativeTimeDef$ and (resp. $\cumulativeCostDef$) can be seen as a random variable giving the time (resp. cost) of a possible run on the probability space $(\Runs_0, 2^{\Runs_0}, \mathbb{P}_{\PTMDP, \strategy{}})$ induced by $\mathbb{P}_{\PTMDP, \strategy{}}$. We denote \shrinkalt{$\runExpect{\PTMDP}{\strategy{}}[X]$}{$\runExpect{\PTMDP}{\strategy{}}[X] = \int_{\Runs_0} X \dd{\runProba{\PTMDP}{\strategy{}}}$} the expected value of a random variable $X$ and \shrinkalt{$\runExpect{\PTMDP}{\strategy{}}[X\mid R]$}{$\runExpect{\PTMDP}{\strategy{}}[X\mid R] = \int_{r\in \Runs_0} X(r) \frac{\runProba{\PTMDP}{\strategy{}}[X=r \cap R]}{\runProba{\PTMDP}{\strategy{}}[R]}\dd{r}$} its conditional expectation given an event $R\subseteq \Runs_0$.

%\section{Reasoning about \glsfmtshort{mtd} in Attack Modeling Formalisms}
%\label{sec:whatYouDid}

\shrinkalt{}{
\section{Motivating Example}\label{sec:example}
\begin{figure}[t]
    \centering
    \begin{tikzpicture}[
            ->,
            >=stealth',
            shorten >=1pt,
            auto,
            sloped,
            semithick,
            xscale=0.78,
            yscale=0.7
        ]
        
        \node[goalUseCase, label=300:{$\lor$}] at (0, 0) (g_0) {Fake electricity\\consumption report ($\AMGroot$)};
        \node[goalUseCase, label=268:{$\land$}] at (-5, -1) (g_tc) {Temper with\\communication\\($g_{tc}$)};
        \node[goalUseCase, label=below:{$\land$}] at (-0.5, -2) (g_th) {Tamper with device's\\hardware ($g_{th}$)};
        \node[goalUseCase, label=-70:{$\land$}] at (5, -1) (g_ts) {Temper with device's\\software ($g_{ts}$)};
        \node[attackUseCase, label={}] at (-6, -3.5) (a_ad) {Alter\\data\\($a_{ad}$)};
        \node[defenseUseCase, label={}] at (-3, -2.1) (d_dk) {Dynamic\\key ($d_{dk}$)};
        \node[attackUseCase, label={}] at (-4.2, -3.5) (a_ic) {Intercept\\connection\\($a_{ic}$)};
        \node[attackUseCase, label={}] at (1, -3.5) (a_p) {Plug into\\device ($a_{p}$)};
        \node[goalUseCase, label=below:{}] at (-2, -3.5) (g_up) {Understand\\protocol ($g_{up}$)};
        \node[attackUseCase] at (-2.7,-6) (a_sp) {study\\protocol\\($a_{sp}$)};
        \node[defenseUseCase, label={}] at (-0.5, -5.5) (d_cp) {Change\\protocol\\($d_{cp}$)};
        \node[goalUseCase, label=260:{$\lor$}] at (4, -3) (g_ac) {Get admin's\\credentials\\($g_{ac}$)};
        \node[attackUseCase] at (1.5,-5.5) (a_bf) {Brute\\force\\($a_{bf}$)};
        \node[attackUseCase] at (3.5,-6) (a_ss) {Shoulder\\surfing\\($a_{ss}$)};
        \node[defenseUseCase, label={}] at (5.5, -6) (d_cc) {Change\\credentials\\($d_{cc}$)};
        \node[goalUseCase, label=below:{}] at (7.7, -3) (g_hs) {Hack\\system\\($g_{hs}$)};
        \node[attackUseCase] at (7.7,-6) (a_fue) {Find \& use\\exploit ($a_{fue}$)};
        \node[defenseUseCase, label={}] at (6.3, -4) (d_dsr) {Dynamic\\software\\rotation\\($d_{dsr}$)};

        \path (g_0) edge (g_tc) edge (g_th) edge (g_ts);
        \path (g_tc) edge (a_ad) edge (a_ic);
        \path (g_th) edge (g_up) edge (a_p);
        \path (g_up) edge (a_sp);
        \path (g_ts) edge (g_ac) edge (g_hs);
        \path (g_ac) edge (a_bf) edge (a_ss);
        \path (g_th) edge (g_ac);
        \path (g_hs) edge (a_fue);
        \path (d_dk) edge [-, dashed, bend right] (a_ad);
        \path (d_cp) edge [-, dashed] (a_sp) edge [-, dashed] (g_up);
        \path (d_cc) edge [-, dashed, bend right] (a_bf) edge [-, dashed, bend right] (a_ss) edge [-, dashed] (g_ac);
        \path (d_dsr) edge [-, dashed] (a_fue);
        
        \matrix [draw, below left, column sep=3pt, row sep=0pt, inner sep=2pt, ampersand replacement=\&, column 1/.style={anchor=base}, column 2/.style={anchor=base west}] at (-3.6, -4.4) {
        \draw[line width=0.5pt, >=stealth', semithick](-0.2, 0) -- (0.2, 0); \& \node[right, scale=0.8, align=left] {edge};\\
        \draw[dashed,line width=0.5pt, semithick, -](-0.2, 0) -- (0.2, 0); \& \node[right, scale=0.8, align=left] {associated\\[-0.5em]defense};\\
        \node [goalUseCase, minimum size=1em, scale=1] {\phantom{a}}; \& \node[scale=0.8]{subgoal};\\
        \node [attackUseCase, minimum size=0.8em, scale=0.8] {\phantom{a}}; \& \node[scale=0.8]{atomic attack}; \\
        \node [defenseUseCase, radius=3pt, minimum size=1em, scale=1] {\phantom{a}}; \& \node[scale=0.8]{\glsfmtshort{mtd}}; \\
        };
                
    \end{tikzpicture}
    \caption[Example of a \glsfmtshort{dag}-based structure with \glsfmtshortpl{mtd} for an electricity meter.]{Example of a \gls{dag}-based structure with \glspl{mtd} for an electricity meter\shrinkalt{.}{ extending the \gls{at} in \figurename~\ref{fig:example-at}. Notice the edge between $g_{th}$ and $g_{ac}$ making the structure a \gls{dag} and no longer a tree.}
    %Attacker's subgoals are in black, atomic attacks are in red, MTDs are in blue.
    Refinements are below subgoals or omitted for single child subgoals.
    %The dashed edges link a MTD with the node it defends.
    %An example of an AMG $\AMG = \AMGDef$, with subgoal  set $\Goals=\{\AMGroot, g_{tc}, g_{th}, g_{ts}, g_{up}, g_{ac}, g_{hs}\}$ \vad{you use this structure before defining it.}, atomic attack set $\AtomicAttacks=\{a_{ad}, a_{ic}, a_{sp}, a_p, a_{bf}, a_{ss}, a_{fue}\}$, and MTD set $\Defenses = \{d_{dk}, d_{cp}, d_{cc}, d_{dsr}\}$. The nodes in black are subgoals, the nodes in red are atomic attacks, and the nodes in blue with rounded corners are MTDs. The refinement of subgoals is indicated below theme. The refinement can be omitted when the subgoal has only one atomic attack or subgoal child. The dashed edges link a MTD $d$ with a node $n$ such that $d\in \defense{n}$.
    }
    \label{fig:example-adt}
\end{figure}
After studying an electricity meter, a team of engineers from an electricity provider can find potential threats and attack paths to report wrong electricity consumption. They obtain the \gls{dag}-based graph of \figurename~\ref{fig:example-adt} composed of black rectangles defining subgoals and the red octagons representing atomic attacks. The engineers also identify the features of a potential attack: the duration of each atomic attack step, their success probability, and their cost for the attacker.
The electricity providers have a set of \glspl{mtd} to harden the attacker's task. The \glspl{mtd} are the nodes in red in \figurename~\ref{fig:example-adt}. They are attached to nodes that they defend by resetting the attack on the nodes.
\figurename~\ref{fig:example-adt} can be read as follows. In order to achieve the main goal of reporting a fake consumption ($g_0$), the attacker can \emph{either} tamper with the communication ($g_{tc}$), the device hardware ($g_{th}$), or the device software ($g_{ts}$). Moreover, the attacker must intercept the connection ($a_{ic}$) \emph{and} alter the data ($a_{ad}$) to temper the communication. However, the attack of altering the data is protected by an \gls{mtd} that changes the communication key periodically ($d_{dk}$). The rest of the graph can be interpreted in similar ways.

The problem is that \glspl{mtd} are costly. For example, changing the communication key when nobody is attacking the system makes the communication longer for the regular user. As a result, the defender must parametrize the \glspl{mtd} carefully. This parameter is the activation frequency of the different \glspl{mtd}. The higher the frequency is, the more the quality of service is impacted, and the more the \gls{mtd} prevents the attack.
Thus, we need to find a way to evaluate a given configuration of \glspl{mtd}. The attacker is in a multi-objective optimization situation because he needs to minimize his cost and attack time. Moreover, the cost and the time are a density function given a strategy for the attacker because the success of the attacks and defenses is stochastic.
%First, there is no obvious way to compare density distributions. In~\cite{rass2015game}, the authors use the moments. Given the lists of moments $(m_1,m_2,\dots)$ of the first distribution $d$ and the list of moments $(m'_1,m'_2,\dots)$ of the second distribution $d'$, we say that $d<d'$ if $m_i<m'_i$ where $i$ is the smallest integer such that $m_i\neq m'_i$. In short it compares first the average values, if equal it compares the variance, etc.
%Second, here we do a multi-objective optimization where we have to compare two pairs of densities (the cost density and the time density).
%Our method does not try to solve this multi-objective optimization problem but will make these two densities available for the user.
}

\section{\glsfmtlong{amg}}\label{sec:AMG}

This section presents the \gls{amg}, our extension of AT~\cite{mauw2005foundations} for \glspl{mtd}.
\shrinkalt{\glspl{mtd} are a type of defense that forces the attacker to redo attacks because the system state has changed. Consequently, we need a model that considers temporal aspects.
In order to allow various attack paths, we will use a rooted \gls{dag} that is more general than a \shrinkalt{tree.}{tree. We will allow conjunction and disjunction refinements for subgoals.}
%However, we will not use the sequential AND operator (like in~\cite{lv2011space,ardi2006towards,arnold2014time}) in this work because it would first need to be decided what happens when some of the first children of the sequential AND are reset by a \gls{mtd}. This would add an extra layer of complexity but is an interesting question for a future work.
}{\subsection{Objective}\glspl{mtd} are a type of defense that forces the attacker to redo attacks because the system state has changed. Consequently, we need a model that considers temporal aspects.
In order to allow various attack paths, we will use a rooted \gls{dag} that is more general than a \shrinkalt{tree.}{tree. We will allow conjunction and disjunction refinements for subgoals.}
%However, we will not use the sequential AND operator (like in~\cite{lv2011space,ardi2006towards,arnold2014time}) in this work because it would first need to be decided what happens when some of the first children of the sequential AND are reset by a \gls{mtd}. This would add an extra layer of complexity but is an interesting question for a future work.
}
\shrinkalt{\begin{figure}[t]
    \centering
    \begin{tikzpicture}[
            ->,
            >=stealth',
            shorten >=1pt,
            auto,
            sloped,
            semithick,
            xscale=0.78,
            yscale=0.7
        ]
        
        \node[goalUseCase, label=300:{$\lor$}] at (0, 0) (g_0) {Fake electricity\\consumption report ($\AMGroot$)};
        \node[goalUseCase, label=268:{$\land$}] at (-5, -1) (g_tc) {Temper with\\communication\\($g_{tc}$)};
        \node[goalUseCase, label=below:{$\land$}] at (-0.5, -2) (g_th) {Tamper with device's\\hardware ($g_{th}$)};
        \node[goalUseCase, label=-70:{$\land$}] at (5, -1) (g_ts) {Temper with device's\\software ($g_{ts}$)};
        \node[attackUseCase, label={}] at (-6, -3.5) (a_ad) {Alter\\data\\($a_{ad}$)};
        \node[defenseUseCase, label={}] at (-3, -2.1) (d_dk) {Dynamic\\key ($d_{dk}$)};
        \node[attackUseCase, label={}] at (-4.2, -3.5) (a_ic) {Intercept\\connection\\($a_{ic}$)};
        \node[attackUseCase, label={}] at (1, -3.5) (a_p) {Plug into\\device ($a_{p}$)};
        \node[goalUseCase, label=below:{}] at (-2, -3.5) (g_up) {Understand\\protocol ($g_{up}$)};
        \node[attackUseCase] at (-2.7,-6) (a_sp) {study\\protocol\\($a_{sp}$)};
        \node[defenseUseCase, label={}] at (-0.5, -5.5) (d_cp) {Change\\protocol\\($d_{cp}$)};
        \node[goalUseCase, label=260:{$\lor$}] at (4, -3) (g_ac) {Get admin's\\credentials\\($g_{ac}$)};
        \node[attackUseCase] at (1.5,-5.5) (a_bf) {Brute\\force\\($a_{bf}$)};
        \node[attackUseCase] at (3.5,-6) (a_ss) {Shoulder\\surfing\\($a_{ss}$)};
        \node[defenseUseCase, label={}] at (5.5, -6) (d_cc) {Change\\credentials\\($d_{cc}$)};
        \node[goalUseCase, label=below:{}] at (7.7, -3) (g_hs) {Hack\\system\\($g_{hs}$)};
        \node[attackUseCase] at (7.7,-6) (a_fue) {Find \& use\\exploit ($a_{fue}$)};
        \node[defenseUseCase, label={}] at (6.3, -4) (d_dsr) {Dynamic\\software\\rotation\\($d_{dsr}$)};

        \path (g_0) edge (g_tc) edge (g_th) edge (g_ts);
        \path (g_tc) edge (a_ad) edge (a_ic);
        \path (g_th) edge (g_up) edge (a_p);
        \path (g_up) edge (a_sp);
        \path (g_ts) edge (g_ac) edge (g_hs);
        \path (g_ac) edge (a_bf) edge (a_ss);
        \path (g_th) edge (g_ac);
        \path (g_hs) edge (a_fue);
        \path (d_dk) edge [-, dashed, bend right] (a_ad);
        \path (d_cp) edge [-, dashed] (a_sp) edge [-, dashed] (g_up);
        \path (d_cc) edge [-, dashed, bend right] (a_bf) edge [-, dashed, bend right] (a_ss) edge [-, dashed] (g_ac);
        \path (d_dsr) edge [-, dashed] (a_fue);
        
        \matrix [draw, below left, column sep=3pt, row sep=0pt, inner sep=2pt, ampersand replacement=\&, column 1/.style={anchor=base}, column 2/.style={anchor=base west}] at (-3.6, -4.4) {
        \draw[line width=0.5pt, >=stealth', semithick](-0.2, 0) -- (0.2, 0); \& \node[right, scale=0.8, align=left] {edge};\\
        \draw[dashed,line width=0.5pt, semithick, -](-0.2, 0) -- (0.2, 0); \& \node[right, scale=0.8, align=left] {associated\\[-0.5em]defense};\\
        \node [goalUseCase, minimum size=1em, scale=1] {\phantom{a}}; \& \node[scale=0.8]{subgoal};\\
        \node [attackUseCase, minimum size=0.8em, scale=0.8] {\phantom{a}}; \& \node[scale=0.8]{atomic attack}; \\
        \node [defenseUseCase, radius=3pt, minimum size=1em, scale=1] {\phantom{a}}; \& \node[scale=0.8]{\glsfmtshort{mtd}}; \\
        };
                
    \end{tikzpicture}
    \caption[Example of a \glsfmtshort{dag}-based structure with \glsfmtshortpl{mtd} for an electricity meter.]{Example of a \gls{dag}-based structure with \glspl{mtd} for an electricity meter\shrinkalt{.}{ extending the \gls{at} in \figurename~\ref{fig:example-at}. Notice the edge between $g_{th}$ and $g_{ac}$ making the structure a \gls{dag} and no longer a tree.}
    %Attacker's subgoals are in black, atomic attacks are in red, MTDs are in blue.
    Refinements are below subgoals or omitted for single child subgoals.
    %The dashed edges link a MTD with the node it defends.
    %An example of an AMG $\AMG = \AMGDef$, with subgoal  set $\Goals=\{\AMGroot, g_{tc}, g_{th}, g_{ts}, g_{up}, g_{ac}, g_{hs}\}$ \vad{you use this structure before defining it.}, atomic attack set $\AtomicAttacks=\{a_{ad}, a_{ic}, a_{sp}, a_p, a_{bf}, a_{ss}, a_{fue}\}$, and MTD set $\Defenses = \{d_{dk}, d_{cp}, d_{cc}, d_{dsr}\}$. The nodes in black are subgoals, the nodes in red are atomic attacks, and the nodes in blue with rounded corners are MTDs. The refinement of subgoals is indicated below theme. The refinement can be omitted when the subgoal has only one atomic attack or subgoal child. The dashed edges link a MTD $d$ with a node $n$ such that $d\in \defense{n}$.
    }
    \label{fig:example-adt}
\end{figure}
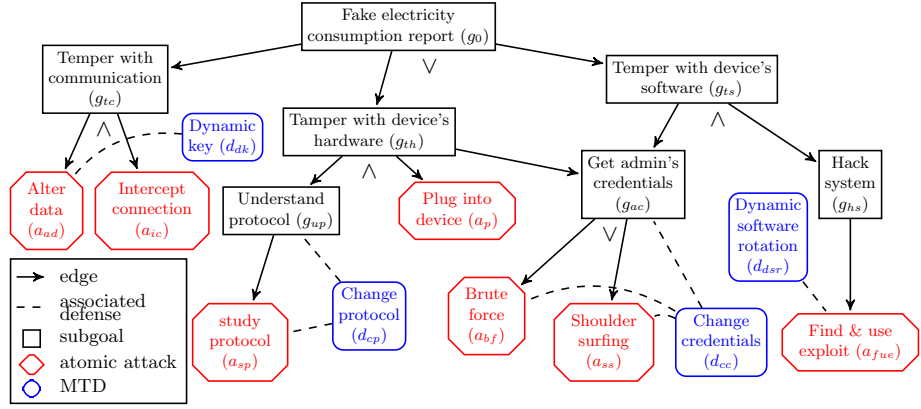We define the \gls{amg} that is a \gls{dag}-based structure with \glspl{mtd} and time, cost, and probabilistic attributes. \figurename~\ref{fig:example-adt} is an example of \gls{amg}.
\shrinkalt{}{\begin{table}[t]
    \centering
    \caption[Attributes of the \glsfmtshort{amg}.]{Attributes of the \gls{amg} $\AMG = \AMGDef$.}
    \begin{tabular}{|c|c|c|c|}
        \hline
        & Attribute & Domain & Definition \\
        \hline
        $a\in \AtomicAttacks$ & $t_a$ & $\mathbb{N}$ & completion time \\
        & $p_a$ & $[0,1]$ & success probability at completion \\
        & $\cost_a$ & $\mathbb{N}$ & activation cost \\
        & $\cost'_a$ & $\mathbb{N}$ & cost rate \\
        \hline
        $d\in \Defenses$ & $t_d$ & $\mathbb{N}$ & activation period \\
        & $p_d$ & $[0,1]$ & success probability at activation \\
        & $\defended{}{d}$ & $\Nodes$ & nodes defended by $d$\\
        \hline
        $g \in \Goals$ & $\operation{g}$ & $\{\land, \lor\}$ & refinement \\
        \hline
    \end{tabular}
    \label{tab:attributes}
\end{table}}
\begin{definition}[Attack Moving Target Defense \gls{dag}]
We define an \gls{amg} as a tuple $\AMG = \AMGDef$, s.t. $\Nodes$ is a set of nodes, the pair $\brac{\Nodes, \Edges}$ forms a rooted \gls{dag} with root $\AMGroot\in \Nodes$, and edges $\Edges \subseteq \Nodes \times \Nodes$. Given the \gls{amg}, the set $\AtomicAttacks$ refers to the leaves of $\brac{\Nodes, \Edges}$ and its elements are called \textit{atomic attacks}, and $\Goals$ refers to the inner-nodes of $\brac{\Nodes, \Edges}$ and its elements are called \textit{subgoals}.
The list $\operation{} = \{\operation{g}\}_{g\in \Goals}$ assigns a \textit{refinement} $\operation{g} \in \{\land, \lor\}$ for each subgoal $g\in \Goals$.
The set $\Defenses$ is the set of \textit{\glspl{mtd}}.
The root $\AMGroot$ is called the \textit{main goal} of the attack. The lists $\cost = (\cost_a)_{a\in \AtomicAttacks}$ and $\propcost = (\propcost_a)_{a\in \AtomicAttacks}$ give an \textit{activation cost} $\cost_a\in \mathbb{N}$ and a \textit{proportional cost} $\propcost_a \in \mathbb{N}$ (cost of each unit of time that the atomic attack $a$ is activated) to each atomic attack $a\in \AtomicAttacks$. The list $\timefunc = (\timefunc_n)_{n\in \AtomicAttacks \cup \Defenses}$ assigns a \textit{completion time} $\timefunc_a \in \mathbb N$ for each atomic attack $a\in \AtomicAttacks$ and an \textit{activation period} $\timefunc_d \in \mathbb N$ for each \gls{mtd} $d \in \Defenses$. The list $\proba = (\proba_n)_{n\in \AtomicAttacks \cup \Defenses}$  gives a \textit{success probability at completion} $\proba_a \in [0,1]$ for each atomic attack $a\in \AtomicAttacks$ and a \textit{success probability at activation} $\proba_d\in [0,1]$ for each \gls{mtd} $d \in \Defenses$.
Finally, the list $\defended{}{} = (\defended{}{d})_{d\in \Defenses}$ assigns the set $\defended{}{d} \subseteq \Nodes$ of nodes that a \gls{mtd} $d$ protects.
\shrinkalt{}{The attributes are summarized in \tablename~\ref{tab:attributes}.}
\end{definition}

For the rest of this paper, when a general \gls{amg} $\AMG$ is given, we assume $\AMG = \AMGDef$, and $\AtomicAttacks$ (resp. $\Goals$) is the set of leaves (resp. inner-nodes) of $\brac{\Nodes, \Edges}$. Given an \gls{amg} $\AMGgeneral$, we denote $\children{\AMG}{n}$ as the set of children of a node $n \in \Nodes$ or simply $\children{}{n}$ when $\AMG$ is evident in the context.
For a defense $n\in \Nodes$, we will use $\defense{n} = \{d\in \Defenses \mid n \in \defended{}{d}\}$ the set of \glspl{mtd} that defend $n$. For a rooted \gls{dag} $\brac{\Nodes, \Edges}$, we call a directed path a sequence of nodes $n_1, \dots, n_k \in \Nodes$ such that the edges $(n_1, n_2), \dots, (n_{k-1}, n_k)$ are in $\Edges$. As $\brac{\Nodes, \Edges}$ is rooted, a directed path exists from $\AMGroot$ to any nodes.

\mysubsection{Informal semantics} Every node has the state \textit{completed} or \textit{uncompleted}. In addition, every atomic attack and \gls{mtd} has the state \textit{activated} or \textit{deactivated}. %  every \gls{mtd} $d\in \Defenses$ can have the state \textit{triggered} $(d\in \TriggeredSet$) or \textit{idle} ($d\not\in \TriggeredSet$).
The \gls{amg} is meant to be interpreted in a timed environment. %as the value of $\ActivatedSet$, $\CompletedSet$ and $\TriggeredSet$ will vary in time.
Indeed, an atomic attack $a\in\AtomicAttacks$ has a \textit{completion clock} $x_a$ initialized when the attack gets activated. When its clock reaches the completion time ($x_a=\timefunc_a$), the attack can succeed (resp. fail) with probability $\proba_a$ (resp. $1-\proba_a$). If the attack succeeds, the atomic attack $a$ is completed. An \gls{mtd} $d \in \Defenses$ is periodically activated when the clock $x_d$ reaches $\timefunc_d$, and the defense can succeed (resp. fail) with probability $\proba_d$ (resp. $1-\proba_d$).
At any time, the system progresses with two sequential steps:
\begin{itemize}
    \item \emph{Evaluation of the defenses.} For each \gls{mtd}, say $d\in \Defenses$, that gets activated ($x_d = \timefunc_d$) and succeeds, we modify the system's state. Every defended node $n \in \defended{}{d}$, gets uncompleted. If $n$ is an atomic attack, it gets deactivated, and its completion clock gets back to $0$.
    \item \emph{Evaluation of the attack progression.} Starting from the deeper nodes (in a \textit{bottom-up} fashion), every subgoal $g \in \Goals$ is completed if its children's conjunction (if $\operation{g}$ is $\land$) or disjunction (if $\operation{g}$ is $\lor$) is completed. We propagate the completion from the leaves to the root of the \gls{dag}.
\end{itemize}
In addition, two asynchronous events can be triggered at any time:
\begin{itemize}
    \item \emph{Activation of an atomic attack.} The attacker can activate atomic attacks that are not activated yet. Their completion clocks are initialized to $0$.
    \item \emph{Completion of an atomic attack.} Every activated atomic attack, say $a$, such that ($x_a = \timefunc_a$) gets deactivated. The atomic attack is completed with probability $\proba_a$, or stays uncompleted with probability $1-\proba_a$.
\end{itemize}

If these asynchronous events happen simultaneously between them, or/and with a sequential step, the precedence is given with uniform probability.
Notice the difference between $d_{dsr}$ and $d_{cp}$ in \figurename~\ref{fig:example-adt}: when $a_{fue}$ is completed once, its parent $g_{hs}$ gets completed forever, while the parent $g_{up}$ of $a_{sp}$ is defended by $d_{cp}$.
Our model assumes that atomic attack probabilities of success are mutually independent and that several activations of the same atomic attack succeed with an independent and identically distributed probability. Moreover, as opposed to~\cite{kumar2015quantitative,GHL+16}, the attacker can activate as many different atomic attacks as he wants at the same time. Nevertheless, the attacker is memoryless: he knows only the current system state. As a result, he cannot count how many times an atomic attack was activated or the previously completed atomic attack sequence. }{\subsection{Model}\label{sec:model}}
%\subsection{Semantic}\label{sec:semantic}
%\input{semantic}
\shrinkalt{}{\subsection{Expressivity}
We allow a node to have several \glspl{mtd} and an \gls{mtd} to defend several nodes because we believe that is happening in real life. Moreover, our model lets us control where the attack has to be restarted when a defense succeeds: if a subgoal $g$ is the conjunction $a_1\land a_2$, we can express some subtle behavior of an \gls{mtd} $d$. For instance, $d$ can turn $g$, $a_1$, and $a_2$ incomplete, and deactivate $a_1$ and $a_2$ (\figurename~\ref{fig:keep-0}). It can also turn only $g$, and $a_1$ incomplete, and deactivate only $a_1$ (\figurename~\ref{fig:keep-a1}).
It can also turn incomplete and deactivate $a_1$ and $a_2$ but keep $g$ completed if it was completed once (\figurename~\ref{fig:keep-g}).

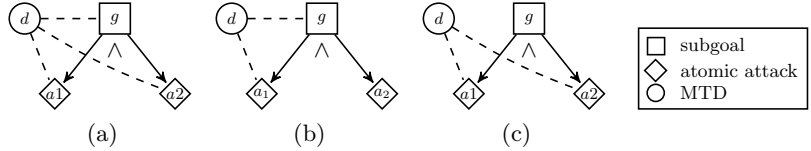
\begin{figure}[t]
\centering
    \subfloat[]{
    \begin{tikzpicture}[
            ->,
            >=stealth',
            shorten >=1pt,
            auto,
            semithick,
            xscale=0.4,
            yscale=0.5
        ]
        
        \node[goal, label=below:{$\land$}] at (0, 0) (g) {$g$};
        \node[attack, label={}] at (-2, -2) (a1) {$a1$};
        \node[attack, label={}] at (2, -2) (a2) {$a2$};
        \node[defense, label={}] at (-3, 0) (d) {$d$};
        
        \path (g) edge (a1) edge (a2);
        \path (d) edge [dashed, -] (g) edge [dashed, -] (a1) edge [dashed, -, bend right=5] (a2);
    \end{tikzpicture}
    \label{fig:keep-0}
    }
    \subfloat[]{
    \begin{tikzpicture}[
            ->,
            >=stealth',
            shorten >=1pt,
            auto,
            semithick,
            xscale=0.4,
            yscale=0.5
        ]
        
        \node[goal, label=below:{$\land$}] at (10, 1) (g) {$g$};
        \node[attack, label={}] at (8, -1) (a1) {$a_1$};
        \node[attack, label={}] at (12, -1) (a2) {$a_2$};
        \node[defense, label={}] at (7, 1) (d) {$d$};
        
        \path (g) edge (a2) edge (a1);
        \path (d) edge [dashed, -] (g) edge [dashed, -] (a1);
    \end{tikzpicture}
    \label{fig:keep-a1}
    }
    \subfloat[]{
    \begin{tikzpicture}[
            ->,
            >=stealth',
            shorten >=1pt,
            auto,
            semithick,
            xscale=0.4,
            yscale=0.5
        ]
        
        \node[goal, label=below:{$\land$}] at (0, 0) (g) {$g$};
        \node[attack, label={}] at (-2, -2) (a1) {$a1$};
        \node[attack, label={}] at (2, -2) (a2) {$a2$};
        \node[defense, label={}] at (-3, 0) (d) {$d$};
        
        \path (g) edge (a1) edge (a2);
        \path (d) edge [dashed, -] (a1) edge [dashed, -, bend right=5] (a2);
    \end{tikzpicture}
    \label{fig:keep-g}
    }
    \captionsetup[subfigure]{labelformat=empty}
    \subfloat[]{
    \begin{tikzpicture}[
            ->,
            >=stealth',
            shorten >=1pt,
            auto,
            semithick,
            xscale=0.4,
            yscale=0.5
        ]
        \matrix [draw, below left, inner sep=2pt] at (0, 0) {
        \node [goal, minimum size=1em, label={[scale=0.8, label distance=5pt]right:{subgoal}}] at (0.05, 0) {}; \\
        \node [attack, minimum size=1.3em, label={[scale=0.8, label distance=4pt]right:{atomic attack}}] {}; \\
        \node [defense, minimum size=1em, label={[scale=0.8, label distance=5pt]right:{\glsfmtshort{mtd}}}] at (0.01, 0) {}; \\
        };

    \end{tikzpicture}
    \label{fig:legend}
    }
    
    \cprotect\caption[Three examples of the expressivity of \glsfmtshort{amg}.]{Three examples of the expressivity of \gls{amg}. In~\subref{fig:keep-0}, at each successful activation of $d$, the subgoal $g$, and the atomic attacks $a_1$ and $a_2$ are uncompleted (if they were completed). Moreover, $a_1$ and $a_2$ are deactivated (if they were activated), and their completion clocks, say $x_{a_1}$ and $x_{a_2}$, are set back to $0$. If we remove $a_2$ from $\defended{\AMG}{d}$ (case~\subref{fig:keep-a1}), the completion status, the activation status, and the completion clock of $a_2$ are not affected by $d$. If we remove $g$ from $\defended{\AMG}{g}$ (case~\subref{fig:keep-g}), then $g$ acts as a backup point (later called a \textit{checkpoint}) that does not get uncompleted if it is completed once.}
    \label{fig:expressivity}
\end{figure}}

\section{Construction of the \glsfmtshort{ptmdp} for \glsfmtshort{amg}}\label{sec:construction-PTMDP}
\mysubsection{Computing attack time, cost, and success probability}
We want to build a \gls{ptmdp} that represents the \gls{amg} because we can exploit this structure to find some near-optimal strategies for specific objectives. Ideally, we have a $2$\textonehalf-player game with the defender, the attacker, and the stochastic environment (counting for~\textonehalf). The defender plays first by choosing the defense periods, and the attacker plays the rest of the game, trying to reach the root of the \gls{amg}. Nevertheless, we simplify the problem by assuming the defender has already chosen a list of defense periods $(t_d)_{d_\in \Defenses}$. The resulting game is a $1$\textonehalf-player game where the attacker plays against the environment.
This simplification has two reasons. First, the attacker effectively plays only against the environment. After all, the defender plays first and only once. Second, the choice of \gls{mtd} periods is not countable, so it is hard to express it in a finite game structure.
We can then compute, in the \gls{ptmdp}, the reachability of the main goal under time or cost constraints and compute strategies for the attacker with minimal expected time or expected cost. With this information, we can evaluate how good is a set of activation periods for the different \glspl{mtd} of the \gls{amg}.
\mysubsection{Representation of the system state}
The \textit{system state} is given by the set of activated atomic attacks, the set of completed nodes, and the completion clocks of the atomic attacks and \glspl{mtd}. Notice that an atomic attack can be activated and completed simultaneously, and the clocks will be in the \gls{ptmdp} clock set. Let $\Omega = 2^\AtomicAttacks \times 2^{\Nodes}$ be the set of possible states.

The space $\Omega$ contains some useless states. When a node without \gls{mtd} is completed, it stays completed for the rest of the analysis. Thus, its descendant node completion and activation status can be unnecessary. Given a set of completed nodes $C$, we call a \textit{checkpoint} a completed node $n$ without defense i.e. $n \in C$ s.t. $\defense{n} = \emptyset$. We will define an equivalence relation $\sim$ over $\Omega$ to reduce the state space. Intuitively, we want two equivalent states for $\sim$ to be naturally equivalent for the attacker in terms of future costs, time, success probability, and possible actions for an optimal attacker (implying that he will not use unnecessary costs to start or continue an atomic attack that leads only to checkpoints).

We introduce the \textit{propagation operator} that computes the set of effectively completed nodes given an initial set of completed nodes (\cf \figurename~\ref{fig:state}). It is defined through a fixed point of a function adding the subgoals $g$ that have all its children completed (if $\operation{g}$ is $\land$) or at least one child completed (if $\operation{g}$ is $\lor$).

\begin{definition}[Propagation operator]
\label{def:propagate}
Let $\AMGgeneral$ be an \gls{amg}.
We define the \textit{propagation operator} $\propagateDefADT{\AMG}: 2^{\Nodes} \mapsto 2^{\Nodes}$ as follows. For $C\subseteq \Nodes$, $\propagateADT{\AMG}{C}$ is the least fixed point greater or equal to $C$ (for $\subseteq$) of $\fixedPoint{}: 2^\Nodes \mapsto 2^\Nodes$ where
%$\fixedPoint{}(N) = N \cup \{g\in \Goals \mid \text{For } \{n_1, \dots, n_j\} = \children{}{g}, ((n_1 \in N)\operation{g} \dots \operation{g} (n_j\in N))\}$
%\begin{align*}
%    \fixedPoint{}(N) = N \cup \{&g\in \Goals \mid \forall n_1, \dots, n_j \in \Nodes,\\
%    &(\{n_1, \dots, n_j\} = \children{}{g}) \Rightarrow ((n_1 \in N)\operation{g} \dots \operation{g} (n_j\in N))\}
%\end{align*}
\begin{align*}
    \fixedPoint{}(N) = N \cup \left\{g\in \Goals \mid 
    (\exists n\in \children{}{g}, n\in N) \operation{g}
    (\forall n\in \children{}{g}, n\in N\right)\}
\end{align*}
%\begin{align*}
%    \fixedPoint{}(N) = N \cup \{g\in \Goals \mid \forall \{n_1, \dots, n_j\} = \children{}{g}, ((n_1 \in N)\operation{g} \dots \operation{g} (n_j\in N))\}
%\end{align*}
%\begin{align*}
%    \fixedPoint{}(N) = N \cup \left\{g\in \Goals \Bigm \vert \begin{cases}
%    \exists n\in \children{}{g}, n\in N & \text{ if } \operation{g} = \lor\\
%    \forall n\in \children{}{g}, n\in N & \text{ if } \operation{g} = \land
%    \end{cases}\right\}
%\end{align*}
\shrinkalt{We extend the definition of $\propagateDefADT{\AMG}$ on $(A,C) \in \Omega$ by propagating only the completed nodes: $\propagateADT{\AMG}{A,C} = (A, \propagateADT{\AMG}{C})$.}{We extend the definition of $\propagateDefADT{\AMG}$ on $\Omega = 2^\AtomicAttacks \times 2^{\Nodes}$ with $\propagateDefADT{\AMG}: \Omega \mapsto \Omega$, such that for $(A, C) \in \Omega$, $\propagateADT{\AMG}{A,C} = (A, \propagateADT{\AMG}{C})$, that is, we propagate only the completed nodes.}
\end{definition}
As $\fixedPoint{}$ is an increasing function, the fixed point is well defined and is the composition $\fixedPoint{}^k(C)$ where $k \in \mathbb N$ verifies $\fixedPoint{}^k(C) = \fixedPoint{}^{k+1}(C)$. We will omit the \gls{amg} and simply write $\propagate{\AMG}{C}$ or $\propagate{\AMG}{A, C}$ when $\AMG$ is clear in the context. %\shrinkalt{We extend the definition of $\propagateDefADT{\AMG}$ on $(A,C) \in \Omega$ by propagating only the completed nodes: $\propagateADT{\AMG}{A,C} = (A, \propagateADT{\AMG}{C})$.}{}

\shrinkalt{}{
We show three intuitive properties: the propagation operator  $\propagateDef{\AMG}$ contains the initially completed nodes, increases for $\subseteq$, and is a projection.
\begin{restatable}{proposition}{propagateprop}\label{prop:propagate}
For an \gls{amg} $\AMG$, and $C,B\subseteq \Nodes$, we have
\begin{inparaenum}[(i)]
    \item \label{item:propagate-contains}$C \subseteq \propagate{\AMG}{C}$, and
    \item \label{item:propagate-increase}$\propagateDef{}$ is increasing for $\subseteq$
    \item \label{item:propagate-projection}$\propagateDef{}$ is a projection, i.e., the composition $\propagateDef{} \circ \propagateDef{} = \propagateDef{}$.
\end{inparaenum}
\end{restatable}
\ifdefined\ALLPROOFS\begin{proof}
For~(\ref{item:propagate-contains}), we need to notice that the $\fixedPoint{}$ from Definition~\ref{def:propagate} always contains its argument, and we use the fact that $\propagate{\AMG}{C} = \fixedPoint{}^k(C)$ for all $C \subseteq \Nodes$ and some $k\in \mathbb N$.

For~(\ref{item:propagate-increase}), we use the fact that $\fixedPoint{}$ is increasing and reaches a fixed point. Let $C_1\subseteq C_2 \subseteq \Nodes$. We can take composistion indices $k_1, k_2\in \mathbb N$ such that $\propagate{\AMG}{C_1} = \fixedPoint{}^{k_1}(C_1)$ and $\propagate{\AMG}{C_2} = \fixedPoint{}^{k_2}(C_2)$. Now,
\begin{align*}
    \propagate{\AMG}{C_1} &= \fixedPoint{}^{k_1}(C_1)\\
    &=\fixedPoint{}^{\max(k_1, k_2)}(C_1)\\
    &\subseteq \fixedPoint{}^{\max(k_1, k_2)}(C_2)\\
    &= \propagate{\AMG}{C_2}
\end{align*}

For~(\ref{item:propagate-projection}), by the Definition~\ref{def:propagate}, $\propagateDef{}$ is a fixed point.
\end{proof}\else\begin{sketchproof}
The function $\fixedPoint{}$ from Definition~\ref{def:propagate} always contains its argument, is increasing, and $\propagateDef{}$ is a fixed point.
\end{sketchproof}\fi
}

Given a set of completed nodes $C$, we call \textit{completed descendants} the set of nodes that have a completed node within all sequences of nodes forming a directed path from $\AMGroot$ (\cf \figurename~\ref{fig:state}).

\begin{definition}[Completed descendants]
\label{def:comp-desc}
Let $\AMGgeneral$ be an \gls{amg}, and $C\subseteq \Nodes$ a set of completed nodes. We define the completed descendants of $C$ as
\begin{align*}
    \subTreeAMG{\AMG}{C} = \{&n \in \Nodes \mid \forall k \in \mathbb N, \forall g_{1}, \dots, g_{k} \in \Goals,
    (g_{1} = \AMGroot \land \forall j\in \{1,\dots, k-1\},\\
    &(g_{j}, g_{{j+1}}) \in \Edges
    \land (g_{k}, n) \in \Edges) \Rightarrow \exists j\in \{1,\dots, k\}, g_{j} \in C\}
\end{align*}
\end{definition}
When $\AMG$ is evident in the context, we will write $\subTree{\AMG}{C}$.

\shrinkalt{}{
We show that the completed descendants are increasing for $\subseteq$ and that the completed descendants of $C$ are exactly the completed descendants of $C\setminus\subTree{\AMG}{C}$.
\begin{restatable}{proposition}{subtreeprop}\label{prop:subtree-inc}
For an \gls{amg} $\AMGgeneral$, $\subTreeDef{\AMG}$ is increasing.
\end{restatable}
\ifdefined\ALLPROOFS\begin{proof}
Let $A\subseteq B \subseteq \Nodes$ if there is $j \in \{1, \dots, k\}$ s.t. $g_{j} \in A$, then $g_{j} \in B$. This proves $\subTree{\AMG}{A} \subseteq \subTree{\AMG}{B}$.
\end{proof}\else\begin{sketchproof}
Follows from Definition~\ref{def:comp-desc}.
\end{sketchproof}\fi

\ifdefined\ALLPROOFS
We introduce Lemma~\ref{lemma:subtree-subtree}, showing that some nodes can be ignored in the completed descendants.
\begin{restatable}{lemma}{subtreelemma}\label{lemma:subtree-subtree}
For an AMG $\AMGgeneral$, and $C\subseteq \Nodes$, we have $\subTree{\AMG}{C\setminus \subTree{\AMG}{C}} = \subTree{\AMG}{C}$.
\end{restatable}
\begin{proof}
By Proposition~\ref{prop:subtree-inc} we have $\subTree{\AMG}{C\setminus \subTree{\AMG}{C}} \subseteq \subTree{\AMG}{C}$, now let $n \in \subTree{\AMG}{C}$ we want to prove $n \in \subTree{\AMG}{C\setminus \subTree{\AMG}{C}}$. By the Definition~\ref{def:comp-desc}, for any directed path $g_{1}, \dots, g_{k} \in \Nodes$ from $\AMGroot$ to a parent of $n$ there is a $j \in \{1,\dots, k\}$ s.t. $g_{j} \in C$. For any such directed path, let $j$ be the smallest integer such that $g_{j} \in C$. If $g_{j} \in \subTree{\AMG}{C}$, then $j > 1$ as $\AMGroot \not\in \subTree{\AMG}{C}$ ($\AMGroot$ never has completed parents, it does not have parents) and the directed path $g_{1}, \dots, g_{j-1}$ exists and must contain a node in $C$, say $g_{l}$ with $l\in \{1,\dots, j-1\}$. This is a contradiction because $l < j$ and $g_{l} \in C$ while $j$ was chosen the smallest. So $g_{j} \not\in \subTree{\AMG}{C}$, and thus $g_{j} \in C\setminus \subTree{\AMG}{C}$. As a result, every directed path from $\AMGroot$ to a parent of $n$ contains an element in $C\setminus \subTree{\AMG}{C}$, i.e. $n\in \subTree{\AMG}{C\setminus \subTree{\AMG}{C}}$.
\end{proof}
\fi % ALLPROOFS
}

Let $\noDef = \{n \in \Nodes \mid \defense{n} = \emptyset\}$. We notice that the set of checkpoints in $C$ is $C\cap \noDef$. We define the \textit{pruning operator} that eliminates unnecessary nodes from the system state, considering that completed descendants of checkpoints can be removed from completed and activated nodes and that completed nodes can be removed from activated nodes (\cf \figurename~\ref{fig:state}).

\begin{definition}[Pruning operator]
Let $\AMGgeneral$ an \gls{amg} and $\Omega = 2^\AtomicAttacks \times 2^{\Nodes}$. We define the pruning operator $\pruneDefADT{\AMG} : \Omega \mapsto \Omega$ s.t. for $(A, C) \in \Omega$,
\shrinkalt{$\pruneADT{\AMG}{A, C} = (A \setminus (\subTree{\AMG}{C \cap \noDef}\cup C), C \setminus \subTree{\AMG}{C \cap \noDef})$.}{
\begin{equation*}
    \pruneADT{\AMG}{A, C} = (A \setminus (\subTree{\AMG}{C \cap \noDef}\cup C), C \setminus \subTree{\AMG}{C \cap \noDef})
\end{equation*}}
\end{definition}

\begin{figure}[t]
    \subfloat[Initial state $(A,C)$.]{
    \begin{tikzpicture}[
            ->,
            >=stealth',
            shorten >=1pt,
            auto,
            semithick,
            xscale=0.35,
            yscale=0.45
        ]
        
        \node[goal, label=above:{}] at (0, 0) (g0) {$\land$};
        \node[goal, label=below:{}, fill=cyan] at (-7, -2) (g5) {$\lor$};
        \node[attack, label={}, fill=cyan] at (-7, -4) (a7) {};
        \node[goal, label=below:{}] at (-4, -1) (g1) {$\lor$};
        \node[defense, label={}] at (-7, 0) (d1) {};
        \node[goal, label=below:{}] at (0, -2) (g2) {$\land$};
        \node[goal, label=290:{}] at (4, -1) (g3) {$\lor$};
        \node[attack, label={}, fill=cyan] at (-5, -3) (a0) {};
        \node[attack, label={}, fill=yellow] at (-3, -3) (a1) {};
        \node[attack, label={}, fill=green] at (-1, -4) (a2) {};
        \node[attack, label={}] at (1, -4) (a3) {};
        \node[attack, label={}, fill=yellow] at (3, -3) (a4) {};
        \node[attack, label={}, fill=yellow] at (5, -3) (a5) {};
        \node[goal, label=below:{}] at (7, -2) (g4) {$\land$};
        \node[defense, label={}] at (7, 0) (d2) {};
        \node[attack, label={}, fill=cyan] at (7, -4) (a6) {};
        
        \path (g0) edge (g1) edge (g2) edge (g3);
        \path (g1) edge (g5) edge (a0) edge (a1);
        \path (g2) edge (a2) edge (a3) edge (a4);
        \path (g3) edge (a4) edge (a5) edge (g4);
        \path (g4) edge (a6);
        \path (g5) edge (a7);
        \path (d1) edge [-, dashed, bend left=10] (a0) edge [-, dashed] (a1) edge [-, dashed] (g1);
        \path (d2) edge [-, dashed] (g4) edge [-, dashed, bend right=50] (a6);
    \end{tikzpicture}
    \label{fig:init}
    }\hspace{0.82cm}
    \subfloat[State $\propagate{\AMG}{A,C}$.]{
    \begin{tikzpicture}[
            ->,
            >=stealth',
            shorten >=1pt,
            auto,
            semithick,
            xscale=0.35,
            yscale=0.45
        ]
        
        \node[goal, label=above:{}] at (0, 0) (g0) {$\land$};
        \node[goal, label=below:{}, fill=cyan] at (-7, -2) (g5) {$\lor$};
        \node[attack, label={}, fill=cyan] at (-7, -4) (a7) {};
        \node[goal, label=below:{}, fill=cyan] at (-4, -1) (g1) {$\lor$};
        \node[defense, label={}] at (-7, 0) (d1) {};
        \node[goal, label=below:{}] at (0, -2) (g2) {$\land$};
        \node[goal, label=290:{}, fill=cyan] at (4, -1) (g3) {$\lor$};
        \node[attack, label={}, fill=cyan] at (-5, -3) (a0) {};
        \node[attack, label={}, fill=yellow] at (-3, -3) (a1) {};
        \node[attack, label={}, fill=green] at (-1, -4) (a2) {};
        \node[attack, label={}] at (1, -4) (a3) {};
        \node[attack, label={}, fill=yellow] at (3, -3) (a4) {};
        \node[attack, label={}, fill=yellow] at (5, -3) (a5) {};
        \node[goal, label=below:{}, fill=cyan] at (7, -2) (g4) {$\land$};
        \node[defense, label={}] at (7, 0) (d2) {};
        \node[attack, label={}, fill=cyan] at (7, -4) (a6) {};
        
        \path (g0) edge (g1) edge (g2) edge (g3);
        \path (g1) edge (g5) edge (a0) edge (a1);
        \path (g2) edge (a2) edge (a3) edge (a4);
        \path (g3) edge (a4) edge (a5) edge (g4);
        \path (g4) edge (a6);
        \path (g5) edge (a7);
        \path (d1) edge [-, dashed, bend left=10] (a0) edge [-, dashed] (a1) edge [-, dashed] (g1);
        \path (d2) edge [-, dashed] (g4) edge [-, dashed, bend right=50] (a6);
    \end{tikzpicture}
    \label{fig:propagate}
    }\\[-0.5em]
    \subfloat[State $\simpleState{A,C} = \pruneDef{\AMG}\circ \propagate{\AMG}{A,C}$.]{
    \begin{tikzpicture}[
            ->,
            >=stealth',
            shorten >=1pt,
            auto,
            semithick,
            xscale=0.35,
            yscale=0.45
        ]
        
        \node[goal, label=above:{}] at (0, 0) (g0) {$\land$};
        \node[goal, label=below:{}, line width=2pt, fill=cyan] at (-7, -2) (g5) {$\lor$};
        \node[attack, label={}] at (-7, -4) (a7) {};
        \node[goal, label=below:{}, fill=cyan] at (-4, -1) (g1) {$\lor$};
        \node[defense, label={}] at (-7, 0) (d1) {};
        \node[goal, label=below:{}] at (0, -2) (g2) {$\land$};
        \node[goal, label=290:{}, line width=2pt, fill=cyan] at (4, -1) (g3) {$\lor$};
        \node[attack, label={}, fill=cyan] at (-5, -3) (a0) {};
        \node[attack, label={}, fill=yellow] at (-3, -3) (a1) {};
        \node[attack, label={}, line width=2pt, fill=cyan] at (-1, -4) (a2) {};
        \node[attack, label={}] at (1, -4) (a3) {};
        \node[attack, label={}, fill=yellow] at (3, -3) (a4) {};
        \node[attack, label={}] at (5, -3) (a5) {};
        \node[goal, label=below:{}] at (7, -2) (g4) {$\land$};
        \node[defense, label={}] at (7, 0) (d2) {};
        \node[attack, label={}] at (7, -4) (a6) {};
        
        \path (g0) edge (g1) edge (g2) edge (g3);
        \path (g1) edge (g5) edge (a0) edge (a1);
        \path (g2) edge (a2) edge (a3) edge (a4);
        \path (g3) edge (a4) edge (a5) edge (g4);
        \path (g4) edge (a6);
        \path (g5) edge (a7);
        \path (d1) edge [-, dashed, bend left=10] (a0) edge [-, dashed] (a1) edge [-, dashed] (g1);
        \path (d2) edge [-, dashed] (g4) edge [-, dashed, bend right=50] (a6);
        
        \node[goal, dashed, rounded corners, minimum size=2.3em, minimum height=3.2em] at (-7,-3.8) (zeta) {};
        \node[goal, dashed, rounded corners, minimum size=2.3em, minimum height=3.2em] at (5,-2.8) (zeta2) {};
        \node[goal, dashed, rounded corners, minimum size=2.8em, minimum height=6.7em] at (7,-3) (zeta3) {};
    \end{tikzpicture}
    \label{fig:reduce}
    }\hspace{0.2em}
    \subfloat{
    \begin{tikzpicture}[
            ->,
            >=stealth',
            shorten >=1pt,
            auto,
            semithick,
            xscale=0.4,
            yscale=0.5
        ]
        
        \matrix [draw, below left, inner sep=2pt, column sep=1pt, row sep=1pt, ampersand replacement=\&, column 1/.style={anchor=base}, column 2/.style={anchor=base west}, column 3/.style={anchor=base}, column 4/.style={anchor=base west}, column 5/.style={anchor=base}, column 6/.style={anchor=base west}] {
        \draw[line width=0.5pt, >=stealth', semithick](-0.2, 0) -- (0.2, 0); \& \node[right, scale=0.8, align=left] {edge};
        \& \draw[dashed,line width=0.5pt, semithick, -](-0.2, 0) -- (0.2, 0); \& \node[right, scale=0.8, align=left] {associated\\[-0.5em]defense};
        \& \node [goal, draw=none, minimum size=1em] {$\land\lor$}; \& \node[scale=0.8]{refinement};\\
        \node [goal, minimum size=1em, scale=1] {\phantom{a}}; \& \node[scale=0.8]{subgoal}; \&
        \node [attack, minimum size=1.3em, scale=0.8] {\phantom{a}}; \& \node[scale=0.8]{atomic attack}; \&
        \node [defense, minimum size=1em, scale=1] {\phantom{a}}; \& \node[scale=0.8]{\glsfmtshort{mtd}}; \\
        \node [goal, fill=yellow, minimum size=1em, minimum width=1.3em, draw=none] {\phantom{a}}; \& \node[scale=0.8]{activated};
        \& \node [goal, fill=cyan, minimum size=1em, minimum width=1.3em, draw=none]{\phantom{a}}; \& \node[scale=0.8]{completed};
        \& \node [goal, fill=green, minimum size=1em, minimum width=1.3em, draw=none,]{\phantom{a}};  \& \node[right, scale=0.8, align=left]{activated\\[-0.5em]and \\[-0.5em]completed};\\
        \node [goal, line width=2pt, minimum size=1em, minimum width=1.3em] {\phantom{a}}; \& \node[scale=0.8]{checkpoint};
        \& \node [goal, dashed, minimum size=1em, minimum width=1.3em] {\phantom{a}}; \& \node[right, scale=0.8, align=left] at (0, 0.2em) {completed\\[-0.5em]descendants}; \\
        };
    \end{tikzpicture}
    }
    \cprotect\caption[Illustration of the simple state in an \glsfmtshort{amg}.]{Illustration of the simple state in an \gls{amg}.
    \figurename~\subref{fig:init} is the initial state $(A,C)$, where $A$ contains the activated atomic attacks in yellow/green, and $C$ contains the completed nodes in cyan/green. \figurename~\subref{fig:propagate} is the propagation where $C$ becomes $\propagate{\AMG}{C}$. \figurename~\subref{fig:reduce} is the pruning with the operator $\pruneDef{\AMG}$, resulting in $\simpleState{A,C}$.
    \shrinkalt{}{The completed descendants of checkpoints after propagation (nodes in $\subTree{\AMG}{\propagate{\AMG}{C}\cap\noDef}$) are removed from $A$ and $C$. The resulting state, in~\subref{fig:reduce}, is $\simpleState{A,C}$.}}
    \label{fig:state}
\end{figure}
We define the \textit{simple state} as the composition of propagation and pruning. 
\begin{definition}[Simple state]
Let $\AMGgeneral$ be an \gls{amg} and $\Omega = 2^\AtomicAttacks \times 2^{\Nodes}$.
For $(A,C)\in \Omega$, we define $\simpleStateADT{\AMG}{A,C} = \pruneDefADT{\AMG}\circ \propagateADT{\AMG}{A, C}$ the simple state of $(A,C)$.
\end{definition}
We will simply write $\pruneDef{\AMG}$ and $\simpleState{\cdot}$ when the \gls{amg} $\AMG$ is clear in the context.
For $(A,C)\in \Omega$, the simple state $\simpleState{A, C}$ contains the minimal information needed to describe the attack state. Indeed, the nodes that are descendants of checkpoints on every path from the main goal will not help achieve it, so they are not present in $\simpleState{A,C}$. Moreover, the activated atomic attacks already completed are also useless, so they are removed. \figurename~\ref{fig:state} illustrates how we get $\simpleState{A, C}$ from $(A,C)$.

\shrinkalt{}{
\ifdefined\ALLPROOFS
In Proposition~\ref{prop:simple-state-proj}, we will show that the simple state is a projection. First, we need Lemma~\ref{lemma:subtree-simple-state}.
\begin{lemma}\label{lemma:subtree-simple-state}
Let $\AMG$ be an AMG, $A_1 \in \AtomicAttacks$ and $C_1\in \Nodes$. Let $(A_2, C_2) = \simpleState{A_1, C_1}$. We have,
\begin{equation*}
    \subTree{\AMG}{\propagate{\AMG}{C_1}\cap \noDef} = \subTree{\AMG}{\propagate{\AMG}{C_2}\cap \noDef}
\end{equation*}
\end{lemma}
\begin{proof}
By definition, $C_2 = \propagate{\AMG}{C_1}\setminus \subTree{\AMG}{\propagate{\AMG}{C_1}\cap \noDef}$ and by Propositions~\ref{prop:propagate}(\ref{item:propagate-increase},~\ref{item:propagate-projection}) and ~\ref{prop:subtree-inc}, we have the first inclusion:
\begin{gather*}
    \propagate{\AMG}{C_2} = \propagate{\AMG}{\propagate{\AMG}{C_1}\setminus \subTree{\AMG}{\propagate{\AMG}{C_1}\cap \noDef}} \subseteq \propagateDef{\AMG} \circ \propagate{\AMG}{C_1} = \propagate{\AMG}{C_1}\\
    \subTree{\AMG}{\propagate{\AMG}{C_2} \cap \noDef} \subseteq \subTree{\AMG}{\propagate{\AMG}{C_1} \cap \noDef}
\end{gather*}

Moreover, as $C_2 = \propagate{\AMG}{C_1}\setminus \subTree{\AMG}{\propagate{\AMG}{C_1}\cap \noDef}$ and by Propositions~\ref{prop:propagate}(\ref{item:propagate-contains}) and~\ref{prop:subtree-inc}, and Lemma~\ref{lemma:subtree-subtree}, we have the second inclusion:
\begin{gather*}
    \propagate{\AMG}{C_1}\setminus \subTree{\AMG}{\propagate{\AMG}{C_1}\cap \noDef} = C_2 \subseteq \propagate{\AMG}{C_2}\\
    \subTree{\AMG}{\propagate{\AMG}{C_1}\cap \noDef} = \subTree{\AMG}{\propagate{\AMG}{C_1}\cap \noDef\setminus \subTree{\AMG}{\propagate{\AMG}{C_1}\cap \noDef}}
    \subseteq \subTree{\AMG}{\propagate{\AMG}{C2}\cap \noDef}
\end{gather*}
\end{proof}
\fi % ALLPROOFS

We show that a simple state is its own simple state.
\begin{restatable}{proposition}{simplestateprop}
\label{prop:simple-state-proj}
For an \gls{amg} $\AMGgeneral$, $A \subseteq \AtomicAttacks$, and $C \subseteq \Nodes$, we have $\simpleState{\simpleState{A,C}} = \simpleState{A,C}$, i.e., $\simpleState{\cdot}$ is a projection.
\end{restatable}
\ifdefined\ALLPROOFS\begin{proof}
Let $(A_1,C_1) = \simpleState{A,C}$ and $(A_2,C_2) = \simpleState{\simpleState{A,C}}$.

First we want to prove $C_1 = C_2$.
By definition,
\begin{align*}
    C_1 &= \propagate{\AMG}{C}\setminus \subTree{\AMG}{\propagate{\AMG}{C}\cap \noDef}\\
    C_2 &= \propagate{\AMG}{C_1}\setminus \subTree{\AMG}{\propagate{\AMG}{C_1}\cap \noDef}
\end{align*}
Moreover, by Lemma~\ref{lemma:subtree-simple-state}, we have $\subTree{\AMG}{\propagate{\AMG}{C}\cap \noDef} = \subTree{\AMG}{\propagate{\AMG}{C_1}\cap \noDef}$ so we just need to prove $C_1 \subseteq \propagate{\AMG}{C_1}$, that is immediate by Proposition~\ref{prop:propagate}(\ref{item:propagate-contains}), and $C_2\subseteq \propagate{\AMG}{C}$. This last point is true by Proposition~\ref{prop:propagate}(\ref{item:propagate-increase},~\ref{item:propagate-projection}):
\begin{equation*}
    C_2 \subseteq \propagate{\AMG}{C_1} = \propagate{\AMG}{\propagate{\AMG}{C}\setminus \subTree{\AMG}{\propagate{\AMG}{C}\cap \noDef}} \subseteq \propagateDef{\AMG} \circ \propagate{\AMG}{C} = \propagate{\AMG}{C}
\end{equation*}
This finishes the proof of $C_1=C_2$.

Let us prove that $A_1 = A_2$. By definition,
\begin{align}
    \nonumber
    A_1 &= A \setminus (\subTree{\AMG}{\propagate{\AMG}{C}\cap \noDef} \cup \propagate{\AMG}{C})\\
    A_2 &= A_1 \setminus (\subTree{\AMG}{\propagate{\AMG}{C_1}\cap \noDef} \cup \propagate{\AMG}{C_1})
    \label{eq:def-A2}
\end{align}
%Starting from~\eqref{eq:boundC2-left} and knowing~\eqref{eq:subtree-equal} and~\eqref{eq:def-A1} we have
Knowing $\subTree{\AMG}{\propagate{\AMG}{C}\cap \noDef} = \subTree{\AMG}{\propagate{\AMG}{C_1}\cap \noDef}$ and $\propagate{\AMG}{C_1} \subseteq \propagate{\AMG}{C}$, we have
\begin{gather*}
    \subTree{\AMG}{\propagate{\AMG}{C_1}\cap \noDef} \cup \propagate{\AMG}{C_1} \subseteq \subTree{\AMG}{\propagate{\AMG}{C}\cap \noDef} \cup \propagate{\AMG}{C}\\
    \underbrace{A\setminus (\subTree{\AMG}{\propagate{\AMG}{C}\cap \noDef} \cup \propagate{\AMG}{C})}_{A_1} \cap (\subTree{\AMG}{\propagate{\AMG}{C_1}\cap \noDef} \cup \propagate{\AMG}{C_1}) = \emptyset
\end{gather*}
Now by eq.~\eqref{eq:def-A2} we have $A_1=A_2$.
\end{proof}\else\begin{sketchproof}
The intuition is that the nodes added by $\propagateDef{\AMG}$ are removed by $\pruneDef{\AMG}$.
\end{sketchproof}\fi

The following proposition means that defense activation do not affect the non-defended nodes completed descendants of a simple state.

\begin{restatable}{proposition}{subtreedefenseprop}\label{prop:subtree-defense}
Let $\AMGgeneral$ be an AMG. Let $A' \subseteq \AtomicAttacks$, $C' \subseteq \Nodes$, and $(A,C) = \simpleState{A', C'}$. Let $D \subseteq \Defenses$ be a set  of defenses. We have,
\begin{equation*}
    \subTree{\AMG}{\propagate{\AMG}{C}\cap \noDef} = \subTree{\AMG}{\propagate{\AMG}{C \setminus \bigcup_{d\in D} \defended{\AMG}{d}}\cap \noDef}
\end{equation*}
\end{restatable}
\ifdefined\ALLPROOFS\begin{proof}
The first inclusion is immediate by Propositions~\ref{prop:propagate}(\ref{item:propagate-increase}) and~\ref{prop:subtree-inc},
\begin{gather*}
    C \setminus \bigcup_{d\in D} \defended{\AMG}{d} \subseteq C\\
    \propagate{\AMG}{C \setminus \bigcup_{d\in D} \defended{\AMG}{d}} \subseteq \propagate{\AMG}{C}\\
    \subTree{\AMG}{\propagate{\AMG}{C \setminus \bigcup_{d\in D} \defended{\AMG}{d}}\cap \noDef} \subseteq \subTree{\AMG}{\propagate{\AMG}{C}\cap \noDef}
\end{gather*}

Now let $n\in \subTree{\AMG}{\propagate{\AMG}{C}\cap \noDef}$. Let $n_1, \dots, n_k$ be a path from $n_1=\AMGroot$ to $n_k = n$. We can take the smallest integer $i \in \{1, \dots, k-1\}$ such that $n_i$ is a checkpoint in $\propagate{\AMG}{C}\cap \noDef$. We have,
\begin{equation}
    n_i \in \propagate{\AMG}{C} = \propagate{\AMG}{\propagate{\AMG}{C'} \setminus \subTree{\AMG}{\propagate{\AMG}{C'}\cap \noDef}}
    \subseteq \propagate{\AMG}{\propagate{\AMG}{C'}}
    = \propagate{\AMG}{C'}\label{eq:C'}
\end{equation}
As $i$ is chosen as the smallest integer of the set, we have $n_i\not \in \subTree{\AMG}{\propagate{\AMG}{C}\cap \noDef}$, otherwise, there would be another checkpoint earlier in the path $n_1, \dots, n_k$. In the proof of Proposition~\ref{prop:simple-state-proj} we proved that $\subTree{\AMG}{\propagate{\AMG}{C}\cap \noDef} = \subTree{\AMG}{\propagate{\AMG}{C'} \cap \noDef}$, so $n_i \not\in \subTree{\AMG}{\propagate{\AMG}{C'} \cap \noDef}$. Using this fact and eq.~\eqref{eq:C'} we have,
\begin{align*}
    n_i \in \propagate{\AMG}{C'} \setminus \subTree{\AMG}{\propagate{\AMG}{C'} \cap \noDef} = C
\end{align*}
Moreover, $n_i \in \propagate{\AMG}{C}\cap \noDef$ implies $n_i \not\in \noDef$ so $n_i \not\in \bigcup_{d\in D} \defended{\AMG}{d}$. Using the increase of $\propagateDef{\AMG}$,
\begin{gather*}
\left\{
    \begin{matrix}
        n_i \in C \setminus \bigcup_{d\in D} \defended{\AMG}{d}\\
        n_i \not \in \noDef
    \end{matrix}
    \right.\\
    n_i \in \propagate{\AMG}{C \setminus \bigcup_{d\in D} \defended{\AMG}{d}}\cap \noDef
\end{gather*}
As a result, whatever the path $n_1,\dots, n_k$ from the root of the AMG, $n$ has a checkpoint in $\propagate{\AMG}{C \setminus \bigcup_{d\in D} \defended{\AMG}{d}}\cap \noDef$. So $n \in \subTree{\AMG}{\propagate{\AMG}{C \setminus \bigcup_{d\in D} \defended{\AMG}{d}}\cap \noDef}$. Finally by double inclusion,
\begin{equation*}
    \subTree{\AMG}{\propagate{\AMG}{C}\cap \noDef} = \subTree{\AMG}{\propagate{\AMG}{C \setminus \bigcup_{d\in D} \defended{\AMG}{d}}\cap \noDef}
\end{equation*}
\end{proof}\else\begin{sketchproof}
The proof uses the fact that $C$ is the right member of a simple-state and that the first checkpoint in a path from $\AMGroot$ in $\propagate{\AMG}{C}\cap \noDef$ is also a checkpoint in $\propagate{\AMG}{C \setminus \bigcup_{d\in D} \defended{\AMG}{d}}\cap \noDef$.
\end{sketchproof}\fi
}

We are now able to define an equivalence relation on the states. Two states are equivalent if they have the same simple states.

\begin{definition}[Equivalent states]
Let $\AMGgeneral$ be an \gls{amg}, and $\Omega = 2^\AtomicAttacks \times 2^{\Nodes}$.
We say that two pairs $(A,C) \in \Omega$ and $(A', C') \in \Omega$ are equivalent, denoted $(A,C) \sim (A',C')$, if $\simpleState{A,C} = \simpleState{A',C'}$.
\end{definition}

We use two new notations on the quotient set $\Omega / \sim$ that let us access the left member $\Aof{\loc}$ and right member $\Cof{\loc}$ of the canonical representative $\simpleState{\loc} = (\Aof{\loc}, \Cof{\loc})$ of an element $\loc \in \Omega/\sim$. We also denote $[\cdot]$ the equivalence class of an element.\shrinkalt{As proved in~\cite{amg},}{By Proposition~\ref{prop:simple-state-proj},} we have $[\simpleState{\loc}] = \loc$, so $\simpleState{\loc}$ is indeed a representative of $\loc$. Moreover we overload the set difference by writing $\loc \setminus (a,b) = [\Aof{\loc}\setminus a, \Cof{\loc}\setminus b]$, and the set union by writing $\loc \cup (a,b) = [\Aof{\loc}\cup a, \Cof{\loc}\cup b]$.

\ifdefined\FROMPTG
Given an input \gls{amg} $\AMG$, we must specify the construction of the associated \gls{ptmdp} $\OutPTMDP = \OutPTMDPDefFromPTG$ where we need to define the elements of the \gls{ptg} $\OutPTG=\OutPTGDef$.

\subsection{Construction of the \glsfmtshort{ptg} for \glsfmtshort{amg}}

First we exhibit the elements of $\OutPTG = \OutPTGDef$.
\else % FROMPTG not defined
Given an input \gls{amg} $\AMG$, we can now exhibit the construction of the associated \gls{ptmdp} $\OutPTMDP = \OutPTMDPDef$.

\shrinkalt{}{\subsection{Construction of the \glsfmtshort{ptmdp} for \glsfmtshort{amg}}

The goal is to exhibit the elements of $\OutPTMDP = \OutPTMDPDef$.}
\fi% FROMPTG

\subsubsection{Locations, initial location, and clock set.}
Let $\AMGgeneral$ be the input \gls{amg}. We define $\Locations_\AMG = 2^\AtomicAttacks\times 2^{\Nodes} /\sim$ the set of locations and $\loc_{\AMG, 0} = [\emptyset, \emptyset]$ the initial location. We define $\Clocks_\AMG = \{x_{a}\}_{a\in \AtomicAttacks} \cup \{x_{d}\}_{d\in \Defenses} \cup \{x_0\}$ the set of clocks associated to the different atomic attacks, \glspl{mtd}, and $x_0$ the global time clock.
\subsubsection{Actions.}
The controllable actions set is $\ControllableActions_\AMG = \left\{\activate_a \mid a \in \AtomicAttacks \right \}$, corresponding to each atomic attack activation.
The uncontrollable actions set is $\UncontrollableActions_\AMG = \UncontrollableActions_\mathsf{mtd} \cup \UncontrollableActions_\mathsf{cmp}$ where,
\shrinkalt{$\UncontrollableActions_\mathsf{mtd} = \left\{\mtd_{d}, \mtdFail_{d} \Bigm\vert d \in \Defenses\right\}$ and $\UncontrollableActions_\mathsf{cmp} = \left\{\completion_{a}, \completionFail_{a} \Bigm\vert a \in \AtomicAttacks \right\}$.}{
\begin{align*}
    \UncontrollableActions_\mathsf{mtd} &= \left\{\mtd_{d}, \mtdFail_{d} \Bigm\vert d \in \Defenses\right\}\\
    \UncontrollableActions_\mathsf{cmp} &= \left\{\completion_{a}, \completionFail_{a} \Bigm\vert a \in \AtomicAttacks \right\}
\end{align*}}
They correspond respectively to the periodical activation of every \gls{mtd} $d\in \Defenses$ with the success of the defense ($\mtd_d$) or failure ($\mtdFail_d$) and the completion of every atomic attack $a \in A$ with the success of the attack ($\completion_a$) or failure ($\completionFail_a$). The set $\UncontrollableActions_\mathsf{mtd}$ contains actions for each \gls{mtd} but not for each subset of \glspl{mtd}. This would be wrong in the general case but a restriction on the \gls{amg} presented later justifies this choice.
\subsubsection{Transitions.}
The set of transitions is $\Transitions_\AMG = \Transitions^\mathsf{act} \cup \Transitions^\mathsf{mtd} \cup \Transitions^\mathsf{cmp}$ with,
\begin{align*}
    \Transitions^\mathsf{act} &= \bigl\{(\loc, \varepsilon,\activate_a, \loc \cup (\{a\},\emptyset)) \mid
    \loc\in \Locations_\AMG, a\in \AtomicAttacks\setminus (\Aof{\loc}\cup \Cof{\loc} \cup\subTree{\AMG}{\Cof{\loc} \cap \noDef})\bigr\}\\
    \Transitions^\mathsf{mtd} &= \bigl\{(\loc, \varepsilon, \mtd_d, \loc \setminus (\defended{\AMG}{d}, \defended{\AMG}{d})) \mid \loc\in \Locations_\AMG, d\in \Defenses\bigr\}\\
    &\hspace{0.5cm}\cup \bigl\{(\loc, \varepsilon, \mtdFail_d, \loc) \mid \loc\in \Locations_\AMG, d\in \Defenses\bigr\}\\
    \Transitions^\mathsf{cmp} &= \big\{(\loc, \varepsilon, \completion_a, \loc\cup (\emptyset, \{a\})) \mid \loc\in \Locations_\AMG, a\in \Aof{\loc}\big\}\\
    &\hspace{0.5cm}\cup \big\{(\loc, \varepsilon, \completionFail_d, \loc\setminus (\{a\}, \emptyset)) \mid \loc\in \Locations_\AMG, a\in \Aof{\loc}\big\}
\end{align*}
where the set $\Transitions^\mathsf{act}$ contains activation edges for every location $\loc \in \Locations_\AMG$ and non-activated, non-completed, and not in the completed descendants of checkpoints, atomic attack $a\in \AtomicAttacks\setminus (\Aof{\loc}\cup \Cof{\loc} \cup \subTree{\AMG}{\Cof{\loc} \cap \noDef})$. The transitions in $\Transitions^\mathsf{mtd}$ correspond to the successful and unsuccessful activation of every \gls{mtd} $d\in\Defenses$.
Finally, $\Transitions^\mathsf{cmp}$ is the set of successful and unsuccessful completion transitions for every location $\loc\in \Locations$ and activated atomic attack $a\in \Aof{\loc}$.  \figurename~\ref{fig:sample-ADMDP} displays a sample of the \gls{ptmdp} with the different kinds of transitions from a given location $[A,C]$.

The reader could expect the clock constraints for $\Transitions^\mathsf{mtd}$ (resp. $\Transitions^\mathsf{cmp}$) to be $x_d\geq \timefunc_d$ (resp. $x_a \leq \timefunc_a$). However, the environment transition density $\density_\AMG{}$, defined later, will assign a null probability to uncontrollable transition before the defense $d$ (resp. atomic attack $a$) verifies $x_d \geq \timefunc_d$ (resp. $x_a \geq \timefunc_a$). So the uncontrollable transitions can have an empty bound $\varepsilon$.

\begin{figure}[t]
    \centering
    \begin{tikzpicture}[
            ->,
            >=stealth',
            shorten >=1pt,
            auto,
            node distance=2.5cm,
            semithick,
            sloped
            ]
        
        \node[p2node] (normal){$A,C$};
        \node[invariant, align=center] (label_inv) at (-2.5,1.2) {$\land_{a\in A}(x_a \leq \timefunc_a)$\\$\land_{d\in\Defenses} (x_d \leq \timefunc_d)$};
        \node[weight] (label_weight) at (-2.5,1.9) {$\sum_{a\in A}\propcost_a$};
        \node[p2node] (activated) [left of=normal] {$A\cup \{a'\},$\\$C$};
        \node[p2node] (completion_failed) [below left of=normal] {$A\setminus \{a\},$\\$C$};
        \node[p2node] (completion_success) [below of=normal] {$A,$\\$C\cup \{a\}$};
        \node[p2node] (mtd_success) [below right of=normal] {$A\setminus \defended{\AMG}{d},$\\$C\setminus \defended{\AMG}{d}$};

        \draw[draw=violet, -, bend right] (normal) edge (label_inv);
        \draw[draw=red, -, bend right=25] (normal) edge (label_weight);
        
        \path (normal) edge [loop above, dashed, looseness=10] node [clock, above=0.3] {$x_{d'} \leftarrow 0$} %node [proba, above=20] {$\dirac{x_{d'} - \timefunc_{d'}}$}
        node [scaledown, above] {$\mtd_{d'}$} (normal);
        \path (normal) edge [loop, dashed, in=25, out=55, looseness=10] node [clock, above=0.3] {$x_{d'} \leftarrow 0$}
        %node [proba, above=20] {$(1-p_{d'})\dirac{x_{d'} - \timefunc_{d'}}$}
        node [scaledown, above] {$\mtdFail_{d'}$} (normal);
        \path (normal) edge node [scaledown, align=center, above=0.3] {$\activate_{a'}$} node[clock,below] {$x_{a'} \leftarrow 0$} node[weight,above] {$\cost_{a'}$} (activated);
        \path (normal)  edge [dashed] node [scaledown, below] {$\mtd_d$}
        % node[proba, above=10] {$p_{d}\dirac{x_{d} - \timefunc_{d}}$}
        node[clock, above] {$x_d\leftarrow 0$} (mtd_success);
        \path (normal) edge [loop, dashed, in=-14, out=16, looseness=10] node [scaledown, align=center, above] {$\mtdFail_d$}
        node [clock, above=0.3] {$x_{d} \leftarrow 0$}
        %node [proba, below=12] {$(1-p_{d})\dirac{x_{d} - \timefunc_{d}}$}
        (normal);
        \path (normal)  edge [dashed] node [scaledown, align=center, below]{$\completion_a$}
        %node [above, proba] {$p_{a}\dirac{x_{a} - \timefunc_{a}}$}
        (completion_success);
        \path (normal) edge [dashed] node [scaledown, align=center] {$\completionFail_a$}
        %node [below, proba] {$(1-p_{a})\dirac{x_{a} - \timefunc_{a}}$}
        (completion_failed);
        
        \matrix [draw, column sep=5pt, row sep=3pt, inner sep=2pt, ampersand replacement=\&, column 1/.style={anchor=base}, column 2/.style={anchor=base west}] at (4, 0) {
        \node [state, minimum size=2em, scale=0.6] {$X,Y$}; \& \node[scale=0.8] {location $[X,Y]$};\\
        \draw[dashed, line width=0.5pt, >=stealth', semithick](-0.25, 0) -- (0.25,0); \& \node[right, scale=0.8, align=left] {uncontrollable\\[-0.5em]transition};\\
        \draw[line width=0.5pt, >=stealth', semithick](-0.25, 0) -- (0.25,0); \& \node[right, scale=0.8, align=left] {controllable\\[-0.5em]transition};\\
        \node [invariant, scale=0.8] {$x \bowtie y$}; \& \node[right, scale=0.8, align=left] at (0, 0.2em) {location\\[-0.5em]invariant};\\
        %\node [guard, scale=0.8] {$x \bowtie y$}; \& \node[scale=0.8, align=left] {edge guard};\\
        \node [weight, scale=0.8] {$c$}; \& \node[right, scale=0.8, align=left] at (0, 0.2em) {location and\\[-0.5em]transition cost};\\
        \node [scale=0.8] {$\action^x_y$}; \& \node[scale=0.8, align=left] {action type};\\
        \node [clock, scale=0.8] {$x \leftarrow 0$}; \& \node[scale=0.8, align=left] {clock reset};\\
        };

    \end{tikzpicture}
    \ifdefined\FROMPTG
    \caption[Sample of the transitions in the \glsfmtshort{ptg}.]{Sample of the transitions in the \gls{ptg} $\OutPTG$ from a location $[A,C]\in \Locations$. Notice that there are as many outgoing transitions from $[A, C]$ as there are such $a\in A$, $a'\in \AtomicAttacks\setminus(A\cup C \cup \subTree{\AMG}{\propagate{\AMG}{C}\cap \noDef})$, $d \in \defense{A\cup C}$, and $d' \in \Defenses \setminus \defense{A \cup C}$.
    Moreover, a way to construct the full \gls{ptg} is to do a transitive closure of these transitions from the initial location~$[\emptyset, \emptyset]$: we build recursively the transitions from the new nodes that are reached until all nodes are expanded.
    }
    \else % FROMPTG not defined
    \caption[Sample of the transitions in the \glsfmtshort{ptmdp}.]{Sample of the transitions in the \gls{ptmdp} $\OutPTMDP$ from a location $[A,C]\in \Locations$. Notice that there are as many outgoing transitions from $[A, C]$ as there are such $a\in A$, $a'\in \AtomicAttacks\setminus(A\cup C \cup \subTree{\AMG}{\propagate{\AMG}{C}\cap \noDef})$, $d \in \defense{A\cup C}$, and $d' \in \Defenses \setminus \defense{A \cup C}$.
    \shrinkalt{}{
    Moreover, a way to construct the full \gls{ptmdp} is to do a transitive closure of these transitions from the initial location~$[\emptyset, \emptyset]$: we build recursively the transitions from the new nodes that are reached until all nodes are expanded.}
    }
    \fi
    \label{fig:sample-ADMDP}
\end{figure}
\subsubsection{Cost, clock reset, and invariant.}
Let $\loc, \loc' \in \Locations_\AMG$, $a\in \AtomicAttacks$ with clock $x_a$, $d\in \Defenses$ with clock $x_d$. We define the cost in locations as $\price_\AMG(\loc) = \sum_{a\in \Aof{\loc}}\propcost_a$. The cost for an activation transitions $e$ of the form $e = (\loc, \varepsilon, \activate_a, \loc')$ is $\price_\AMG(e_a) = \cost_a$ and for any other type of transition the cost is null. The clock reset function $\clockReset_\AMG$ is defined as follows when such a transition exists:
\shrinkalt{
$\clockReset_\AMG(\loc, \varepsilon, \activate_a, \loc') = \{x_a\}$, $\clockReset_\AMG(\loc, \varepsilon, \mtd_d, \loc') = \clockReset_\AMG(\loc, \varepsilon, \mtdFail_d, \loc') = \{x_d\}$, and $\clockReset_\AMG(\loc, \varepsilon, \completion_a, \loc') = \clockReset_\AMG(\loc, \varepsilon, \completionFail_a, \loc') = \{x_a\}$.
}{
\begin{align*}
    \clockReset_\AMG(\loc, \varepsilon, \activate_a, \loc') &= \{x_a\}\\
    \clockReset_\AMG(\loc, \varepsilon, \mtd_d, \loc') &= \clockReset_\AMG(\loc, \varepsilon, \mtdFail_d, \loc') = \{x_d\}\\
    \clockReset_\AMG(\loc, \varepsilon, \completion_a, \loc') &= \clockReset_\AMG(\loc, \varepsilon, \completionFail_a, \loc') = \{x_a\}
\end{align*}
}
The invariant function $\invariants_\AMG$ is $\invariants_\AMG(\loc) = \bigwedge_{a\in \Aof{\loc}} (x_a \leq \timefunc_a) \bigwedge_{d \in \Defenses} (x_d \leq \timefunc_d)$.

\ifdefined\FROMPTG
Overall, we have $\OutPTG = \OutPTGDef$. Now, we present a restriction on our structure before exhibiting the density $\density_\AMG$ of $\OutPTMDP$.

\subsection{A restriction on the \glsfmtshort{amg}.}
As seen in Section~\ref{sec:AMG}, when there is simultaneous activation of several \glspl{mtd}, all the successfully defended nodes are removed before evaluating the new state. This sequentiality is essential because, for $\loc \in \Locations_\AMG$ and $d_1, d_2 \in \Defenses$ two \glspl{mtd} that get successfully activated at the same time, it does \textbf{not} hold in general \shrinkalt{that $\loc\setminus(\defended{}{d_1}\cup\defended{}{d_2}, \defended{}{d_1}\cup\defended{}{d_2})$ is equal to $(\loc\setminus(\defended{}{d_1} , \defended{}{d_1})) \setminus (\defended{}{d_2} , \defended{}{d_2})$.}{
\begin{equation}
    \loc\setminus(\defended{}{d_1}\cup\defended{}{d_2}, \defended{}{d_1}\cup\defended{}{d_2}) = (\loc\setminus(\defended{}{d_1} , \defended{}{d_1})) \setminus (\defended{}{d_2} , \defended{}{d_2})
    \label{eq:db-def}
\end{equation}
}
This means that multiple transitions in our \gls{ptmdp} in a row that remove the completed and activated nodes associated with several \gls{mtd} are not equivalent to a single transition that simultaneously removes all the completed and activated nodes. As a result, from any location $\loc$, we have to put a transition for every element of $2^\Defenses$ that is the possible set of \glspl{mtd} activated at a given time. This exponential size is not desired, so we will add a restriction on the \gls{amg} to have only $O(\vert\Defenses\vert)$ outgoing defense edges from every location.
We define a relation that expresses that an \gls{mtd} directly follows another one.
\begin{definition}
Given an \gls{amg} $\AMG$, 
we define $\triangleright_\AMG$ as a binary relationship on $\Defenses$ s.t. for $d_1, d_2 \in \Defenses$,
\shrinkalt{$d_1\triangleright_\AMG d_2$ if there exists $n_1 \in \defended{}{d_1}$ and $n_2 \in \children{}{n_1}$ with $n_2 \not\in \defended{}{d_1}$ and $n_2 \in \defended{}{d_2}$. The relation $d_1 \triangleright_\AMG d_2$ is read ``$d_2$ follows $d_1$ in $\AMG$''.}{
\begin{align*}
    d_1\triangleright_\AMG d_2& \Longleftrightarrow \exists n_1 \in \defended{}{d_1}, \exists n_2 \in \children{}{n_1},n_2 \not\in \defended{}{d_1} \land n_2 \in \defended{}{d_2}
\end{align*}
the relation $d_1 \triangleright_\AMG d_2$ is read ``$d_2$ follows $d_1$ in $\AMG$''.
}
\end{definition}
In words, $d_1 \triangleright_\AMG d_2$ if $d_2$ defends a node $n_2$ that is a child of a node $n_1$ defended by $d_1$, and $d_1$ does not defend $n_2$. We will simply write $d_1 \triangleright d_2$ when evident.
Using this relation, we show a sufficient condition s.t., if several \glspl{mtd} $d_1, \dots, d_k$ are activated successfully at the same time, we can virtually activate them sequentially and obtain the same result as if they were activated simultaneously.
\shrinkalt{
\begin{restatable}{theorem}{defensechaintheorem}\label{thm:defense-chain}
Let $\AMGgeneral$ be an \gls{amg}. Suppose the directed graph of the relation $\triangleright$, i.e., $\brac{\Defenses, \{(d_1, d_2) \in \Defenses \times \Defenses \mid d_1 \triangleright d_2\}}$, has no cycle. Then, for all $D \subseteq \Defenses$, we can order the elements of $D$ in a sequence $(d_1, \dots, d_k)$ s.t. for all $i\in \{1, \dots, k\}$ and integer $j<i$, $d_j \ntriangleright d_i$. Moreover, for all $\loc \in \Locations_\AMG$, this order verifies,
\begin{equation*}
    \loc\setminus(\cup_{j = 1}^k \defended{\AMG}{d_j}, \cup_{j=1}^k \defended{\AMG}{d_j}) = \loc \setminus (\defended{\AMG}{d_1}, \defended{\AMG}{d_1}) \dots \setminus (\defended{\AMG}{d_k}, \defended{\AMG}{d_k})
\end{equation*}
\end{restatable}
}{
\begin{restatable}{lemma}{defensetriangle}\label{lemma:defense-triangle}
Let $\AMG$ an \gls{amg}. For $d_1, d_2 \in \Defenses$, $\loc\in \Locations_\AMG$ a location, if $d_1 \ntriangleright d_2$ then, eq.~\eqref{eq:db-def} holds, i.e.,
\begin{align*}
    \loc\setminus(\defended{}{d_1}\cup\defended{}{d_2}, \defended{}{d_1}\cup\defended{}{d_2}) = (\loc\setminus(\defended{}{d_1} , \defended{}{d_1})) \setminus (\defended{}{d_2} , \defended{}{d_2})
\end{align*}
\end{restatable}
\ifdefined\ALLPROOFS\begin{proof}
Let $d_1, d_2\in \Defenses$, we assume that $d_1 \ntriangleright d_2$. Let
\begin{align*}
    \loc_1 &= \loc\setminus(\defended{}{d_1}\cup\defended{}{d_2}, \defended{}{d_1}\cup\defended{}{d_2})\\
    \loc_2 &= \loc\setminus(\defended{}{d_1} , \defended{}{d_1})\\
    \loc_3 &= \loc_2 \setminus (\defended{}{d_2} , \defended{}{d_2}) = \loc\setminus(\defended{}{d_1} , \defended{}{d_1}) \setminus (\defended{}{d_2} , \defended{}{d_2})
\end{align*}
We have
\begin{align}
    \Cof{\loc_1} &= \propagate{\AMG}{\Cof{\loc}\setminus(\defended{\AMG}{d_1}\cup \defended{\AMG}{d_2})} \setminus \subTree{\AMG}{\propagate{\AMG}{\Cof{\loc}\setminus(\defended{\AMG}{d_1}\cup \defended{\AMG}{d_2})} \cap \noDef}\label{eq:def-cofl1}\\
    \Cof{\loc_2} &= \propagate{\AMG}{\Cof{\loc}\setminus\defended{\AMG}{d_1}} \setminus \subTree{\AMG}{\propagate{\AMG}{\Cof{\loc}\setminus\defended{\AMG}{d_1}}\cap \noDef}\nonumber\\%\label{eq:def-cofl2}\\
    \Cof{\loc_3} &= \propagate{\AMG}{\Cof{\loc_2}\setminus\defended{\AMG}{d_2}} \setminus \subTree{\AMG}{\propagate{\AMG}{\Cof{\loc_2}\setminus \defended{\AMG}{d_2}}\cap \noDef}\label{eq:def-cofl3}
\end{align}
And the condition $d_1 \ntriangleright d_2$ is equivalent to
\begin{equation}\label{eq:not-d2-follows-d1}
    \forall n_1 \in \defended{\AMG}{d_1}, \forall n_2 \in \children{}{n_1},
    n_2 \in \defended{}{d_1} \lor n_2 \not\in \defended{}{d_2}
\end{equation}
By Proposition~\ref{prop:simple-state-proj}, $\Cof{\simpleState{\loc}} = \Cof{\loc}$, that is 
\begin{equation*}
    \Cof{\loc} = \propagate{\AMG}{\Cof{\loc}} \setminus \subTree{\AMG}{\propagate{\AMG}{\Cof{\loc}} \cap \noDef}
\end{equation*}
we can deduce,
\begin{equation*}%\label{eq:remove-subtree-ok}
    \Cof{\loc} \setminus \subTree{\AMG}{\propagate{\AMG}{\Cof{\loc}} \cap \noDef} = \Cof{\loc}
\end{equation*}
Now we have
\begin{align*}
    \Cof{\loc} \setminus (\defended{\AMG}{d_1} \cup \defended{\AMG}{d_2}) &= \Cof{\loc} \setminus \defended{\AMG}{d_1} \setminus \defended{\AMG}{d_2} \setminus \subTree{\AMG}{\propagate{\AMG}{\Cof{\loc}} \cap \noDef}\\
    &\subseteq \propagate{\AMG}{\Cof{\loc} \setminus \defended{\AMG}{d_1}} \setminus \defended{\AMG}{d_2} \setminus \subTree{\AMG}{\propagate{\AMG}{\Cof{\loc}} \cap \noDef}\\
    &= \propagate{\AMG}{\Cof{\loc} \setminus \defended{\AMG}{d_1}}  \setminus \subTree{\AMG}{\propagate{\AMG}{\Cof{\loc} \setminus \defended{\AMG}{d_1}} \cap \noDef} \setminus \defended{\AMG}{d_2}\\
    &= \Cof{\loc_2} \setminus \defended{\AMG}{d_2}
\end{align*}

We want to show by induction
\begin{equation*}
    \propagate{\AMG}{\Cof{\loc}\setminus\defended{\AMG}{d_1}} \setminus \subTree{\AMG}{\propagate{\AMG}{\Cof{\loc}}} \setminus \defended{\AMG}{d_2} \subseteq \propagate{\AMG}{\Cof{\loc}\setminus\defended{\AMG}{d_1} \setminus \defended{\AMG}{d_2}}
\end{equation*}

For $k \in \mathbb N$, let
\begin{align*}
    N_k &= f_{\propagateDef{}}^k(\Cof{\loc}\setminus\defended{\AMG}{d_1}) \setminus \subTree{\AMG}{\propagate{\AMG}{\Cof{\loc}}}\\
    M_k &= f_{\propagateDef{}}^k(\Cof{\loc}\setminus\defended{\AMG}{d_1} \setminus\defended{\AMG}{d_2})\\
    \mathcal H^1_k & : N_k \setminus \defended{\AMG}{d_2} \subseteq M_k\\
    \mathcal H^2_k & : N_k \cap \defended{\AMG}{d_1} \cap \defended{\AMG}{d_2} \subseteq M_k
\end{align*}
where $\mathcal H^1_k \land \mathcal H^2_k$ is our induction hypothesis.
We have $\mathcal H^1_0 \land \mathcal H^2_0$. For $k \in \mathbb N$ we suppose $\mathcal H^1_k \land \mathcal H^2_k$, we want to show $\mathcal H^1_{k+1} \land \mathcal H^2_{k+1}$. First we prove $\mathcal H^1_{k+1}$.

Let $n \in N_{k+1} \setminus \defended{\AMG}{d_2}$.
\begin{itemize}
    \item If $n \in N_k \setminus \defended{\AMG}{d_2}$, then $n\in M_k$ and by monotony, $n \in M_{k+1}$.
    \item Else, we have $n\not\in N_k \setminus \defended{\AMG}{d_2}$. We notice that $N_k \subseteq \Cof{\loc}$ so $n \in \Cof{\loc} \setminus\defended{\AMG}{d_2}$.
    \begin{itemize}
        \item If $n\not\in \defended{\AMG}{d_1}$, then $n\in \Cof{\loc} \setminus \defended{\AMG}{d_1} \setminus \defended{\AMG}{d_2} \subseteq M_k$ by monotony.
        \item Else, $n\in \defended{\AMG}{d_1}$. As such, $n$ can not be a checkpoint because a checkpoint must be an undefended node. Moreover, if $n\in \defended{\AMG}{d_2}$, then $n\in N_k \cap \defended{\AMG}{d_1} \cap \defended{\AMG}{d_2}$ and by $\mathcal H^2_k$ we have $n\in M_k\subseteq M_{k+1}$. For the rest we suppose $n \not \in \defended{\AMG}{d_2}$, which implies $n\not \in f_{\propagateDef{}}^k(\Cof{\loc}\setminus\defended{\AMG}{d_1})$.
        
        As $n\in f_{\propagateDef{}}^{k+1}(\Cof{\loc}\setminus\defended{\AMG}{d_1}) \setminus f_{\propagateDef{}}^k(\Cof{\loc}\setminus\defended{\AMG}{d_1})$, we can take $n_1, \dots, n_l \in \children{}{n}$ such that for all $i \in \{1, \dots, k\}$, we have $n_i \in f_{\propagateDef{}}^k(\Cof{\loc}\setminus\defended{\AMG}{d_1})$ and $(n_1, \dots, n_k)$ is the list of all children of $n$ if $\operation{n}$ is a conjunction, or is a non empty list of nodes if $\operation{n}$ is a disjunction. As $n$ is not a checkpoint and does not have checkpoint in $\subTree{\AMG}{\propagate{\AMG}{\Cof{\loc}}}$ on every path from $\AMGroot$ we can conclude that for all $i\in \{1, \dots, k\}$, $n_i \not \in \subTree{\AMG}{\propagate{\AMG}{\Cof{\loc}}}$, so $n_i \in N_k$.
        
        By eq.~\eqref{eq:not-d2-follows-d1}, we have that for all $i\in \{1,\dots, k\}$, $n_i\in \defended{\AMG}{d_1}$ or $n_i\not\in \defended{\AMG}{d_2}$.
        \begin{itemize}
            \item If $n_i \not\in \defended{\AMG}{d_2}$, then $n_i\in N_k \setminus \defended{\AMG}{d_2} \subseteq M_k$ by $\mathcal H^1_k$.
            \item Otherwise, $n_i \in \defended{\AMG}{d_1} \cap \defended{\AMG}{d_2}$, so $n_i \in N_k \cap \defended{\AMG}{d_1} \cap \defended{\AMG}{d_2} \subseteq M_k$ by $\mathcal H^2_k$.
        \end{itemize}
        In any cases, for all $i\in \{1,\dots, k\}$, $n_i \in M_k$ so $n \in M_{k+1}$.
    \end{itemize}
\end{itemize}
We treated all the cases and proved $\mathcal H^1_{k+1}$.

Now we prove $\mathcal H^2_{k+1}$. Let $n \in N_{k+1} \cap \defended{\AMG}{d_1} \cap \defended{\AMG}{d_2}$.
\begin{itemize}
    \item If $n\in f_{\propagateDef{}}^k(\Cof{\loc}\setminus\defended{\AMG}{d_1})$ then $n \in N_{k} \cap \defended{\AMG}{d_1} \cap \defended{\AMG}{d_2} \subseteq M_{k+1}$ by $\mathcal H^2_k$.
    \item Otherwise, $n\in f_{\propagateDef{}}^{k+1}(\Cof{\loc}\setminus\defended{\AMG}{d_1}) \setminus f_{\propagateDef{}}^k(\Cof{\loc}\setminus\defended{\AMG}{d_1})$. As before, we can take $n_1, \dots, n_l \in \children{}{n}$ such that for all $i \in \{1, \dots, k\}$, we have $n_i \in f_{\propagateDef{}}^k(\Cof{\loc}\setminus\defended{\AMG}{d_1})$ and $(n_1, \dots, n_k)$ is the list of all children of $n$ if $\operation{n}$ is a conjunction, or is a non empty list of nodes if $\operation{n}$ is a disjunction. As $n$ is not a checkpoint and does not have checkpoint in $\subTree{\AMG}{\propagate{\AMG}{\Cof{\loc}}}$ on every path from $\AMGroot$ we can conclude that for all $i\in \{1, \dots, k\}$, $n_i \not \in \subTree{\AMG}{\propagate{\AMG}{\Cof{\loc}}}$, so $n_i \in N_k$.
    
    By eq.~\eqref{eq:not-d2-follows-d1}, we have that for all $i\in \{1,\dots, k\}$, $n_i\in \defended{\AMG}{d_1}$ or $n_i\not\in \defended{\AMG}{d_2}$.
        \begin{itemize}
            \item If $n_i \not\in \defended{\AMG}{d_2}$, then $n_i\in N_k \setminus \defended{\AMG}{d_2} \subseteq M_k$ by $\mathcal H^1_k$.
            \item Otherwise, $n_i \in \defended{\AMG}{d_1} \cap \defended{\AMG}{d_2}$, so $n_i \in N_k \cap \defended{\AMG}{d_1} \cap \defended{\AMG}{d_2} \subseteq M_k$ by $\mathcal H^2_k$.
        \end{itemize}
        In any cases, for all $i\in \{1,\dots, k\}$, $n_i \in M_k$ so $n \in M_{k+1}$.
\end{itemize}
We treated all the cases and proved $\mathcal H^2_{k+1}$.

To summarize, we proved the following
\begin{equation*}%\label{eq:result-summary}
\left\{
\begin{matrix}
    &\Cof{\loc} \setminus (\defended{\AMG}{d_1} \cup \defended{\AMG}{d_2}) \subseteq \Cof{\loc_2} \setminus \defended{\AMG}{d_2}\\
    &\Cof{\loc_2} \setminus \defended{\AMG}{d_2} \subseteq \Cof{\loc_1}
\end{matrix}
\right.
\end{equation*}
We need to notice that with Lemma~\ref{lemma:subtree-subtree}, we have
\begin{equation*}
    \subTree{\AMG}{\propagate{\AMG}{\Cof{\loc}\setminus(\defended{\AMG}{d_1}\cup \defended{\AMG}{d_2})} \cap \noDef} = \subTree{\AMG}{\Cof{\loc_1} \cap \noDef}
\end{equation*}
So using eq.~\eqref{eq:def-cofl1} and~\eqref{eq:def-cofl3}.
\begin{align*}
    \Cof{\loc_1} &= \propagate{\AMG}{\Cof{\loc}\setminus(\defended{\AMG}{d_1}\cup \defended{\AMG}{d_2})} \setminus \subTree{\AMG}{\Cof{\loc_1} \cap \noDef}\\
    &\subseteq \propagate{\AMG}{\Cof{\loc_2}\setminus\defended{\AMG}{d_2}} \setminus \subTree{\AMG}{\Cof{\loc_2}\setminus \defended{\AMG}{d_2}\cap \noDef}\\
    &\subseteq \propagate{\AMG}{\Cof{\loc_2}\setminus\defended{\AMG}{d_2}} \setminus \subTree{\AMG}{\propagate{\AMG}{\Cof{\loc_2}\setminus \defended{\AMG}{d_2}}\cap \noDef}\\
    &= \Cof{\loc_3}\\
    \Cof{\loc_3} &= \propagate{\AMG}{\Cof{\loc_2}\setminus\defended{\AMG}{d_2}} \setminus \subTree{\AMG}{\propagate{\AMG}{\Cof{\loc_2}\setminus \defended{\AMG}{d_2}}\cap \noDef}\\
    &\subseteq \propagate{\AMG}{\Cof{\loc_1}} \setminus \subTree{\AMG}{\propagate{\AMG}{\Cof{\loc} \setminus (\defended{\AMG}{d_1} \cup \defended{\AMG}{d_2})}\cap \noDef}\\
    &= \Cof{\loc_1}
\end{align*}
So we have $\Cof{\loc_1} = \Cof{\loc_3}$.

We want to show $\Aof{\loc_1} = \Aof{\loc_3}$. By Proposition~\ref{prop:subtree-defense}
\begin{align*}
    \subTree{\AMG}{\propagate{\AMG}{\Cof{\loc}}\cap \noDef} &= \subTree{\AMG}{\propagate{\AMG}{\Cof{\loc} \setminus \defended{\AMG}{d_1}}\cap \noDef}\\
    &= \subTree{\AMG}{\propagate{\AMG}{\Cof{\loc} \setminus \defended{\AMG}{d_1}} \subTree{\AMG}{\propagate{\AMG}{\Cof{\loc} \setminus \defended{\AMG}{d_1}} \cap \noDef}\cap \noDef}\\
    &= \subTree{\AMG}{\Cof{\loc_2}\cap \noDef}\\
    &= \subTree{\AMG}{\propagate{\AMG}{\Cof{\loc_2}} \setminus\subTree{\AMG}{\propagate{\AMG}{\Cof{\loc_2}} \cap \noDef}\cap \noDef}\\
    &=\subTree{\AMG}{\propagate{\AMG}{\Cof{\loc_2}}\cap \noDef}
\end{align*}
Furthermore, using the definition definition
\begin{align*}
    \Aof{\loc_1} &= \Aof{\loc} \setminus \left(\subTree{\AMG}{\propagate{\AMG}{\Cof{\loc} \setminus (\defended{\AMG}{d_1}\cup \defended{\AMG}{d_2})}\cap \noDef} \cup \propagate{\AMG}{\Cof{\loc} \setminus (\defended{\AMG}{d_1}\cup \defended{\AMG}{d_2}})\right)\\
    &= \Aof{\loc} \setminus \left(\subTree{\AMG}{\propagate{\AMG}{\Cof{\loc}}\cap \noDef} \cup \Cof{\loc_1}\right)\\
    \Aof{\loc_2} &= \Aof{\loc} \setminus \left(\subTree{\AMG}{\propagate{\AMG}{\Cof{\loc} \setminus \defended{\AMG}{d_1}}\cap \noDef} \cup \propagate{\AMG}{\Cof{\loc} \setminus \defended{\AMG}{d_1}}\right)\\
    &= \Aof{\loc} \setminus \left(\subTree{\AMG}{\propagate{\AMG}{\Cof{\loc}}\cap \noDef} \cup \Cof{\loc_2}\right)\\
    \Aof{\loc_3} &= \Aof{\loc_2} \setminus \left(\subTree{\AMG}{\propagate{\AMG}{\Cof{\loc_2} \setminus \defended{\AMG}{d_2}}\cap \noDef} \cup \propagate{\AMG}{\Cof{\loc_2} \setminus \defended{\AMG}{d_2}}\right)\\
    &= \Aof{\loc} \setminus \left(\subTree{\AMG}{\propagate{\AMG}{\Cof{\loc}}\cap \noDef} \cup \Cof{\loc_2} \cup \Cof{\loc_3}\right)\\
\end{align*}

We have $\Aof{\loc} = \Aof{\simpleState{\loc}}$ so $\Aof{\loc} = \Aof{\loc} \setminus ( \Cof{\loc} \cup \subTree{\AMG}{\propagate{\AMG}{\loc} \cap \noDef})$. As $\Cof{\loc_1}$, $\Cof{\loc_2}$ and $\Cof{\loc_3}$ are all included in  $\Cof{\loc}$, we have $\Aof{\loc_1} = \Aof{\loc_3}$.
\end{proof}\else\begin{sketchproof}
The intuition is that, if eq.~\eqref{eq:db-def} does not hold,
then $d_1 \triangleright d_2$ because the propagation operator to compute the state $\loc \setminus (\defended{\AMG}{d_1}, \defended{\AMG}{d_1}) = \pruneDef{\AMG} \circ \propagate{\AMG}{\Aof{\loc} \setminus \defended{\AMG}{d_1}, \Cof{\loc} \setminus \defended{\AMG}{d_1}}$ uses completed nodes in $\defended{\AMG}{d_2}$.
\end{sketchproof}\fi
\begin{restatable}{theorem}{defensechaintheorem}\label{thm:defense-chain}
Let $\AMGgeneral$ be an \gls{amg}. Suppose the directed graph of the relation $\triangleright$, i.e., $\brac{\Defenses, \{(d_1, d_2) \in \Defenses \times \Defenses \mid d_1 \triangleright d_2\}}$, has no cycle. Then, for all $D \subseteq \Defenses$, we can order the elements of $D$ in a sequence $(d_1, \dots, d_k)$ s.t. for all $i\in \{1, \dots, k\}$ and integer $j<i$, $d_j \ntriangleright d_i$. Moreover, for all $\loc \in \Locations_\AMG$, this order verifies,
\begin{equation*}
    \loc\setminus(\cup_{j = 1}^k \defended{\AMG}{d_j}, \cup_{j=1}^k \defended{\AMG}{d_j}) = \loc \setminus (\defended{\AMG}{d_1}, \defended{\AMG}{d_1}) \dots \setminus (\defended{\AMG}{d_k}, \defended{\AMG}{d_k})
\end{equation*}
\end{restatable}
\ifdefined\ALLPROOFS\begin{proof}
Suppose $\brac{\Defenses, \{(d_1, d_2) \in \Defenses \times \Defenses \mid d_1 \triangleright d_2\}}$ has no cycle. This means that there is no chain of the form $d_1\triangleright d_2 \triangleright \dots \triangleright d_k \triangleright d_1$. So we can choose i such that $d_j \ntriangleright d_i$ for all the $j\in [1,\dots, k]$ (notice that $d \ntriangleright d$ always holds). Let $d$ be a new defense such that $\defended{\AMG}{d} = \cup_{j\in \{1,\dots, k\}\setminus \{i\}} \defended{\AMG}{d_j}$. We have $d\ntriangleright d_i$, so, by Lemma~\ref{lemma:defense-triangle}, it holds
\begin{equation*}
    \loc\setminus(\cup_{j = 1}^k \defended{\AMG}{d_j}, \cup_{j=1}^k \defended{\AMG}{d_j}) = (\loc \setminus(\cup_{j\in \{1, \dots, k\}\setminus\{i\}} \defended{\AMG}{d_j}, \cup_{j\in \{1, \dots, k\}\setminus\{i\}} \defended{\AMG}{d_j})) \setminus (\defended{\AMG}{d_i}, \defended{\AMG}{d_i})
\end{equation*}
And recursively we can always activate one defense at a time.
\end{proof}\else\begin{sketchproof}
By induction on Lemma~\ref{lemma:defense-triangle}, this proves the theorem.
\end{sketchproof}\fi
}
\shrinkalt{We refer to~\cite{amg} for the proof.}{} Now, we impose that the input \gls{amg} $\AMG$ verifies that the directed graph of the relation $\triangleright$, that is $\brac{\Defenses, \{(d_1, d_2) \in \Defenses \times \Defenses \mid d_1 \triangleright d_2\}}$, has no cycle.
So if at some point the \glspl{mtd} $d_1,\dots, d_k \in \Defenses$ are successfully activated at a given time, we can evaluate them sequentially starting with $d_i$ where it holds that  $d_j\ntriangleright\  d_i$ for all the $j\in \{1,\dots, k\}$.
\subsection{Environment density}
Let $\loc \in \Locations_\AMG$, $\delay\in \Delays$, $v\in \Valuations$ be a valid valuation, and
\shrinkalt{
\begin{align*}
    A^\delay_{(\loc,v)} &= \{a\in \Aof{\loc} \mid v(x_{a}) + \delay = \timefunc_{a}\}\\
    D^\delay_{(\loc,v)} &= \{d\in \Defenses \mid v(x_{d}) + \delay = \timefunc_{d}\land \forall d' \in \Defenses, d\triangleright d' \Rightarrow v(x_{d'}) + \delay \neq \timefunc_{d'}\}\\
    \coefmu{\delay}{\loc,v} &= 1/\vert A^\delay_{(\loc,v)}\cup D^\delay_{(\loc,v)}\vert
\end{align*}
be respectively the set of activated atomic attacks completed after the delay $\delay$, the set of \glspl{mtd} $d$ activated after $\delay$ s.t. any other \gls{mtd} $d'$ in relation $d \triangleright d'$ is not active after the same delay, and, the inverse of their number of elements. By convention, $1/0=0$ in $\coefmu{\delay}{\loc,v}$.
}{
\begin{align*}
    A^\delay_{(\loc,v)} &= \{a\in \Aof{\loc} \mid v(x_{a}) + \delay = \timefunc_{a}\}\\
    D^\delay_{(\loc,v)} &= \{d\in \Defenses \mid v(x_{d}) + \delay = \timefunc_{d}\land \forall d' \in \Defenses, d\triangleright d' \Rightarrow v(x_{d'}) + \delay \neq \timefunc_{d'}\}\\
\end{align*}
be respectively the set of activated atomic attacks completed after the delay $\delay$, and the set of \glspl{mtd} $d$ activated after $\delay$ s.t. any other defense $d'$ in relation $d \triangleright d'$ is not active after the same delay. Moreover, we define $\coefmu{\delay}{\loc,v}$, the inverse of their number of elements if it exists.
\begin{equation*}
    \coefmu{\delay}{\loc,v} = \begin{cases}
    0 &\text{if } \vert A^\delay_{(\loc,v)}\cup D^\delay_{(\loc,v)}\vert = 0,\\
    \frac{1}{\vert A^\delay_{(\loc,v)}\cup D^\delay_{(\loc,v)}\vert} & \text{otherwise.}
    \end{cases}
\end{equation*}
}
For $a\in\AtomicAttacks$, $d\in \Defenses$, we define $\density_\AMG{(\loc, v)}$ as,
\shrinkalt{
\begin{align*}
    \density_\AMG{(\loc, v)}(\delay, \mtd_{d}) &= \coefmu{\delay}{\loc,v} \proba_{d} \dirac{v(x_{d}) + \delay - \timefunc_{d}}\\
    \density_\AMG{(\loc, v)}(\delay, \mtdFail_{d}) &= \coefmu{\delay}{\loc,v} (1-\proba_{d}) \dirac{v(x_{d}) + \delay - \timefunc_{d}}\\
    \density_\AMG{(\loc, v)}(\delay, \completion_{a}) &= \coefmu{\delay}{\loc,v} \proba_{a} \dirac{v(x_{a}) + \delay - \timefunc_{a}}\\
    \density_\AMG{(\loc, v)}(\delay, \completionFail_{a}) &= \coefmu{\delay}{\loc,v} (1-\proba_{a}) \dirac{v(x_{a}) + \delay - \timefunc_{a}}
\end{align*}
}{
\begin{align}
    \density_\AMG{(\loc, v)}(\delay, \mtd_{d}) &= \coefmu{\delay}{\loc,v} \proba_{d} \dirac{v(x_{d}) + \delay - \timefunc_{d}}\label{eq:mu1}\\
    \density_\AMG{(\loc, v)}(\delay, \mtdFail_{d}) &= \coefmu{\delay}{\loc,v} (1-\proba_{d}) \dirac{v(x_{d}) + \delay - \timefunc_{d}}\label{eq:mu2}\\
    \density_\AMG{(\loc, v)}(\delay, \completion_{a}) &= \coefmu{\delay}{\loc,v} \proba_{a} \dirac{v(x_{a}) + \delay - \timefunc_{a}}\label{eq:mu3}\\
    \density_\AMG{(\loc, v)}(\delay, \completionFail_{a}) &= \coefmu{\delay}{\loc,v} (1-\proba_{a}) \dirac{v(x_{a}) + \delay - \timefunc_{a}}\label{eq:mu4}
\end{align}
}
where $\diracDef$ is the Dirac distribution used for discrete probabilities.
This density function reflects that %the deadline $x_a$ = $\timefunc_a$ for an atomic attack $a$, and the deadline $x_d= \timefunc_d$ for a MTD $d$, and no defense $d'\in \Defenses$ is such that $d\triangleright d'$ and $x_{d'} = \timefunc_{d'}$ means that the uncontrollable actions satisfy their activation condition.
the uncontrollable actions satisfying their activation condition after a delay $\delay$ are chosen with uniform probability (through the use of $\coefmu{\delay}{\loc,v}$). The probability of success (resp. failure) is chosen with probability $\proba_a$ (resp. $1-\proba_a$) for an atomic attack $a$, and $\proba_d$ (resp. $1-\proba_d$) for an \gls{mtd} $d$.

%Altogether, we define the PTMDP associated with the AMG $\AMG$ as $\OutPTMDP = \OutPTMDPDef$.

\else % FROMPTG not defined
Now we have all the elements to present a restriction on our structure before exhibiting the last element $\density_\AMG$ of $\OutPTMDP$.

\subsubsection{A restriction on the \glsfmtshort{amg}.}

\subsubsection{Environment's density.}

\fi% FROMPTG

\shrinkalt{}{Finally, we built $\OutPTMDP = \OutPTMDPDef$ from an input \gls{amg} $\AMG$.}

\shrinkalt{
\mysubsection{Using the \glsfmtshort{ptmdp} for \glsfmtshort{mtd}.}
Given an \gls{amg} $\AMGgeneral$ and its associated \gls{ptmdp} $\PTMDP_\AMG$, the goal of the attacker is to reach the \gls{ptmdp} state $\PTMDPgoal = [\emptyset, \{\AMGroot\}]$ (we assume there is no \gls{mtd} on $\AMGroot$) representing the completion of the main goal $\AMGroot$.
\shrinkalt{}{The goal node $\PTMDPgoal$ is $[\emptyset, \{\AMGroot\}]$ because we assume that there is no defense on $\AMGroot$ (indeed, once $\AMGroot$ is reached, the game ends, so the defenses are useless here). It implies $\subTree{\AMG}{\{\AMGroot\}} = \Nodes \setminus \{\AMGroot\}$ and as a result, for any set $A\subseteq \AtomicAttacks$ and $C\subseteq \Nodes$ where $\AMGroot \in C$, we have $\simpleState{A,C} = [\emptyset, \{\AMGroot\}]$.}
We can evaluate and optimize an attacker strategy $\strategy{}$ on $\PTMDP_\AMG$.
Indeed, $\strategy{}$ generates a probability measure $\runProba{\PTMDP_\AMG}{ \strategy{}}$ on subsets of $\Runs_0$ (c.f. Section~\ref{sec:background}). 
We define $\RunsGoal$ the subset of $\Runs$ s.t. the final location is the goal node $\PTMDPgoal$, i.e., $\RunsGoal = \{(q_i, e_i, t_i, c_i, q_i')_{i \in \{1,\dots,k\}} \in \Runs^k \mid k \in \mathbb{N} \land \exists v \in \Valuations, q_k = (\PTMDPgoal, v)\}$ and two random variables $\cumulativeTimeGoal$ and $\cumulativeCostGoal$ giving the attack time and attack cost in the following way. For a run $r\in \Runs_0$, \shrinkalt{
and $r_1$ the smallest run  (if exists) in $\RunsGoal$ s.t. $r=r_1\cdot r_2$ with $r_2\in \Runs$, $\cumulativeTimeGoal(r) = \cumulativeTime{r_1}$ (resp. $\cumulativeCostGoal(r) = \cumulativeCost{r_1}$) if $r_1$ exists or $\cumulativeTimeGoal(r) = \infty$ (resp. $\cumulativeCostGoal(r) = \infty$).
}{
\begin{align*}
    \cumulativeTimeGoal(r) &= \min_{\substack{r_1 \in \RunsGoal, r_2\in \Runs\\\text{s.t.}\ r = r_1\cdot r_2}}\cumulativeTime{r_1}\\
    \cumulativeCostGoal(r) &= \min_{\substack{r_1 \in \RunsGoal, r_2\in \Runs\\\text{s.t.}\ r = r_1\cdot r_2}}\cumulativeCost{r_1}
\end{align*}
Where the minimum of an empty set is the infinity.
}
Notice that $\runExpect{\PTMDP_\AMG}{ \strategy{}}[\cumulativeTimeGoal]$ (resp. $\runExpect{\PTMDP_\AMG}{ \strategy{}}[\cumulativeCostGoal]$) does not exist if $\runProba{\PTMDP_\AMG}{ \strategy{}}[\cumulativeTimeGoal = \infty] > 0$ (resp. $\runProba{\PTMDP_\AMG}{ \strategy{}}[\cumulativeCostGoal = \infty] > 0$) and by convention if $\runProba{\PTMDP_\AMG}{ \strategy{}}[\cumulativeTimeGoal = \infty] = 0$ (resp. $\runProba{\PTMDP_\AMG}{ \strategy{}}[\cumulativeCostGoal = \infty] = 0$) we consider that $\infty \times 0 = 0$ in the computation of the expected value.

Suppose the defender prefers distributions according to their expected values (this is not the only way to compare distributions c.f.~\cite{rass2015game}, maybe the defender wants a very low variance). Then, the most dangerous attacker would have strategies minimizing the expected values $\runExpect{\PTMDP_\AMG}{\strategy}[\cumulativeTimeGoal]$ and $\runExpect{\PTMDP_\AMG}{\strategy}[\cumulativeCostGoal]$. These optimal points draw the Pareto frontier, that is, the set of points s.t. decreasing the expected attack time (resp. cost) would increase the expected attack cost (resp. time). As a result, we are interested in computing the Pareto frontier to see the impact of the defenses.

}{
\section{Using the \glsfmtshort{ptmdp} for \glsfmtshort{mtd}}\label{sec:using}

}
%\subsection{Prove equivalence (?)}
%\input{equivalence}
%\subsection{Complexity}
%\input{complexity}

%\section{Experiment}\label{sec:experiment}
%\textsc{Uppaal Stratego}~\cite{david2015uppaal} can be used to compute strategies to solve a cost/time-bounded reachability objective with near-optimal cost or time (separately). We had to make some adjustments, described in the appendix, to express the \gls{ptmdp} as a \textsc{Uppaal} structure.
%\subsection{Translation in \textsc{Uppaal Stratego}}
%\input{tex/translation}
%\subsection{Use Case}
%\input{tex/usecase}

%\section{Related Work}
%\input{tex/related-work}

%\section{Conclusion}\label{sec:conclusion}
%\input{tex/conclusion}

\section{Experiment and Discussion}\label{sec:experiment}
\textsc{Uppaal Stratego}~\cite{david2015uppaal} can be used to compute strategies to solve a cost/time-bounded reachability objective with near-optimal cost or time (separately). We had to make some adjustments, described \shrinkalt{in~\cite{amg}}{in Section~\ref{sec:translation}}, to express the \gls{ptmdp} as a \textsc{Uppaal} structure.
\shrinkalt{}{
We present the use case in Section~\ref{sec:usecase} and discuss the results in Section~\ref{sec:discussion}.
\subsection{Translation in \textsc{Uppaal Stratego}}\label{sec:translation}
\begin{figure}
    \centering
    \subfloat[ADT]{
    \begin{tikzpicture}[
            -,
            auto,
            semithick,
            xscale=0.6,
            yscale=0.6
        ]
        
        \node[goal, label=below:{$\lor$}] at (0, 0) (g0) {$\AMGroot$};
        \node[attack, label={}] at (-2, -2) (a0) {$a_0$};
        \node[defense, label={}] at (-3, -1) (d0) {$d_0$};
        \node[attack, label={}] at (2, -2) (a1) {$a_1$};
        
        \path (g0) edge (a0) edge (a1);
        \path (d0) edge [dashed] (a0);
        
        \matrix [draw, column sep=5pt, row sep=2pt, inner sep=2pt, ampersand replacement=\&, column 1/.style={anchor=base}, column 2/.style={anchor=base west}] at (7, -1) {
        \node [goal, minimum size=1em] {\phantom{a}}; \& \node[scale=0.8] {subgoal};  \\
        \node [attack, minimum size=1.3em] {\phantom{a}}; \& \node[scale=0.8] {atomic attack}; \\
        \node [defense, minimum size=1em] {\phantom{a}}; \& \node[scale=0.8] {\glsfmtshort{mtd}}; \\
        };
    \end{tikzpicture}
    \label{fig:simple-adt}
    }\qquad
    \subfloat[Attributes]{
    \begin{tabular}{|c|c|c|c|}
        \hline
        $\timefunc_{a_0}$ & $\proba_{a_0}$ & $\cost_{a_0}$ & $\propcost_{a_0}$ \\
        \hline
        20 & 1 & 10 & 0 \\
        \hline
        \hline
        $\timefunc_{a_1}$ & $\proba_{a_1}$ & $\cost_{a_1}$ & $\propcost_{a_1}$ \\
        \hline
        10 & 0.5 & 0 & 2 \\
        \hline
    \end{tabular}
    \begin{tabular}{|c|c|}
        \hline
        $\timefunc_{d_0}$ & $\proba_{d_0}$\\
        \hline
        10 & 1 \\
        \hline
        \hline
        $\operation{g_0}$ & $\defense{g_0}$\\
        \hline
        $\lor$ & $a_0$ \\
        \hline
    \end{tabular}
    \label{tab:simple-adt}
    }
    
    \cprotect\caption[Simple \glsfmtshort{amg} with its attributes.]{Simple \gls{amg} in \figurename~\subref{fig:simple-adt} and its attributes in \tablename~\subref{tab:simple-adt} that are used for the \gls{ptmdp} in \figurename~\ref{fig:ex-ADMDP}.}
    \label{fig:ex-ADMDP-attributes}
\end{figure}

\begin{figure}
    \includegraphics[width=\columnwidth]{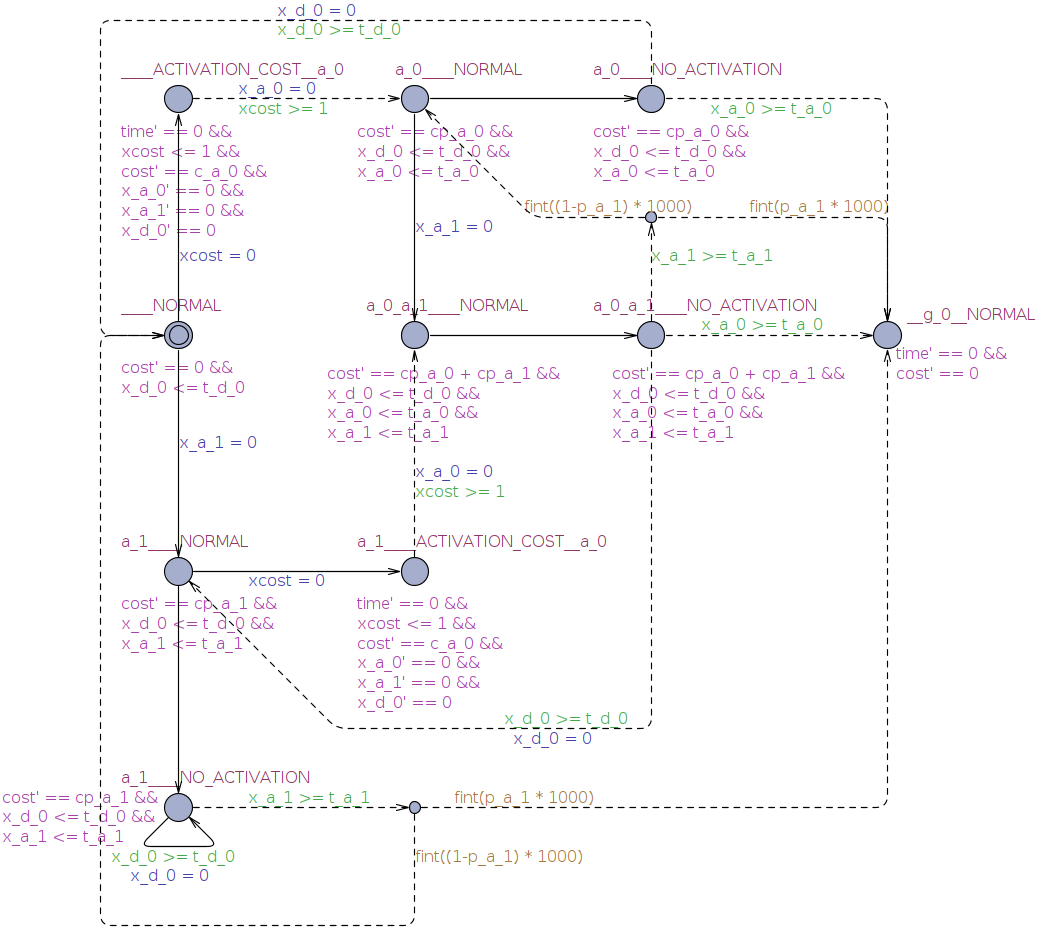}
    \cprotect\caption[\glsfmtshort{ptmdp} obtained from the simple \glsfmtshort{amg} in \figurename~\ref{fig:ex-ADMDP-attributes}.]{\gls{ptmdp} $\PTMDP_\AMG$ obtained from the simple \gls{amg} $\AMG$ in \figurename~\ref{fig:simple-adt} with the attributes given in Table~\ref{tab:simple-adt}. The location names are given by a string of the form \verb+<A>__<C>__<TYPE>+ where \verb+A+ is the set of activated nodes \verb+C+ is the set of completed nodes and \verb+TYPE+ is the type of node which can be \verb+ACTIVATION_COST_<a>+ for the activation of an atomic attack where \verb+a+ is the newly activated atomic attack, \verb+NO_ACTIVATION+ for the location where no more activation are permitted, or \verb+NORMAL+ for the location where activations are still permitted. Starting from the initial location \verb+____NORMAL+ (location with a circle in UPPAAL), the upward transition leads to the activation of $a_0$ with an ``activation cost'' node (sice $\cost_{a_0} \neq 0$), and the downward transition leads directly to the activation of $a_1$ (since $\cost_{a_1}=0$). The location \verb+a_0____NORMAL+ is reached from \verb+____ACTIVATION_COST_a_0+ after one unit of time (thanks to the constraint \verb+xcost >= 1+ and the invariant \verb+xcost <= 1+) where the cost has been increased by \verb+w_a_0+. From \verb+a_0____NORMAL+, the attacker can decide to activate $a_1$ as well (downward transition) or not (right transition). For instance, if the attacker does not activate $a_1$ (right transition), then when $x_{a_0} \geq \timefunc_{a_0}$ the transition is taken, and the new state is the rightmost state \verb+__g_0__NORMAL+, that is the final location $\AMGroot$.}
    \label{fig:ex-ADMDP}
\end{figure}

We had to make some adjustments to express the \gls{ptmdp} as a \textsc{Uppaal} structure. 
First, \textsc{Uppaal Stratego}'s strategy space for the controllable player is limited to memoryless non-lazy strategies. This means that the strategy depends only on the current state and does not delay a possible transition. As a result, the action of activating no more atomic attacks should be a transition. So for each location $\loc\in \Locations$, we had the extra location $\loc_\text{no activation}$ and a transition without cost, guard, or clock reset $(\loc,\varepsilon, \noactivation,\loc_\text{no activation})$. The \gls{mtd} activation edges and the atomic attack completion edges start from these ``no activation'' locations.

%We used \textsc{Uppaal Stratego}~\cite{david2015uppaal} to exploit the ADMDP.

Second, \textsc{Uppaal} uses clocks with a dynamic rate called \textit{hybrid clocks}. We will use a hybrid clock (\verb+cost+) to compute the cost. However, these hybrid clocks cannot be incremented (as in \verb|cost = cost + c_a|) while it is supposed to happen when activating an atomic attack with cost \verb+c_a+.
Instead, we create a new intermediary state for each activation transition, say $(\loc, \varepsilon, \activate_a, \loc')$, where we stay exactly one unit of time (this is achieved with a new clock \verb+xcost+). The cost rate of this intermediary state is equal to the cost of activation, that is \verb|cost' == c_a|, and all other clocks rates (atomic attacks clocks, \glspl{mtd} clocks, time clock) are set to zero.

Third, we need to express the stochasticity of the environment transitions $\density{}$ with the semantic of \textsc{Uppaal Stratego}. We remind that eq.~\eqref{eq:mu1} to~\eqref{eq:mu4} define the density $\density{(\loc, v)}$ for a given location $\loc\in\Locations$ and valuation $v\in\Valuations$. As the density is null before reaching the deadline of an \gls{mtd} $d$ activation (\textit{resp.} an atomic attack $a$ completion), we add the guard $x_d \geq \timefunc_d$ (\textit{resp.} $x_a\geq \timefunc_a$) to the outgoing \gls{mtd} activation transition (\textit{resp.} atomic attack completion transition). %For each \gls{mtd} transition $\mtd_d$ linked to a \gls{mtd} $d$, we need to add the guard $x_{d'} < t_{d'}$ for all \gls{mtd} $d'$ that defend a child of a subgoal defended by $d$ to force \glspl{mtd} on deeper nodes to get activated first.
If several \gls{mtd} activations or atomic attack completion transitions (in $\Transitions^\mathsf{mtd}\cup \Transitions^\mathsf{cmp}$) are activated simultaneously, \textsc{Uppaal Stratego} assumes that one of them is chosen uniformly. this corresponds to the use of $\coefmu{l}{\loc,v}$ in eq.~\eqref{eq:mu1} to~\eqref{eq:mu4}. Then we use a \textsc{Uppaal} branchpoint to succeed the \gls{mtd} $d$ activation with probability $\proba_d$ and fail with probability $1-\proba_d$, and to succeed the atomic attack $a$ completion with probability $\proba_a$ and fail with probability $1-\proba_a$.

Fourth, we have to force \textsc{Uppaal} to take a defense activation transition for an \gls{mtd} $d_1$ only if there is no \gls{mtd} $d_2$ s.t $d_1\triangleright d_2$ that is activated at the same time. To encode this, for every locations $\loc, \loc' \in \Locations$, defense $d_1$, and activation transitions $(\loc, x_{d_1} \geq \timefunc_{d_1}, \mtd_{d_1}, \loc')$ and $(\loc, x_{d_1} \geq \timefunc_{d_1}, \mtdFail_{d_1}, \loc')$, we add the clock constraint $x_{d_2} < \timefunc_{d_2}$ whenever the \gls{mtd} $d_2$ verifies $d_1 \triangleright d_2$. As an example, if there is only one \gls{mtd}, say $d_2$, s.t. $d_1 \triangleright d_2$, then the transitions become $(\loc, x_{d_1} \geq \timefunc_{d_1} \land x_{d_2} < \timefunc_{d_2}, \mtd_{d_1}, \loc')$ and $(\loc, x_{d_1} \geq \timefunc_{d_1} \land x_{d_2} < \timefunc_{d_2}, \mtdFail_{d_1}, \loc')$.

\figurename~\ref{fig:ex-ADMDP} displays the automatic translation of the \gls{amg} in \figurename~\ref{fig:simple-adt} with the attributes of \tablename~\ref{tab:simple-adt}.
\mysubsection{Use Case}\label{sec:usecase}}
\begin{table}
    \centering
    \small
    \begin{tabular}{c|ccccccc|cccc|}
        \cline{2-12}
         & $a_{ad}$ & $a_{ic}$ & $a_{sp}$ & $a_{p}$ & $a_{bf}$ & $a_{ss}$ & $a_{fue}$ & $d_{dk}$ & $d_{cp}$ & $d_{cc}$ & $d_{dsr}$\\
        \hline
        \multicolumn{1}{|c|}{$\timefunc$}      & 8   & 4   & 440 & 1   & 1     & 30  & 720 & \multicolumn{4}{c|}{to optimize}\\
        \multicolumn{1}{|c|}{$\proba$}      & 0.5 & 0.3 & 0.8 & 1   & 0.001 & 0.2 & 0.8 & 1 & 0.5 & 1 & 1\\
        \multicolumn{1}{|c|}{$\cost$}  & 10  & 0   & 20  & 0   & 0     & 10  & 10  &&&&\\
        \multicolumn{1}{|c|}{$\propcost$} & 20  & 5 0 & 0   & 100 & 1     & 0   & 0   &&&&\\
        \hline
    \end{tabular}
    \caption[Attributes for the atomic attacks and \glsfmtshortpl{mtd} from the \glsfmtshort{amg} in \figurename~\ref{fig:example-adt}.]{Attributes for the atomic attacks and \glspl{mtd} from the \gls{amg} in \figurename~\ref{fig:example-adt}.}
    \label{tab:values-attributes}
\end{table}

We implement the \gls{amg} $\AMG$ in \figurename~\ref{fig:example-adt} with the attributes given in \tablename~\ref{tab:values-attributes} and translate it into a \textsc{Uppaal Stratego} model. 
We aim to draw the Pareto frontier of optimal expected attack time and cost. However, \textsc{Uppaal Stratego} can only find the \shrinkalt{}{memoryless non-lazy} strategies minimizing the following conditional expected values \shrinkalt{
$\mathbb{E}_{\PTMDP_\AMG, \strategy}[\cumulativeTimeGoal \mid \cumulativeTimeGoal < \timefunc_{\max}]$, $\mathbb{E}_{\PTMDP_\AMG, \strategy}[\cumulativeCostGoal \mid \cumulativeTimeGoal < \timefunc_{\max}]$, $\mathbb{E}_{\PTMDP_\AMG, \strategy}[\cumulativeTimeGoal \mid \cumulativeCostGoal < \cost_{\max}]$, and $\mathbb{E}_{\PTMDP_\AMG, \strategy}[\cumulativeCostGoal \mid \cumulativeCostGoal < \cost_{\max}]$
}{
\begin{align*}
    &\mathbb{E}_{\PTMDP_\AMG, \strategy}[\cumulativeTimeGoal \mid \cumulativeTimeGoal < \timefunc_{\max}]\\
    &\mathbb{E}_{\PTMDP_\AMG, \strategy}[\cumulativeCostGoal \mid \cumulativeTimeGoal < \timefunc_{\max}]\\
    &\mathbb{E}_{\PTMDP_\AMG, \strategy}[\cumulativeTimeGoal \mid \cumulativeCostGoal < \cost_{\max}]\\
    &\mathbb{E}_{\PTMDP_\AMG, \strategy}[\cumulativeCostGoal \mid \cumulativeCostGoal < \cost_{\max}]
\end{align*}
}
with time limit $\timefunc_{\max}$ and cost limit $\cost_{\max}$~\cite{david2014time}.

We vary these limits to explore the different minimizing strategies. The minimal value, say $\mathbb{E}_{\PTMDP_\AMG, \strategy}[\cumulativeTimeGoal \mid \cumulativeTimeGoal < \timefunc_{\max}]$ in the first case, is not useful if we are not provided the probability of the associated condition, here $\mathbb{P}_{\PTMDP_\AMG, \strategy}[\cumulativeTimeGoal < \timefunc_{\max}]$. Indeed, If $\mathbb{E}_{\PTMDP_\AMG, \strategy}[\cumulativeTimeGoal \mid \cumulativeTimeGoal < \timefunc_{\max}]$ is one hour, we could think that there is a major attack path. But if the associated probability $\mathbb{P}_{\PTMDP_\AMG, \strategy}[\cumulativeTimeGoal < \timefunc_{\max}]$ is very low, then this attack is very unlikely to succeed. \shrinkalt{}{For instance, the attacker can break a system in one minute if he guesses the admin password at the first try, but this is very unlikely to happen.}
Consequently, reasoning with the conditional probabilities, we should draw a Pareto surface in the three-dimension space of conditional expected time, conditional expected cost, and probability of the condition.
\shrinkalt{}{This Pareto surface contains more information than the two-dimension Pareto frontier of expected time and cost. Indeed, the 2D frontier is the cut of the surface for a conditional probability axis equal to one.}
However, we will not reason on the 3D surface because we do not control the probability of the condition in \textsc{Uppaal Stratego} strategy optimization and we would need exponentially more points to draw the surface instead of the frontier. To simplify, we assume that a strategy minimizing the conditional expected value might be a strategy giving non-conditional expected values close to the expected cost/time Pareto frontier. This is a strong assumption, and finding a better optimization method is a necessary future work.
\shrinkalt{}{By varying the time and cost bounds we extract optimal strategies and plot their unconditional expected time and cost (\figurename~\ref{subfig:result-one-conf}). Repeating this procedure for different \gls{mtd} activation frequencies we can compare the different Pareto frontiers (\figurename~\ref{subfig:results})\footnote{To reproduce the experiment: \url{https://github.com/gballot/mtd}.}.}

\shrinkalt{
\begin{figure}[t]
    \centering
    \includegraphics[width=0.95\textwidth]{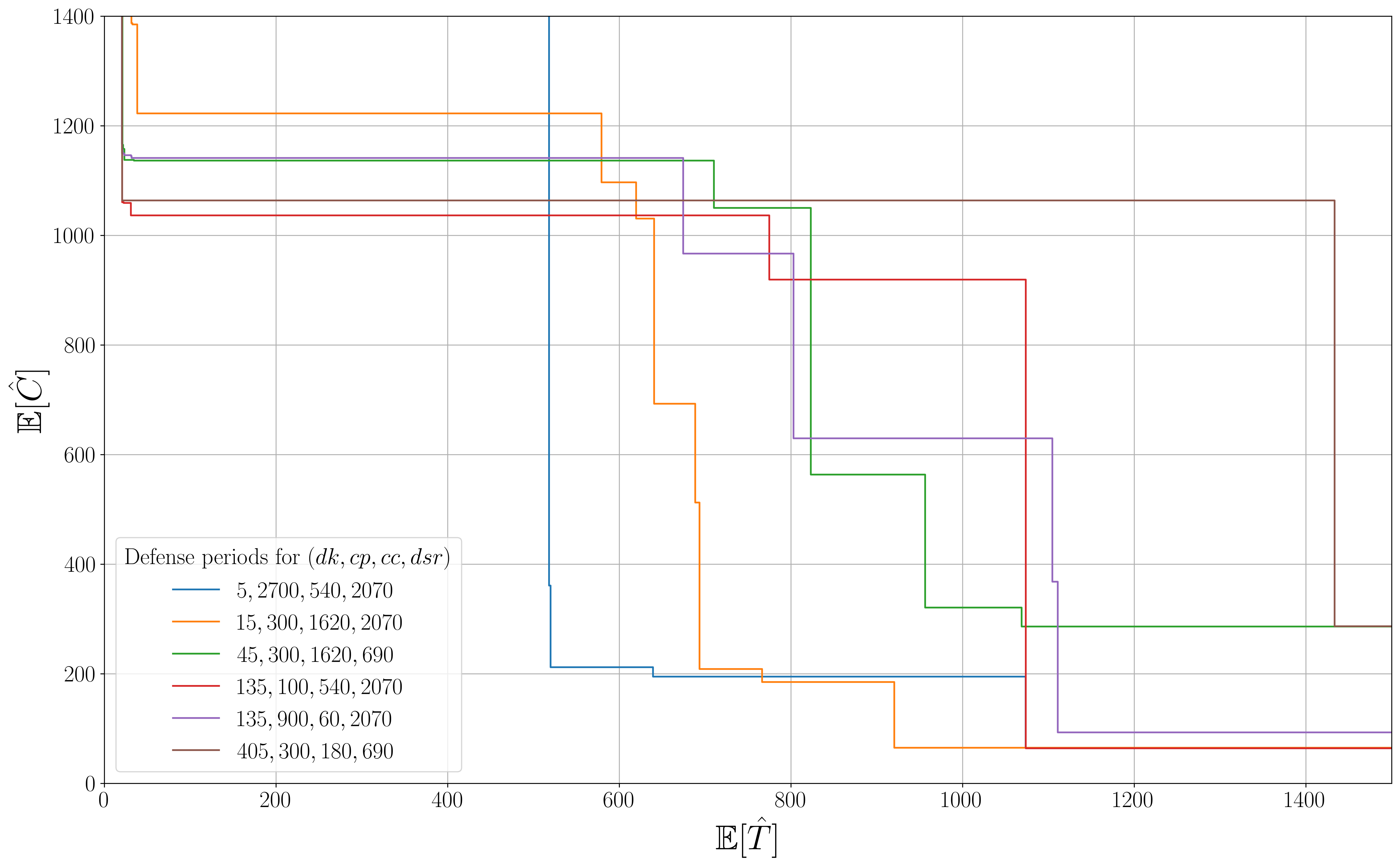}
    \cprotect\caption[Results from the \glsfmtshort{amg} of \figurename~\ref{fig:example-adt} with attributes in \tablename~\ref{tab:values-attributes}.]{We consider the \gls{amg} of \figurename~\ref{fig:example-adt} with attributes in \tablename~\ref{tab:values-attributes}. We impose \shrinkalt{$\log(t_{d_{dk}} t_{d_{cp}} t_{d_{cc}} t_{d_{dsr}})$ to be constant for the different defensive configurations}{the defense periods to verify $\log_3(t_{d_{dk}}/5) + \log_3(t_{d_{cp}}/100) + \log_3(t_{d_{cc}}/20) + \log_3(t_{d_{dsr}}/230) = 8$} to simulate a defensive budget.
    }
    \label{fig:results}
\end{figure}

}{
\begin{figure}
    \centering
    \subfloat[The Pareto frontier for a specific defensive configuration.]{
    \includegraphics[width=0.95\textwidth]{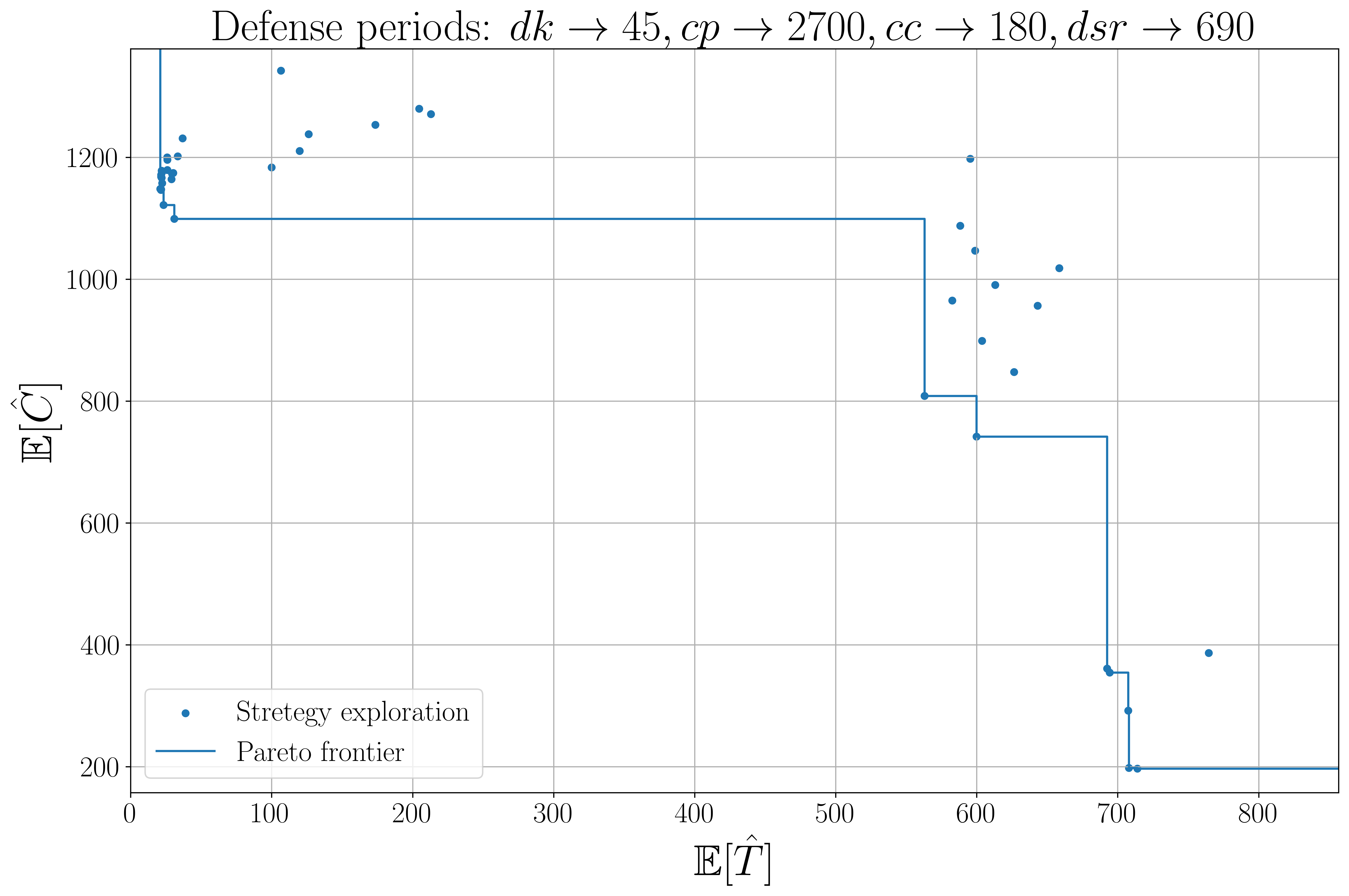}
    \label{subfig:result-one-conf}
    }
    
    \subfloat[Different Pareto frontiers for a defensive budget of $8$.]{
    \includegraphics[width=0.95\textwidth]{images/fig-some.png}
    \label{subfig:results}
    }
    \cprotect\caption[Results from the \glsfmtshort{amg} of \figurename~\ref{fig:example-adt} with attributes in \tablename~\ref{tab:values-attributes}.]{We consider the \gls{amg} of \figurename~\ref{fig:example-adt} with attributes in \tablename~\ref{tab:values-attributes}. In~\subref{subfig:results}, we impose the defense periods to verify $\log_3(t_{d_{dk}}/5) + \log_3(t_{d_{cp}}/100) + \log_3(t_{d_{cc}}/20) + \log_3(t_{d_{dsr}}/230) = 8$ to simulate a defensive budget.
    %is assumed a cost per time unit of $\frac{1}{5}$ for $t_{d_{dk}}$, $\frac{1}{100}$ for $t_{d_{cp}}$, $\frac{1}{200}$ for $t_{d_{cc}}$, and $\frac{1}{250}$ for $t_{d_{dsr}}$. There is no configuration that is uniformly better than the others. \toremove{It is up to the user to choose the MTD configurations that fit his needs best. For instance, if the system is worth less than 8000 cost units for the attacker and the attack can last long, then the cyan or yellow configurations are good.}
    }
    \label{fig:results}
\end{figure}

}

\shrinkalt{}{\mysubsection{Discussion}\label{sec:discussion}}
\shrinkalt{}{First we discusse the result of our methodology on the use case. Then we discuss more generally the output of this paper.

\subsection{On the Use Case}}

\shrinkalt{We}{The result of the use case is given in \figurename~\ref{subfig:result-one-conf} for a particular defensive configuration.
We also} report\footnote{To reproduce the experiment: \url{https://github.com/gballot/mtd}.} the Pareto frontiers for different sets of \gls{mtd} activation periods in \figurename~\shrinkalt{\ref{fig:results}}{\ref{subfig:results}}.
Reasoning about the \gls{amg} in \figurename~\ref{fig:example-adt} and the attributes in Table~\ref{tab:values-attributes}, we notice a fast and costly attack with the atomic attack $a_{ad}$ and defended by the \gls{mtd} $d_{dk}$, a medium-fast medium-costly attack with the atomic attacks $a_{sp}$, $a_{p}$, and the subgoal $g_{ac}$ defended by the \glspl{mtd}$d_{cp}$ and $d_{cc}$, and a long cheap attack with the atomic attack $a_{fue}$ and the subgoal $g_{ac}$ defended by the \gls{mtd} $d_{dsr}$. As expected, the frontiers with small period for $d_{dk}$ ($\timefunc_{d_{dk}} = 5$) limit the attack time to more than 500 time units even with unlimited cost (blue line). Furthermore, small period for $d_{cp}$ ($\timefunc_{d_{cp}}\leq 300$) is efficient in increasing the cost of long attacks to more than 200 cost units (frontiers in green and brown) provided that the cheap attack path is protected with $\timefunc_{d_{dsr}} < \timefunc_{a_{fue}} = 720$ (otherwise, we have the orange or red frontier with low cost for long attacks). We also notice that the \gls{mtd} $d_{cc}$ influences the cost of long attacks even when $\timefunc_{d_{dsr}}$ is high (purple line), but this influence is only about $40$ cost units for $\timefunc_{d_{cc}} = 60$.

This example is simple \shrinkalt{as}{and could be solved by hand as} the expected attack cost and time are increasing with each \gls{mtd} activation frequency. However, in more complex systems, this is not true. For example, when different \glspl{mtd} defend parent and child nodes, it could be better to have the same frequency for two \glspl{mtd} (so they are coupled) than having one \gls{mtd} slightly more frequent. This justifies that optimizing $(\timefunc_d)_{d\in \Defenses}$ cannot be component by component in the general case and need powerful tools like \glspl{ptmdp}. 

\shrinkalt{}{
\subsection{General Discussion}

The \gls{amg} is a suitable formalism for graphical modeling of attacks as it relies on the well-spread \gls{at} formalism. The attack scenario is modeled hierarchically, and multi-step complex attacks are expressed. It is convenient to attach \glspl{mtd} to nodes of the graph. The interpretation of the \gls{amg} as a \gls{ptmdp} reflects the stochasticity, cost, and time dependency of \glspl{mtd}. It allows to optimize the attacker strategy given a defensive configuration and to decide which one is best suited to defend the system. The simple use case points out the applicability and relevance of our method. The Pareto frontiers we extract from the \gls{amg} give a good insight into the impacts of the defender strategy (\ie the set of \gls{mtd} activation frequencies) on the best attacker strategies.

\subsubsection*{\glsfmtshort{amg} limitations.}

The \gls{amg} suffers from some limitations:
\begin{inparaenum}[(i)]
    \item \label{item:limit-identify} the \gls{amg} assumes that the user can identify the attacks and defenses and their attributes (probability, cost, and time),
    \item \label{item:limit-adt} nodes are defended by disjunctions of \glspl{mtd}, but we could nest the countermeasures and use be conjunctions as in \gls{adt},
    \item \label{item:limit-distribution} we only consider success or failure after a given time rather than general distributions.
\end{inparaenum}

To address the limitation~(\ref{item:limit-identify}) we could test the parameter robustness of the expected time and cost to see if small parameter changes induce a big difference in the computed values. The items~(\ref{item:limit-adt}) and~(\ref{item:limit-distribution}) are left for future work. 

\subsubsection*{\glsfmtshort{ptmdp} translation and optimization limitations.}

Moreover, the current method for \gls{mtd} activation frequency optimization has other limitations:
\begin{inparaenum}[(i)]
    \item \label{item:limit-uppaal}\textsc{Uppaal Stratego} solves limited types of objectives, leading us to make too strong assumptions about the problem (\cf Section~\ref{sec:usecase}),
    \item \label{item:limit-size} it does not scale to much larger problems due to the exponential size of the \gls{ptmdp} compared to the \gls{amg}, and
    \item \label{item:limit-mtds} as the number of \glspl{mtd} increases, the number of Pareto frontiers to analyze grows exponentially, leading the user to confusion.
\end{inparaenum}

To address the limitation~(\ref{item:limit-uppaal}), we will think about a better optimization process specific to the problem we want to achieve. For limitation~(\ref{item:limit-size}), we have to improve the construction and maybe limit the aspects we are dealing with (time, cost, and probability). For limitation~(\ref{item:limit-mtds}),  we could set design conditions (\eg a desired minimal Pareto frontier) so the optimizer only displays the configurations satisfying this minimal requirement.
}

\section{Related Work}\label{sec:RW}
A well-established formalism for \gls{at} with defenses is the \glsxtrfull{adt}~\cite{kordy2010foundations,kordy2014attack}. \gls{amg} is different from \gls{adt}.
On the one hand, the \gls{amg} restricts \gls{adt} because an \gls{adt} node can have disjunction and conjunction of countermeasures, which can be nested.
On the other hand, the \gls{amg} extends the \gls{adt} in two ways. First, there is an attribute on the inner 
\shrinkalt{nodes.}{nodes (the $\defense{g}$ for the subgoals $g \in \Goals$), that can also be seen as adding new refinements, namely $\{\land_D, \lor_D \mid D\subseteq \Defenses\}$.}
Second, the \gls{amg} allows a \gls{dag} structure for the nodes and the defenses. Notice that \gls{adt} can have the same label on different nodes, so it is as expressive as a \gls{dag} for some semantics (that is the case for the propositional semantic, for instance).
Other formalisms derive from \gls{adt} (see the surveys~\cite{kordy2014dag,widel2019beyond}). In particular, in~\cite{hermanns2016value}, Hermanns \etal define \textit{Attack Defense Diagram}, which is more expressive than most attack-defense formalism but does not explicitly model \glspl{mtd}. Moreover, security engineers may find our model best balanced between expressivity and ease of use, mainly thanks to our tool for strategy optimization with \textsc{Uppaal Stratego}. In~\cite{hansen2021adtlang}, Hansen \etal come with the comprehensive tool support for modeling \gls{adt} extended with dynamic defender policies and atomic attack expiring. However, atomic attack expiry dates are relative to their activation date and not to defenses, making them unsuitable for \glspl{mtd}.
In~\cite{arnold2014time}, the authors consider both time and stochasticity in an \gls{at} whose basic actions have the \gls{cdf} of the completion of atomic attacks. The \gls{cdf} is propagated to the parents to get the \gls{cdf} for the whole tree. This method does not use automata but directly computes the \gls{cdf} through an alternative representation of the \gls{cdf} called \textit{acyclic phase-type distribution}.

Our work combines \gls{dag}-based attack-defense modeling for defense optimization and \gls{mtd} activation frequency optimization. These two aspects have been studied separately in the following papers.
In~\cite{kumar2015quantitative}, the authors translate an \gls{at} into a network of \glsxtrfull{pta}, thanks to a \gls{pta} interpretation for each node of the tree. They can then use \textsc{Uppaal Cora} to uncover the best attack path regarding costs and time. In~\cite{gadyatskaya2016modelling,hansen2017quantitative}, the authors consider the \gls{adt} to construct a network of \gls{pta} and analyze the impact of enabling different defenses on the best attack. These papers do not use \textsc{Uppaal Stratego}, and for that reason, they need to iterate on faster and faster attacks to get the fastest one (resp. iterate on cheaper to get the cheapest). Instead, in our analysis, \textsc{Uppaal Stratego} optimizes the strategy for the attacker directly. Moreover, it does not apply to time-based defenses like \glspl{mtd}.
In~\cite{ayrault2021moving,li2019optimal,feng2017stackelberg}, the authors study the optimal activation frequencies for \glspl{mtd} with a game theoretic approach. They model the attacker and the defender with a Stackelberg game (the defender plays first, and the attacker plays the rest of the game). However, they only consider single step attacks. The authors of~\cite{ayrault2021moving} can formulate the game equilibrium and compute the optimal parameters for the defender directly and the authors of~\cite{li2019optimal,feng2017stackelberg} derive a semi-Markovian decision process from the game to optimize the activation frequencies of the \glspl{mtd}. Many other papers deal with \gls{mtd} with a game theoretic approach including\cite{umsonst2021bayesian,sengupta2017game,clark2015game}. The authors of~\cite{umsonst2021bayesian} considers \glspl{mtd} against stealthy sensor attacks and derive a Bayesian game to extract optimal \gls{mtd} strategy even with only the prior of the
possible attacker goals. The paper~\cite{sengupta2017game} focus on web applications, and~\cite{clark2015game} focuses on IP address randomization.

\section{Conclusion and Future Work}\label{sec:conclusion}
In this paper, we introduced the \gls{amg}, a \gls{dag}-based attack-defense model that considers the time, cost, and stochastic properties of \glspl{mtd} and attacks. This new model permits to hierarchically model threats on complex systems defended with \glspl{mtd}.
We constructed a \gls{ptmdp} from this \gls{amg} that induces a probability measure on the sets of runs. Thanks to this measure, we define a reachability objective with time and cost constraints and present the optimization problem for the attacker's strategy. We can then find the \gls{mtd} activation frequencies that will protect our system the best according to the user preferences. We implemented the automatic construction of the \gls{ptmdp} from the \gls{amg} and used \textsc{Uppaal Stratego} to illustrate the applicability of the optimal strategy computation in a use case. It displayed the influence of four \glspl{mtd} on an electricity meter on the best attacker's strategy in a two-dimension optimization of attack time and cost.

We plan to explore the dependency between the defense activation frequencies to find a way to optimize them in future work. We should consider each aspects (time, cost, probability) independently to deal with them in a non-exponential way. We also plan to extend the \gls{amg} to include the full \gls{adt} expressivity and show how to consider non-\gls{mtd} defense in a broader formalism. Finally, we consider implementing our own tool to find the strategies giving the Pareto frontier of the attack cost and attack time.

\bibliographystyle{splncs04}
\bibliography{references}

\ifdefined\ALLPROOFS
\else % ALLPRROFS not defined
\appendix
\newpage
\section{Appendix: Proof}
This appendix reminds and proves the properties, lemmas, and theorem of this thesis, with the help of new Lemmas~\ref{lemma:subtree-subtree} and~\ref{lemma:subtree-simple-state}.

%\chapter{Proofs}\label{ch:proofs}
\subsection{Propagation operator}

We remind and prove Proposition~\ref{prop:propagate}.
\propagateprop*
\begin{proof}
For~(\ref{item:propagate-contains}), we need to notice that the $\fixedPoint{}$ from Definition~\ref{def:propagate} always contains its argument, and we use the fact that $\propagate{\AMG}{C} = \fixedPoint{}^k(C)$ for all $C \subseteq \Nodes$ and some $k\in \mathbb N$.

For~(\ref{item:propagate-increase}), we use the fact that $\fixedPoint{}$ is increasing and reaches a fixed point. Let $C_1\subseteq C_2 \subseteq \Nodes$. We can take composistion indices $k_1, k_2\in \mathbb N$ such that $\propagate{\AMG}{C_1} = \fixedPoint{}^{k_1}(C_1)$ and $\propagate{\AMG}{C_2} = \fixedPoint{}^{k_2}(C_2)$. Now,
\begin{align*}
    \propagate{\AMG}{C_1} &= \fixedPoint{}^{k_1}(C_1)\\
    &=\fixedPoint{}^{\max(k_1, k_2)}(C_1)\\
    &\subseteq \fixedPoint{}^{\max(k_1, k_2)}(C_2)\\
    &= \propagate{\AMG}{C_2}
\end{align*}

For~(\ref{item:propagate-projection}), by the Definition~\ref{def:propagate}, $\propagateDef{}$ is a fixed point.
\end{proof}

\subsection{Completed descendants}

We remind and prove Proposition~\ref{prop:subtree-inc}.
\subtreeprop*
\begin{proof}
Let $A\subseteq B \subseteq \Nodes$ if there is $j \in \{1, \dots, k\}$ s.t. $g_{j} \in A$, then $g_{j} \in B$. This proves $\subTree{\AMG}{A} \subseteq \subTree{\AMG}{B}$.
\end{proof}

We introduce Lemma~\ref{lemma:subtree-subtree}, showing that some nodes can be ignored in the completed descendants.
\begin{restatable}{lemma}{subtreelemma}\label{lemma:subtree-subtree}
For an AMG $\AMGgeneral$, and $C\subseteq \Nodes$, we have $\subTree{\AMG}{C\setminus \subTree{\AMG}{C}} = \subTree{\AMG}{C}$.
\end{restatable}
\begin{proof}
By Proposition~\ref{prop:subtree-inc} we have $\subTree{\AMG}{C\setminus \subTree{\AMG}{C}} \subseteq \subTree{\AMG}{C}$, now let $n \in \subTree{\AMG}{C}$ we want to prove $n \in \subTree{\AMG}{C\setminus \subTree{\AMG}{C}}$. By the Definition~\ref{def:comp-desc}, for any directed path $g_{1}, \dots, g_{k} \in \Nodes$ from $\AMGroot$ to a parent of $n$ there is a $j \in \{1,\dots, k\}$ s.t. $g_{j} \in C$. For any such directed path, let $j$ be the smallest integer such that $g_{j} \in C$. If $g_{j} \in \subTree{\AMG}{C}$, then $j > 1$ as $\AMGroot \not\in \subTree{\AMG}{C}$ ($\AMGroot$ never has completed parents, it does not have parents) and the directed path $g_{1}, \dots, g_{j-1}$ exists and must contain a node in $C$, say $g_{l}$ with $l\in \{1,\dots, j-1\}$. This is a contradiction because $l < j$ and $g_{l} \in C$ while $j$ was chosen the smallest. So $g_{j} \not\in \subTree{\AMG}{C}$, and thus $g_{j} \in C\setminus \subTree{\AMG}{C}$. As a result, every directed path from $\AMGroot$ to a parent of $n$ contains an element in $C\setminus \subTree{\AMG}{C}$, i.e. $n\in \subTree{\AMG}{C\setminus \subTree{\AMG}{C}}$.
\end{proof}

\subsection{Simple state}

To demonstrate Proposition~\ref{prop:simple-state-proj}, we first need Lemma~\ref{lemma:subtree-simple-state}.
\begin{lemma}\label{lemma:subtree-simple-state}
Let $\AMG$ be an AMG, $A_1 \in \AtomicAttacks$ and $C_1\in \Nodes$. Let $(A_2, C_2) = \simpleState{A_1, C_1}$. We have,
\begin{equation*}
    \subTree{\AMG}{\propagate{\AMG}{C_1}\cap \noDef} = \subTree{\AMG}{\propagate{\AMG}{C_2}\cap \noDef}
\end{equation*}
\end{lemma}
\begin{proof}
By definition, $C_2 = \propagate{\AMG}{C_1}\setminus \subTree{\AMG}{\propagate{\AMG}{C_1}\cap \noDef}$ and by Propositions~\ref{prop:propagate}(\ref{item:propagate-increase},~\ref{item:propagate-projection}) and ~\ref{prop:subtree-inc}, we have the first inclusion:
\begin{gather*}
    \propagate{\AMG}{C_2} = \propagate{\AMG}{\propagate{\AMG}{C_1}\setminus \subTree{\AMG}{\propagate{\AMG}{C_1}\cap \noDef}} \subseteq \propagateDef{\AMG} \circ \propagate{\AMG}{C_1} = \propagate{\AMG}{C_1}\\
    \subTree{\AMG}{\propagate{\AMG}{C_2} \cap \noDef} \subseteq \subTree{\AMG}{\propagate{\AMG}{C_1} \cap \noDef}
\end{gather*}

Moreover, as $C_2 = \propagate{\AMG}{C_1}\setminus \subTree{\AMG}{\propagate{\AMG}{C_1}\cap \noDef}$ and by Propositions~\ref{prop:propagate}(\ref{item:propagate-contains}) and~\ref{prop:subtree-inc}, and Lemma~\ref{lemma:subtree-subtree}, we have the second inclusion:
\begin{gather*}
    \propagate{\AMG}{C_1}\setminus \subTree{\AMG}{\propagate{\AMG}{C_1}\cap \noDef} = C_2 \subseteq \propagate{\AMG}{C_2}\\
    \subTree{\AMG}{\propagate{\AMG}{C_1}\cap \noDef} = \subTree{\AMG}{\propagate{\AMG}{C_1}\cap \noDef\setminus \subTree{\AMG}{\propagate{\AMG}{C_1}\cap \noDef}}
    \subseteq \subTree{\AMG}{\propagate{\AMG}{C2}\cap \noDef}
\end{gather*}
\end{proof}

We remind and prove Proposition~\ref{prop:simple-state-proj}.
\simplestateprop*
\begin{proof}
Let $(A_1,C_1) = \simpleState{A,C}$ and $(A_2,C_2) = \simpleState{\simpleState{A,C}}$.

First we want to prove $C_1 = C_2$.
By definition,
\begin{align*}
    C_1 &= \propagate{\AMG}{C}\setminus \subTree{\AMG}{\propagate{\AMG}{C}\cap \noDef}\\
    C_2 &= \propagate{\AMG}{C_1}\setminus \subTree{\AMG}{\propagate{\AMG}{C_1}\cap \noDef}
\end{align*}
Moreover, by Lemma~\ref{lemma:subtree-simple-state}, we have $\subTree{\AMG}{\propagate{\AMG}{C}\cap \noDef} = \subTree{\AMG}{\propagate{\AMG}{C_1}\cap \noDef}$ so we just need to prove $C_1 \subseteq \propagate{\AMG}{C_1}$, that is immediate by Proposition~\ref{prop:propagate}(\ref{item:propagate-contains}), and $C_2\subseteq \propagate{\AMG}{C}$. This last point is true by Proposition~\ref{prop:propagate}(\ref{item:propagate-increase},~\ref{item:propagate-projection}):
\begin{equation*}
    C_2 \subseteq \propagate{\AMG}{C_1} = \propagate{\AMG}{\propagate{\AMG}{C}\setminus \subTree{\AMG}{\propagate{\AMG}{C}\cap \noDef}} \subseteq \propagateDef{\AMG} \circ \propagate{\AMG}{C} = \propagate{\AMG}{C}
\end{equation*}
This finishes the proof of $C_1=C_2$.

Let us prove that $A_1 = A_2$. By definition,
\begin{align}
    \nonumber
    A_1 &= A \setminus (\subTree{\AMG}{\propagate{\AMG}{C}\cap \noDef} \cup \propagate{\AMG}{C})\\
    A_2 &= A_1 \setminus (\subTree{\AMG}{\propagate{\AMG}{C_1}\cap \noDef} \cup \propagate{\AMG}{C_1})
    \label{eq:def-A2}
\end{align}
%Starting from~\eqref{eq:boundC2-left} and knowing~\eqref{eq:subtree-equal} and~\eqref{eq:def-A1} we have
Knowing $\subTree{\AMG}{\propagate{\AMG}{C}\cap \noDef} = \subTree{\AMG}{\propagate{\AMG}{C_1}\cap \noDef}$ and $\propagate{\AMG}{C_1} \subseteq \propagate{\AMG}{C}$, we have
\begin{gather*}
    \subTree{\AMG}{\propagate{\AMG}{C_1}\cap \noDef} \cup \propagate{\AMG}{C_1} \subseteq \subTree{\AMG}{\propagate{\AMG}{C}\cap \noDef} \cup \propagate{\AMG}{C}\\
    \underbrace{A\setminus (\subTree{\AMG}{\propagate{\AMG}{C}\cap \noDef} \cup \propagate{\AMG}{C})}_{A_1} \cap (\subTree{\AMG}{\propagate{\AMG}{C_1}\cap \noDef} \cup \propagate{\AMG}{C_1}) = \emptyset
\end{gather*}
Now by eq.~\eqref{eq:def-A2} we have $A_1=A_2$.
\end{proof}

We remind and prove Proposition~\ref{prop:subtree-defense}.
\subtreedefenseprop*
\begin{proof}
The first inclusion is immediate by Propositions~\ref{prop:propagate}(\ref{item:propagate-increase}) and~\ref{prop:subtree-inc},
\begin{gather*}
    C \setminus \bigcup_{d\in D} \defended{\AMG}{d} \subseteq C\\
    \propagate{\AMG}{C \setminus \bigcup_{d\in D} \defended{\AMG}{d}} \subseteq \propagate{\AMG}{C}\\
    \subTree{\AMG}{\propagate{\AMG}{C \setminus \bigcup_{d\in D} \defended{\AMG}{d}}\cap \noDef} \subseteq \subTree{\AMG}{\propagate{\AMG}{C}\cap \noDef}
\end{gather*}

Now let $n\in \subTree{\AMG}{\propagate{\AMG}{C}\cap \noDef}$. Let $n_1, \dots, n_k$ be a path from $n_1=\AMGroot$ to $n_k = n$. We can take the smallest integer $i \in \{1, \dots, k-1\}$ such that $n_i$ is a checkpoint in $\propagate{\AMG}{C}\cap \noDef$. We have,
\begin{equation}
    n_i \in \propagate{\AMG}{C} = \propagate{\AMG}{\propagate{\AMG}{C'} \setminus \subTree{\AMG}{\propagate{\AMG}{C'}\cap \noDef}}
    \subseteq \propagate{\AMG}{\propagate{\AMG}{C'}}
    = \propagate{\AMG}{C'}\label{eq:C'}
\end{equation}
As $i$ is chosen as the smallest integer of the set, we have $n_i\not \in \subTree{\AMG}{\propagate{\AMG}{C}\cap \noDef}$, otherwise, there would be another checkpoint earlier in the path $n_1, \dots, n_k$. In the proof of Proposition~\ref{prop:simple-state-proj} we proved that $\subTree{\AMG}{\propagate{\AMG}{C}\cap \noDef} = \subTree{\AMG}{\propagate{\AMG}{C'} \cap \noDef}$, so $n_i \not\in \subTree{\AMG}{\propagate{\AMG}{C'} \cap \noDef}$. Using this fact and eq.~\eqref{eq:C'} we have,
\begin{align*}
    n_i \in \propagate{\AMG}{C'} \setminus \subTree{\AMG}{\propagate{\AMG}{C'} \cap \noDef} = C
\end{align*}
Moreover, $n_i \in \propagate{\AMG}{C}\cap \noDef$ implies $n_i \not\in \noDef$ so $n_i \not\in \bigcup_{d\in D} \defended{\AMG}{d}$. Using the increase of $\propagateDef{\AMG}$,
\begin{gather*}
\left\{
    \begin{matrix}
        n_i \in C \setminus \bigcup_{d\in D} \defended{\AMG}{d}\\
        n_i \not \in \noDef
    \end{matrix}
    \right.\\
    n_i \in \propagate{\AMG}{C \setminus \bigcup_{d\in D} \defended{\AMG}{d}}\cap \noDef
\end{gather*}
As a result, whatever the path $n_1,\dots, n_k$ from the root of the AMG, $n$ has a checkpoint in $\propagate{\AMG}{C \setminus \bigcup_{d\in D} \defended{\AMG}{d}}\cap \noDef$. So $n \in \subTree{\AMG}{\propagate{\AMG}{C \setminus \bigcup_{d\in D} \defended{\AMG}{d}}\cap \noDef}$. Finally by double inclusion,
\begin{equation*}
    \subTree{\AMG}{\propagate{\AMG}{C}\cap \noDef} = \subTree{\AMG}{\propagate{\AMG}{C \setminus \bigcup_{d\in D} \defended{\AMG}{d}}\cap \noDef}
\end{equation*}
\end{proof}

\subsection{Sequential defense activation}

We remind and prove Lemma~\ref{lemma:defense-triangle}.
\defensetriangle*
\begin{proof}
Let $d_1, d_2\in \Defenses$, we assume that $d_1 \ntriangleright d_2$. Let
\begin{align*}
    \loc_1 &= \loc\setminus(\defended{}{d_1}\cup\defended{}{d_2}, \defended{}{d_1}\cup\defended{}{d_2})\\
    \loc_2 &= \loc\setminus(\defended{}{d_1} , \defended{}{d_1})\\
    \loc_3 &= \loc_2 \setminus (\defended{}{d_2} , \defended{}{d_2}) = \loc\setminus(\defended{}{d_1} , \defended{}{d_1}) \setminus (\defended{}{d_2} , \defended{}{d_2})
\end{align*}
We have
\begin{align}
    \Cof{\loc_1} &= \propagate{\AMG}{\Cof{\loc}\setminus(\defended{\AMG}{d_1}\cup \defended{\AMG}{d_2})} \setminus \subTree{\AMG}{\propagate{\AMG}{\Cof{\loc}\setminus(\defended{\AMG}{d_1}\cup \defended{\AMG}{d_2})} \cap \noDef}\label{eq:def-cofl1}\\
    \Cof{\loc_2} &= \propagate{\AMG}{\Cof{\loc}\setminus\defended{\AMG}{d_1}} \setminus \subTree{\AMG}{\propagate{\AMG}{\Cof{\loc}\setminus\defended{\AMG}{d_1}}\cap \noDef}\nonumber\\%\label{eq:def-cofl2}\\
    \Cof{\loc_3} &= \propagate{\AMG}{\Cof{\loc_2}\setminus\defended{\AMG}{d_2}} \setminus \subTree{\AMG}{\propagate{\AMG}{\Cof{\loc_2}\setminus \defended{\AMG}{d_2}}\cap \noDef}\label{eq:def-cofl3}
\end{align}
And the condition $d_1 \ntriangleright d_2$ is equivalent to
\begin{equation}\label{eq:not-d2-follows-d1}
    \forall n_1 \in \defended{\AMG}{d_1}, \forall n_2 \in \children{}{n_1},
    n_2 \in \defended{}{d_1} \lor n_2 \not\in \defended{}{d_2}
\end{equation}
By Proposition~\ref{prop:simple-state-proj}, $\Cof{\simpleState{\loc}} = \Cof{\loc}$, that is 
\begin{equation*}
    \Cof{\loc} = \propagate{\AMG}{\Cof{\loc}} \setminus \subTree{\AMG}{\propagate{\AMG}{\Cof{\loc}} \cap \noDef}
\end{equation*}
we can deduce,
\begin{equation*}%\label{eq:remove-subtree-ok}
    \Cof{\loc} \setminus \subTree{\AMG}{\propagate{\AMG}{\Cof{\loc}} \cap \noDef} = \Cof{\loc}
\end{equation*}
Now we have
\begin{align*}
    \Cof{\loc} \setminus (\defended{\AMG}{d_1} \cup \defended{\AMG}{d_2}) &= \Cof{\loc} \setminus \defended{\AMG}{d_1} \setminus \defended{\AMG}{d_2} \setminus \subTree{\AMG}{\propagate{\AMG}{\Cof{\loc}} \cap \noDef}\\
    &\subseteq \propagate{\AMG}{\Cof{\loc} \setminus \defended{\AMG}{d_1}} \setminus \defended{\AMG}{d_2} \setminus \subTree{\AMG}{\propagate{\AMG}{\Cof{\loc}} \cap \noDef}\\
    &= \propagate{\AMG}{\Cof{\loc} \setminus \defended{\AMG}{d_1}}  \setminus \subTree{\AMG}{\propagate{\AMG}{\Cof{\loc} \setminus \defended{\AMG}{d_1}} \cap \noDef} \setminus \defended{\AMG}{d_2}\\
    &= \Cof{\loc_2} \setminus \defended{\AMG}{d_2}
\end{align*}

We want to show by induction
\begin{equation*}
    \propagate{\AMG}{\Cof{\loc}\setminus\defended{\AMG}{d_1}} \setminus \subTree{\AMG}{\propagate{\AMG}{\Cof{\loc}}} \setminus \defended{\AMG}{d_2} \subseteq \propagate{\AMG}{\Cof{\loc}\setminus\defended{\AMG}{d_1} \setminus \defended{\AMG}{d_2}}
\end{equation*}

For $k \in \mathbb N$, let
\begin{align*}
    N_k &= f_{\propagateDef{}}^k(\Cof{\loc}\setminus\defended{\AMG}{d_1}) \setminus \subTree{\AMG}{\propagate{\AMG}{\Cof{\loc}}}\\
    M_k &= f_{\propagateDef{}}^k(\Cof{\loc}\setminus\defended{\AMG}{d_1} \setminus\defended{\AMG}{d_2})\\
    \mathcal H^1_k & : N_k \setminus \defended{\AMG}{d_2} \subseteq M_k\\
    \mathcal H^2_k & : N_k \cap \defended{\AMG}{d_1} \cap \defended{\AMG}{d_2} \subseteq M_k
\end{align*}
where $\mathcal H^1_k \land \mathcal H^2_k$ is our induction hypothesis.
We have $\mathcal H^1_0 \land \mathcal H^2_0$. For $k \in \mathbb N$ we suppose $\mathcal H^1_k \land \mathcal H^2_k$, we want to show $\mathcal H^1_{k+1} \land \mathcal H^2_{k+1}$. First we prove $\mathcal H^1_{k+1}$.

Let $n \in N_{k+1} \setminus \defended{\AMG}{d_2}$.
\begin{itemize}
    \item If $n \in N_k \setminus \defended{\AMG}{d_2}$, then $n\in M_k$ and by monotony, $n \in M_{k+1}$.
    \item Else, we have $n\not\in N_k \setminus \defended{\AMG}{d_2}$. We notice that $N_k \subseteq \Cof{\loc}$ so $n \in \Cof{\loc} \setminus\defended{\AMG}{d_2}$.
    \begin{itemize}
        \item If $n\not\in \defended{\AMG}{d_1}$, then $n\in \Cof{\loc} \setminus \defended{\AMG}{d_1} \setminus \defended{\AMG}{d_2} \subseteq M_k$ by monotony.
        \item Else, $n\in \defended{\AMG}{d_1}$. As such, $n$ can not be a checkpoint because a checkpoint must be an undefended node. Moreover, if $n\in \defended{\AMG}{d_2}$, then $n\in N_k \cap \defended{\AMG}{d_1} \cap \defended{\AMG}{d_2}$ and by $\mathcal H^2_k$ we have $n\in M_k\subseteq M_{k+1}$. For the rest we suppose $n \not \in \defended{\AMG}{d_2}$, which implies $n\not \in f_{\propagateDef{}}^k(\Cof{\loc}\setminus\defended{\AMG}{d_1})$.
        
        As $n\in f_{\propagateDef{}}^{k+1}(\Cof{\loc}\setminus\defended{\AMG}{d_1}) \setminus f_{\propagateDef{}}^k(\Cof{\loc}\setminus\defended{\AMG}{d_1})$, we can take $n_1, \dots, n_l \in \children{}{n}$ such that for all $i \in \{1, \dots, k\}$, we have $n_i \in f_{\propagateDef{}}^k(\Cof{\loc}\setminus\defended{\AMG}{d_1})$ and $(n_1, \dots, n_k)$ is the list of all children of $n$ if $\operation{n}$ is a conjunction, or is a non empty list of nodes if $\operation{n}$ is a disjunction. As $n$ is not a checkpoint and does not have checkpoint in $\subTree{\AMG}{\propagate{\AMG}{\Cof{\loc}}}$ on every path from $\AMGroot$ we can conclude that for all $i\in \{1, \dots, k\}$, $n_i \not \in \subTree{\AMG}{\propagate{\AMG}{\Cof{\loc}}}$, so $n_i \in N_k$.
        
        By eq.~\eqref{eq:not-d2-follows-d1}, we have that for all $i\in \{1,\dots, k\}$, $n_i\in \defended{\AMG}{d_1}$ or $n_i\not\in \defended{\AMG}{d_2}$.
        \begin{itemize}
            \item If $n_i \not\in \defended{\AMG}{d_2}$, then $n_i\in N_k \setminus \defended{\AMG}{d_2} \subseteq M_k$ by $\mathcal H^1_k$.
            \item Otherwise, $n_i \in \defended{\AMG}{d_1} \cap \defended{\AMG}{d_2}$, so $n_i \in N_k \cap \defended{\AMG}{d_1} \cap \defended{\AMG}{d_2} \subseteq M_k$ by $\mathcal H^2_k$.
        \end{itemize}
        In any cases, for all $i\in \{1,\dots, k\}$, $n_i \in M_k$ so $n \in M_{k+1}$.
    \end{itemize}
\end{itemize}
We treated all the cases and proved $\mathcal H^1_{k+1}$.

Now we prove $\mathcal H^2_{k+1}$. Let $n \in N_{k+1} \cap \defended{\AMG}{d_1} \cap \defended{\AMG}{d_2}$.
\begin{itemize}
    \item If $n\in f_{\propagateDef{}}^k(\Cof{\loc}\setminus\defended{\AMG}{d_1})$ then $n \in N_{k} \cap \defended{\AMG}{d_1} \cap \defended{\AMG}{d_2} \subseteq M_{k+1}$ by $\mathcal H^2_k$.
    \item Otherwise, $n\in f_{\propagateDef{}}^{k+1}(\Cof{\loc}\setminus\defended{\AMG}{d_1}) \setminus f_{\propagateDef{}}^k(\Cof{\loc}\setminus\defended{\AMG}{d_1})$. As before, we can take $n_1, \dots, n_l \in \children{}{n}$ such that for all $i \in \{1, \dots, k\}$, we have $n_i \in f_{\propagateDef{}}^k(\Cof{\loc}\setminus\defended{\AMG}{d_1})$ and $(n_1, \dots, n_k)$ is the list of all children of $n$ if $\operation{n}$ is a conjunction, or is a non empty list of nodes if $\operation{n}$ is a disjunction. As $n$ is not a checkpoint and does not have checkpoint in $\subTree{\AMG}{\propagate{\AMG}{\Cof{\loc}}}$ on every path from $\AMGroot$ we can conclude that for all $i\in \{1, \dots, k\}$, $n_i \not \in \subTree{\AMG}{\propagate{\AMG}{\Cof{\loc}}}$, so $n_i \in N_k$.
    
    By eq.~\eqref{eq:not-d2-follows-d1}, we have that for all $i\in \{1,\dots, k\}$, $n_i\in \defended{\AMG}{d_1}$ or $n_i\not\in \defended{\AMG}{d_2}$.
        \begin{itemize}
            \item If $n_i \not\in \defended{\AMG}{d_2}$, then $n_i\in N_k \setminus \defended{\AMG}{d_2} \subseteq M_k$ by $\mathcal H^1_k$.
            \item Otherwise, $n_i \in \defended{\AMG}{d_1} \cap \defended{\AMG}{d_2}$, so $n_i \in N_k \cap \defended{\AMG}{d_1} \cap \defended{\AMG}{d_2} \subseteq M_k$ by $\mathcal H^2_k$.
        \end{itemize}
        In any cases, for all $i\in \{1,\dots, k\}$, $n_i \in M_k$ so $n \in M_{k+1}$.
\end{itemize}
We treated all the cases and proved $\mathcal H^2_{k+1}$.

To summarize, we proved the following
\begin{equation*}%\label{eq:result-summary}
\left\{
\begin{matrix}
    &\Cof{\loc} \setminus (\defended{\AMG}{d_1} \cup \defended{\AMG}{d_2}) \subseteq \Cof{\loc_2} \setminus \defended{\AMG}{d_2}\\
    &\Cof{\loc_2} \setminus \defended{\AMG}{d_2} \subseteq \Cof{\loc_1}
\end{matrix}
\right.
\end{equation*}
We need to notice that with Lemma~\ref{lemma:subtree-subtree}, we have
\begin{equation*}
    \subTree{\AMG}{\propagate{\AMG}{\Cof{\loc}\setminus(\defended{\AMG}{d_1}\cup \defended{\AMG}{d_2})} \cap \noDef} = \subTree{\AMG}{\Cof{\loc_1} \cap \noDef}
\end{equation*}
So using eq.~\eqref{eq:def-cofl1} and~\eqref{eq:def-cofl3}.
\begin{align*}
    \Cof{\loc_1} &= \propagate{\AMG}{\Cof{\loc}\setminus(\defended{\AMG}{d_1}\cup \defended{\AMG}{d_2})} \setminus \subTree{\AMG}{\Cof{\loc_1} \cap \noDef}\\
    &\subseteq \propagate{\AMG}{\Cof{\loc_2}\setminus\defended{\AMG}{d_2}} \setminus \subTree{\AMG}{\Cof{\loc_2}\setminus \defended{\AMG}{d_2}\cap \noDef}\\
    &\subseteq \propagate{\AMG}{\Cof{\loc_2}\setminus\defended{\AMG}{d_2}} \setminus \subTree{\AMG}{\propagate{\AMG}{\Cof{\loc_2}\setminus \defended{\AMG}{d_2}}\cap \noDef}\\
    &= \Cof{\loc_3}\\
    \Cof{\loc_3} &= \propagate{\AMG}{\Cof{\loc_2}\setminus\defended{\AMG}{d_2}} \setminus \subTree{\AMG}{\propagate{\AMG}{\Cof{\loc_2}\setminus \defended{\AMG}{d_2}}\cap \noDef}\\
    &\subseteq \propagate{\AMG}{\Cof{\loc_1}} \setminus \subTree{\AMG}{\propagate{\AMG}{\Cof{\loc} \setminus (\defended{\AMG}{d_1} \cup \defended{\AMG}{d_2})}\cap \noDef}\\
    &= \Cof{\loc_1}
\end{align*}
So we have $\Cof{\loc_1} = \Cof{\loc_3}$.

We want to show $\Aof{\loc_1} = \Aof{\loc_3}$. By Proposition~\ref{prop:subtree-defense}
\begin{align*}
    \subTree{\AMG}{\propagate{\AMG}{\Cof{\loc}}\cap \noDef} &= \subTree{\AMG}{\propagate{\AMG}{\Cof{\loc} \setminus \defended{\AMG}{d_1}}\cap \noDef}\\
    &= \subTree{\AMG}{\propagate{\AMG}{\Cof{\loc} \setminus \defended{\AMG}{d_1}} \subTree{\AMG}{\propagate{\AMG}{\Cof{\loc} \setminus \defended{\AMG}{d_1}} \cap \noDef}\cap \noDef}\\
    &= \subTree{\AMG}{\Cof{\loc_2}\cap \noDef}\\
    &= \subTree{\AMG}{\propagate{\AMG}{\Cof{\loc_2}} \setminus\subTree{\AMG}{\propagate{\AMG}{\Cof{\loc_2}} \cap \noDef}\cap \noDef}\\
    &=\subTree{\AMG}{\propagate{\AMG}{\Cof{\loc_2}}\cap \noDef}
\end{align*}
Furthermore, using the definition definition
\begin{align*}
    \Aof{\loc_1} &= \Aof{\loc} \setminus \left(\subTree{\AMG}{\propagate{\AMG}{\Cof{\loc} \setminus (\defended{\AMG}{d_1}\cup \defended{\AMG}{d_2})}\cap \noDef} \cup \propagate{\AMG}{\Cof{\loc} \setminus (\defended{\AMG}{d_1}\cup \defended{\AMG}{d_2}})\right)\\
    &= \Aof{\loc} \setminus \left(\subTree{\AMG}{\propagate{\AMG}{\Cof{\loc}}\cap \noDef} \cup \Cof{\loc_1}\right)\\
    \Aof{\loc_2} &= \Aof{\loc} \setminus \left(\subTree{\AMG}{\propagate{\AMG}{\Cof{\loc} \setminus \defended{\AMG}{d_1}}\cap \noDef} \cup \propagate{\AMG}{\Cof{\loc} \setminus \defended{\AMG}{d_1}}\right)\\
    &= \Aof{\loc} \setminus \left(\subTree{\AMG}{\propagate{\AMG}{\Cof{\loc}}\cap \noDef} \cup \Cof{\loc_2}\right)\\
    \Aof{\loc_3} &= \Aof{\loc_2} \setminus \left(\subTree{\AMG}{\propagate{\AMG}{\Cof{\loc_2} \setminus \defended{\AMG}{d_2}}\cap \noDef} \cup \propagate{\AMG}{\Cof{\loc_2} \setminus \defended{\AMG}{d_2}}\right)\\
    &= \Aof{\loc} \setminus \left(\subTree{\AMG}{\propagate{\AMG}{\Cof{\loc}}\cap \noDef} \cup \Cof{\loc_2} \cup \Cof{\loc_3}\right)\\
\end{align*}

We have $\Aof{\loc} = \Aof{\simpleState{\loc}}$ so $\Aof{\loc} = \Aof{\loc} \setminus ( \Cof{\loc} \cup \subTree{\AMG}{\propagate{\AMG}{\loc} \cap \noDef})$. As $\Cof{\loc_1}$, $\Cof{\loc_2}$ and $\Cof{\loc_3}$ are all included in  $\Cof{\loc}$, we have $\Aof{\loc_1} = \Aof{\loc_3}$.
\end{proof}

Now we can prove the Theorem~\ref{thm:defense-chain} that we remind here.
\defensechaintheorem*
\begin{proof}
Suppose $\brac{\Defenses, \{(d_1, d_2) \in \Defenses \times \Defenses \mid d_1 \triangleright d_2\}}$ has no cycle. This means that there is no chain of the form $d_1\triangleright d_2 \triangleright \dots \triangleright d_k \triangleright d_1$. So we can choose i such that $d_j \ntriangleright d_i$ for all the $j\in [1,\dots, k]$ (notice that $d \ntriangleright d$ always holds). Let $d$ be a new defense such that $\defended{\AMG}{d} = \cup_{j\in \{1,\dots, k\}\setminus \{i\}} \defended{\AMG}{d_j}$. We have $d\ntriangleright d_i$, so, by Lemma~\ref{lemma:defense-triangle}, it holds
\begin{equation*}
    \loc\setminus(\cup_{j = 1}^k \defended{\AMG}{d_j}, \cup_{j=1}^k \defended{\AMG}{d_j}) = (\loc \setminus(\cup_{j\in \{1, \dots, k\}\setminus\{i\}} \defended{\AMG}{d_j}, \cup_{j\in \{1, \dots, k\}\setminus\{i\}} \defended{\AMG}{d_j})) \setminus (\defended{\AMG}{d_i}, \defended{\AMG}{d_i})
\end{equation*}
And recursively we can always activate one defense at a time.
\end{proof}
\fi % ALLPROOFS

\end{document}